\documentclass{article}

\usepackage{arxiv}
\usepackage[utf8]{inputenc} % allow utf-8 input
\usepackage[T1]{fontenc}    % use 8-bit T1 fonts
\usepackage{float}
\usepackage[letterpaper]{geometry}
\usepackage{booktabs}       % professional-quality tables
\usepackage{amsfonts}       % blackboard math symbols
\usepackage{nicefrac}       % compact symbols for 1/2, etc.
\usepackage{lipsum}         % Can be removed after putting your text content
\usepackage{graphicx}
\usepackage{natbib}
\usepackage{doi}
\usepackage{multirow}%
\usepackage{amssymb,xparse,mathtools}%
\usepackage{amsthm}%
\usepackage{mathrsfs}%
\usepackage[title]{appendix}%
\usepackage{xcolor}%
\usepackage{textcomp}%
\usepackage{algorithm}%
\usepackage{algpseudocode}%
\usepackage{listings}%
\usepackage{comment}%
\usepackage{setspace} 
\usepackage{enumerate}

\DeclareMathOperator{\di}{\mathsf{d}}
\DeclareMathOperator{\tr}{\mathrm{tr}}
\DeclareMathOperator{\vv}{\mathrm{vec}}
\DeclareMathOperator{\vech}{\mathrm{vech}}
\newtheorem{theorem}{Theorem}%  meant for continuous numbers
\newtheorem{proposition}[theorem]{Proposition}% 

\input Definitions    

\usepackage{hyperref}       % hyperlinks
\hypersetup{
	colorlinks=true,
	linkcolor=blue,
	filecolor=magenta,      
	urlcolor=blue,
	citecolor=blue,
	pdftitle={Asynchronous Distributed ECME Algorithm for Matrix Variate Non-Gaussian Responses},
	pdfpagemode=FullScreen,
}
\usepackage{microtype}      % microtypography

\title{Asynchronous Distributed ECME Algorithm for Matrix Variate Non-Gaussian Responses}

\newif\ifuniqueAffiliation
\uniqueAffiliationtrue

\ifuniqueAffiliation % Standard variant of author block
\author{ \href{https://orcid.org/0000-0003-3265-6330}{\includegraphics[scale=0.06]{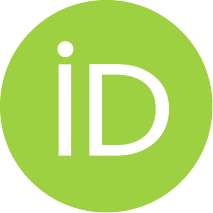}\hspace{1mm}Qingyang Liu} \\
	Department of Statistics\\
	University of Wisconsin-Madison\\
	Madison, WI 53706 \\
	\texttt{qliu432@wisc.edu} \\
	\And
	\href{https://orcid.org/0000-0002-5483-9579}{\includegraphics[scale=0.06]{orcid.pdf}\hspace{1mm}Sanvesh Srivastava} \\
	Department of Statistics and Actuarial Science\\
	University of Iowa\\
	Iowa City, IA 52242 \\
	\texttt{sanvesh-srivastava@uiowa.edu} \\
	\And
	\href{https://orcid.org/0000-0001-5421-1725}{\includegraphics[scale=0.06]{orcid.pdf}\hspace{1mm}Dipankar Bandyopadhyay} \\
	Department of Biostatistics\\
	Virginia Commonwealth University\\
	Richmond, VA 23219 \\
	\texttt{dbandyop@vcu.edu} \\
}

\defcitealias{gallaugher2019three}{Gallaugher and McNicholas, 2019}
\defcitealias{gallaugher2017matrix}{Gallaugher and McNicholas, 2017}
\defcitealias{Gallaugher2024}{Gallaugher and Zhu (2024)}

\begin{document}
\maketitle

\begin{abstract}
We propose a regression model with matrix-variate skew-t response (REGMVST) for analyzing irregular longitudinal data with skewness, symmetry, or heavy tails. REGMVST models matrix-variate responses and predictors, with rows indexing longitudinal measurements per subject. It uses the matrix-variate skew-t (MVST) distribution to handle skewness and heavy tails, a damped exponential correlation (DEC) structure for row-wise dependencies across irregular time profiles, and leaves the column covariance unstructured. For estimation, we initially develop an ECME algorithm for parameter estimation and further mitigate its computational bottleneck via an asynchronous and distributed ECME (ADECME) extension. ADECME accelerates the E-step through parallelization, and retains the simplicity of the conditional M-step, enabling scalable inference. Simulations using synthetic data and a case study exploring matrix-variate periodontal disease endpoints derived from electronic health records demonstrate ADECME’s superiority in efficiency and convergence, over the alternatives. We also provide theoretical support for our empirical observations and identify regularity assumptions for ADECME's optimal performance. An accompanying R package is available at \href{https://github.com/rh8liuqy/STMATREG}{https://github.com/rh8liuqy/STMATREG}.
\end{abstract}
\keywords{Asynchronous Parallel Computations, EM-type Algorithm, Heavy Tail, Matrix-Variate Distribution, Skewness}

\section{Introduction}
Matrix-variate distributions \citep{chen2005matrix} have broad applications in fields that record multiple measurements on a sample. In these applications, the observed data is a matrix with rows and columns representing the samples and measurements. The flexible parameterization of these distributions allows separate column and row dependencies modeling via row and column covariance matrices \citep{nguyen1997note,gupta1997characterization,dutilleul1999mle,viroli2012matrix,Gupta2018}. Despite their flexibility, regression models with matrix-variate outcomes remain less explored. For modeling (continuous) non-Gaussian data encountered in real-world applications, limited options exist -- an well-recognized one being the matrix-variate skew-t (MVST) distribution \citepalias{gallaugher2017matrix}. The MVST distribution effectively models skewness and heavy-tailed errors in regression settings. 

However, in longitudinal studies, where multiple measurements are collected for each subject over time, accounting for temporal dependence becomes crucial. Although the MVST distribution provides a flexible framework for capturing skewness and heavy tails, its implementation through a vanilla EM algorithm \citep{mclachlan2008algorithm} poses significant computational challenges, particularly when applied to `large' longitudinal electronic health records (EHR) data --serving as the central motivation for this study. First, direct maximum likelihood (ML) estimation proves unstable partially due to the modified Bessel function in the log-likelihood \citepalias{gallaugher2017matrix}. Second, while the Expectation Conditional Maximization Either (ECME) algorithm \citep{dempster1977maximum,ECME} addresses this instability, it remains computationally burdensome for large datasets. The estimation problem is further exacerbated in presence of irregular subject profiles observed in EHR, where the number of longitudinal repeats and the time between the repeats vary considerably between subjects.

To overcome these limitations, we develop an asynchronous and distributed ECME (ADECME) extension that enables efficient parameter estimation for massive datasets while maintaining the simplicity and stability of the ``parent'' ECME algorithm \citep{srivastava2019asynchronous}. Furthermore, we incorporate the damped exponential correlation (DEC) structure \citep{munoz1992parametric} into the row covariance matrix of the MVST-distributed response, thereby mitigating the profile irregularity by explicitly modeling the dependence between repeated measurements.

In summary, our central contributions are as follows:

\begin{enumerate}
	\item We propose REGMVST, a flexible matrix-variate regression framework based on the MVST distribution that simultaneously models: (a) skewness and heavy tails in responses, (b) subject-specific observation dimensions, and (c) longitudinal dependencies through a DEC-structured row covariance matrix.
	
	\item We develop ADECME, a novel computational approach that enhances MVST parameter estimation via: (a) a distributed E step enabled by the MVST's stochastic representation, (b) asynchronous updates that minimize the synchronization overhead. This approach achieves significant computational speedups over the usual ECME, while preserving numerical simplicity, stability, and convergence guarantees.
	
	\item We establish ADECME's theoretical properties and its empirical validity through comprehensive convergence analysis and performance evaluations. Our simulations and real-world case study on a periodontal disease dataset derived from EHR \citep{michalowicz2023perio} demonstrate ADECME's superiority over both parallel (PECME) and regular ECME implementations across various data scales.
	
\end{enumerate}

%\subsection{Literature Review}

Extensive research on matrix-variate regression has primarily addressed settings with matrix-valued covariates, whereas models for matrix-valued responses remain relatively underexplored. Examples of such models include regularized exponential family regression \citep{zhou2014regularized}, matrix-variate logistic regression for EEG data \citep{Hung2012}, and its extensions to include measurement error \citep{Fang2020}. Unlike these methods, models for skewed (non-Gaussian) matrix-variate responses, with subject-specific measurements arranged as rows, offer unique advantages for longitudinal data analysis by preserving the natural data structure. The row and column covariance matrices capture the within-subject temporal and between-variable dependencies, respectively. This framework maintains the structural correspondence with matrix covariates, avoids vectorization artifacts, and proves particularly powerful for irregular longitudinal designs, since, flexible row dimensions accommodate varying observation times without compromising interpretable column-wise relationships.

Motivated by these properties, \citetalias{Gallaugher2024} develop hidden Markov models for time series analysis using the MVST distribution. Unlike REGMVST, this approach focuses on time-series data and uses MVST distribution for the emission distribution of hidden states. Similar to REGMVST, \citet{viroli2012matrix} treats both responses and covariates as matrix-valued but relies on the restrictive matrix-variate normal (MVN) distribution. However, this approach is less robust than REGMVST, which simultaneously models skewness and heavy tails through its normal variance-mean mixture construction. In contrast to these works, REGMVST extends the MVN framework by introducing a MVST distribution to handle non-Gaussian features, incorporates a DEC structure for longitudinal dependencies, and proposes an asynchronous distributed ECME algorithm (ADECME) to enable scalable inference for large datasets.

The remainder of this paper is organized as follows. Section \ref{sec:Statistical Model} introduces the MVST distribution and the associated regression models. In Section \ref{sec:Maximum Likelihood Estimation}, we develop the ADECME and other algorithms, all designed for the REGMVST model. We also provide theoretical guarantees of the convergence of the ADECME algorithm. In Section \ref{sec:simulation_study}, we present simulation studies to study the numerical properties of our proposal compared to relevant alternatives, using synthetic data generated under finite sample sizes, large sample sizes, and model mis-specification. Furthermore, we illustrate our methodology via application to the aforementioned real data derived from PD EHR in Section \ref{sec:real_data_application}. Finally, Section \ref{sec:conclusion} concludes alluding to future work. All technical details, such as proofs of Theorems and other theoretical excursions are relegated to the Appendix.

\section{Statistical Model} \label{sec:Statistical Model}
\subsection{The MVST Distribution}\label{sec: The Matrix Variate Distributions}

The MVST distribution is defined as a variance-mean mixture of the MVN distribution. An $n \times p$ random matrix $\mathbf{Y}$ follows the MVN distribution with a $n \times p$ location matrix $\mathbf{M}$, a $n \times n$ row covariance matrix $\boldsymbol{\Sigma}$, and a $p \times p$ column covariance matrix $\boldsymbol{\Psi}$, denoted as $\mathbf{Y} \sim \operatorname{MVN}_{n \times p} \left(\mathbf{M}, \boldsymbol{\Sigma}, \boldsymbol{\Psi}\right)$, if and only if the associated random vector follows a multivariate normal distribution, such that $\operatorname{vec} \left(\mathbf{Y}\right) \sim \mathcal{N}_{np} \left(\operatorname{vec} \left(\mathbf{M}\right), \boldsymbol{\Psi} \otimes \boldsymbol{\Sigma} \right)$; see Theorem 2.7.3 in \cite{Gupta2018}. The MVN distribution is not suitable for modeling data originating from skewed and/or heavy-tailed distributions. To mitigate that, the MVST distribution \citepalias{gallaugher2017matrix} was introduced as the marginal distribution of a linear combination of a location $\mathbf{M}$, a latent variable $W$, and a random matrix $\mathbf{V}$ following an MVN distribution.  Specifically, if the random matrix $\mathbf{Y}$ is defined as
\begin{equation}
	\mathbf{Y} = \mathbf{M} + W \mathbf{A} + \sqrt{W} \mathbf{V}, \quad W \sim \text{Inverse-Gamma} \left(\nu/2,\nu/2\right), \quad \mathbf{V} \sim \operatorname{MVN}_{n \times p} \left(\boldsymbol{0}, \boldsymbol{\Sigma}, \boldsymbol{\Psi}\right),
	\label{eq:MVST_stochastic_definition}
\end{equation}
then, the marginal distribution of $\Yb$ is $\operatorname{MVST}_{n\times p}\left(\mathbf{M}, \mathbf{A}, \boldsymbol{\Sigma}, \boldsymbol{\Psi}, \nu\right)$ distribution,
where, the inverse-gamma distribution in \eqref{eq:MVST_stochastic_definition} has $\nu/2$ as its shape and scale parameters. The density of $\mathbf{Y}$ is
\begin{equation}
	\begin{aligned}
		f_{\text {MVST }}(\mathbf{Y} ; \boldsymbol{\Theta})= & \frac{2\left(\frac{\nu}{2}\right)^{\frac{\nu}{2}} \exp \left\{\operatorname{tr}\left(\boldsymbol{\Sigma}^{-1}(\mathbf{Y}-\mathbf{M}) \boldsymbol{\Psi}^{-1} \mathbf{A}^{\top}\right)\right\}}{(2 \pi)^{\frac{n p}{2}}|\boldsymbol{\Sigma}|^{\frac{p}{2}}|\boldsymbol{\Psi}|^{\frac{n}{2}} \Gamma\left(\frac{\nu}{2}\right)}\left(\frac{\delta(\mathbf{Y} ; \mathbf{M}, \boldsymbol{\Sigma}, \boldsymbol{\Psi})+\nu}{\rho(\mathbf{A}, \boldsymbol{\Sigma}, \boldsymbol{\Psi})}\right)^{-\frac{\nu+n p}{4}} \\
		& \times K_{-\frac{\nu+n p}{2}}\left(\sqrt{[\rho(\mathbf{A}, \boldsymbol{\Sigma}, \boldsymbol{\Psi})][\delta(\mathbf{Y} ; \mathbf{M}, \boldsymbol{\Sigma}, \boldsymbol{\Psi})+\nu]}\right) ,
	\end{aligned}
	\label{eq:MVST_pdf}
\end{equation}
where, $\boldsymbol{\Theta} = \left(\mathbf{M},\mathbf{A},\boldsymbol{\Sigma},\boldsymbol{\Psi},\nu\right)$ is the collection of parameters of interest, $K$ is the modified Bessel function of the second kind, $\delta(\mathbf{Y} ; \mathbf{M}, \boldsymbol{\Sigma}, \boldsymbol{\Psi})=\operatorname{tr}\left(\boldsymbol{\Sigma}^{-1}(\mathbf{Y}-\mathbf{M}) \boldsymbol{\Psi}^{-1}(\mathbf{Y}-\mathbf{M})^{\top}\right)$, and
$
\rho(\mathbf{A}, \boldsymbol{\Sigma}, \boldsymbol{\Psi})=\operatorname{tr}\left(\boldsymbol{\Sigma}^{-1} \mathbf{A} \boldsymbol{\Psi}^{-1} \mathbf{A}^{\top}\right)
$.

Note, an identifiability issue arises in both the MVN and MVST distributions because the covariance matrices are only determined up to a multiplicative constant. This means, the scale of the row and column covariance matrices, $\boldsymbol{\Sigma}$ and $\boldsymbol{\Psi}$, is not unique, as shown by the equivalence $\boldsymbol{\Psi} \otimes \boldsymbol{\Sigma} = (\boldsymbol{\Psi}/c) \otimes (c\boldsymbol{\Sigma})$ for any nonzero constant $c$ \citep{dutilleul1999mle}. A common way to resolve this identifiability issue is to restrict either $\boldsymbol{\Psi}$ or $\boldsymbol{\Sigma}$ to be a correlation matrix. We will discuss our approach to tackling this identifiability issue later in Section~\ref{sec:regression model} within the regression setting.

\begin{figure}[ht]
	\centering
	\includegraphics[width = 0.8\textwidth]{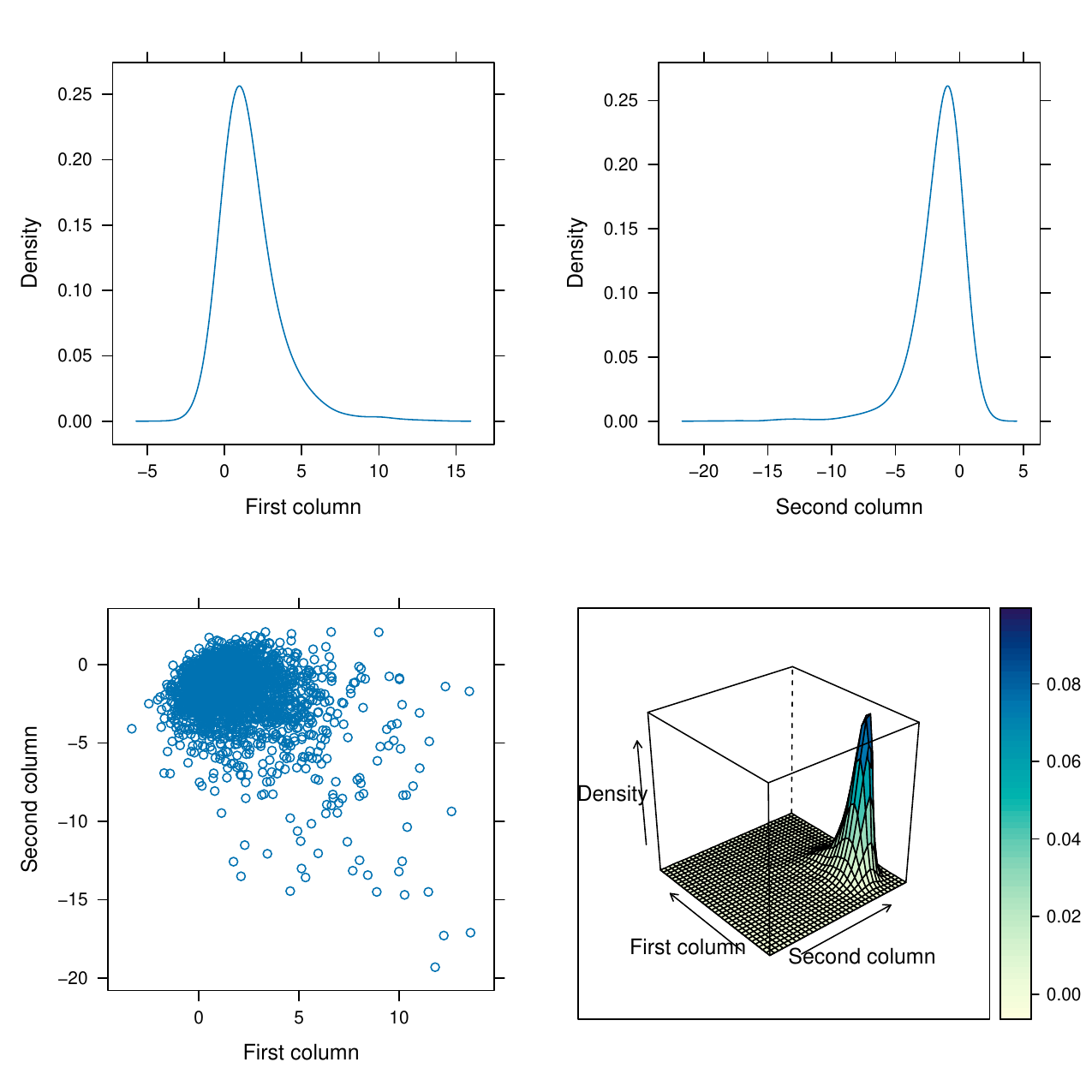}
	\caption{\label{fig:motivation_plot_second2} Plots displaying 1000 realizations drawn from a MVST density. The upper panel displays the (skewed) smoothed densities obtained via Gaussian KDE, corresponding to the first and second columns of the MVST density. In the lower panel, the bottom left presents the (bivariate) scatterplot of the observations from the first and second columns, while the bottom right plots the corresponding bivariate KDE. }
\end{figure}

Consider a simple example that demonstrates the MVST distribution's capacity for modeling skewness and heavy tails. We simulated 1,000 observations from a $3 \times 2$ MVST distribution with the following specifications: (1) location matrix $\mathbf{M} = \mathbf{0}$, (2) degrees of freedom $\nu = 5$ to induce heavy tails, and (3) row and column covariance matrices with unit diagonals and 0.5 off-diagonals. To induce skewness, the skewness matrix $\mathbf{A}$ was specified such that its first column was $\mathbf{1}$ and its second column was $-\mathbf{1}$. Gaussian kernel density estimation (KDE) of the first response dimension showed right-skewed densities (Figure \ref{fig:motivation_plot_second2}, top left), while the second dimension exhibited left-skewed densities (top right). The scatterplot (bottom left) confirmed the specified covariance structure through strong linear associations, and the bivariate KDE (bottom right) simultaneously revealed dimension-specific skewness directions alongside preserved correlation patterns. Together with visible outliers across all panels, these results validate the MVST's ability to jointly model directionally heterogeneous skewness, heavy-tailed distributions (governed by $\nu$), and flexible dependence structures.

\subsection{Regression Model}\label{sec:regression model}

Consider the REGMVST model setup. Let $\mathbf{Y}_i \in \RR^{n_i \times p}$ and $\mathbf{X}_i \in \RR^{n_i \times q}$ be the outcome and covariate matrices for the $i$-th subject, $i = 1, \dots, N$. The row dimensions of the response and covariance matrices \emph{varies} across subjects to accommodate the varying number of repeated measurements across subjects. The REGMVST model posits
\begin{align}
	\label{eq:2}
	\mathbf{Y}_i = \mathbf{X}_{i} \boldsymbol{\beta} + \boldsymbol{e}_{i}, \quad \boldsymbol{e}_i \sim \operatorname{MVST}\left(\boldsymbol{0}, \mathbf{A}_i, \boldsymbol{\Sigma}_i, \boldsymbol{\Psi}, \nu\right), \quad \betab \in \RR^{q \times p},
\end{align}
where $\boldsymbol{\beta}$ is the matrix of regression coefficients, $\mathbf{A}_i = \boldsymbol{1}_{n_i} \mathcal{A}$ represents the vector of skewness, $\boldsymbol{1}_{n_i}$ is a column vector of length $n_i$ consisting of ones, $\mathcal{A}$ is a row vector of length $p$, $\nu$ denotes the degrees of freedom, $\boldsymbol{\Psi}$ is the column covariance matrix with dimension $p \times p$, and $\boldsymbol{\Sigma}_i$ is a $n_i \times n_i$ correlation matrix that models the dependencies in $n_i$ repeated measures across $p$ columns of $\Yb_i$.

We employ the DEC \citep{munoz1992parametric} structure for $\boldsymbol{\Sigma}_i$ to simultaneously address the challenges of parameter identifiability, longitudinal dependence, and model flexibility. This approach resolves the identifiability issue from Section~\ref{sec: The Matrix Variate Distributions} by constraining $\boldsymbol{\Sigma}_i$ to a DEC correlation matrix, thereby fixing the scale. The correlation matrix is formally defined element-wise for the $j$-th row and $k$-th column as 
\begin{equation}
	\Sigma_{ijk} = \rho_{1}^{|t_{ij} - t_{ik}|^{\rho_{2}}}, \quad 0 \le \rho_1, \rho_2 < 1, \quad j, k = 1,\dots,n_{i}
	\label{eq:DEC_definition}
\end{equation}
where $\mathbf{t}_{i} = (t_{i1}, t_{i2}, \dots, t_{in_{i}})$ denotes the observation times for subject $i$. The DEC correlation structure parsimoniously models $\Sigmab_i$ using parameters $\rho_1$ and $\rho_2$. The temporal dependence is naturally captured through the time intervals $|t_{ij} - t_{ik}|$, with $(\rho_1,\rho_2)$ enabling flexible correlation patterns.  Notably, unlike the original DEC specification, we restrict $\rho_2$ to the interval $[0, 1)$ rather than the entire non-negative real line to ensure numerical stability. This restriction prevents the correlation matrix from becoming nearly singular for large time intervals, which can occur with large values of $\rho_2$.

The REGMVST model in \eqref{eq:2} with the DEC correlation structure in \eqref{eq:DEC_definition} implies that the parameters of interest are $\boldsymbol{\vartheta} = \left(\boldsymbol{\beta}, \mathcal{A}, \boldsymbol{\Psi}, \nu, \rho_{1}, \rho_{2}\right)$. Given the observed data $\Dcal_{\text{obs}} = (\Yb_i, \Xb_i, \mathbf{t}_{i}: i = 1, \ldots, N)$, the observed data likelihood function of the REGMVST model follows from \eqref{eq:MVST_pdf}:
\begin{equation}
	\begin{aligned}
		f_{\text {MVST }}(\Dcal_{\text{obs}} ; \boldsymbol{\vartheta})= & \prod_{i = 1}^{N} \left\{\frac{2\left(\frac{\nu}{2}\right)^{\frac{\nu}{2}} \exp \left\{\operatorname{tr}\left(\boldsymbol{\Sigma}_{i}^{-1}(\mathbf{Y}_{i}-\mathbf{M}_{i}) \boldsymbol{\Psi}^{-1} \mathbf{A}_{i}^{\top}\right)\right\}}{(2 \pi)^{\frac{n p}{2}}|\boldsymbol{\Sigma}_{i}|^{\frac{p}{2}}|\boldsymbol{\Psi}|^{\frac{n}{2}} \Gamma\left(\frac{\nu}{2}\right)} \right.\\
		& \left(\frac{\delta(\mathbf{Y}_{i} ; \mathbf{M}_{i}, \boldsymbol{\Sigma}_{i}, \boldsymbol{\Psi})+\nu}{\rho(\mathbf{A}_{i}, \boldsymbol{\Sigma}_{i}, \boldsymbol{\Psi})}\right)^{-\frac{\nu+n_{i} p}{4}} \\
		& \left. K_{-\frac{\nu+n_{i} p}{2}}\left(\sqrt{[\rho(\mathbf{A}_{i}, \boldsymbol{\Sigma}_{i}, \boldsymbol{\Psi})][\delta(\mathbf{Y}_{i} ; \mathbf{M}_{i}, \boldsymbol{\Sigma}_{i}, \boldsymbol{\Psi})+\nu]}\right) \right\},
	\end{aligned}
	\label{eq:regression_likelihood}
\end{equation}
where $\Mb_i = \Xb_i \betab$. The direct numerical maximization of the log likelihood, $\log f_{\text {MVST }}(\mathbf{Y}; \boldsymbol{\vartheta})$, with respect to $\boldsymbol{\vartheta}$ is unstable due to the presence of the modified Bessel function of the second kind. To overcome this issue, \cite{gallaugher2017matrix} proposed  an expectation-conditional maximization (ECM) algorithm \citep{MENG1993}. However, their ECM algorithm is restricted to independent and identically distributed (i.i.d.) observations and is not applicable to the REGMVST model. Specifically, their ECM algorithm cannot be directly used for parameter estimation in the REGMVST model for three reasons. First, the location parameter matrix is defined by $\mathbf{M}_i = \mathbf{X}_i \boldsymbol{\beta}$, which violates the i.i.d. assumption. Second, $\boldsymbol{\Sigma}_i$ is an $n_i \times n_i$ covariance matrix, which also violates the i.i.d. assumption. Finally, the matrices $\boldsymbol{\Sigma}_1, \ldots, \boldsymbol{\Sigma}_N$ depend implicitly on the parameters $(\rho_1, \rho_2)$.

\section{Maximum Likelihood Estimation}\label{sec:Maximum Likelihood Estimation}

To overcome the issue of stable parameter estimation, we leverage the hierarchical representation of the MVST distribution to develop three ECME-type algorithms for parameter estimation. The hierarchical definition of the MVST distribution in \eqref{eq:MVST_pdf} enables6 analytic expressions for conditional means that are useful in deriving the ECME algorithm updates. Specifically, under the regression setting, we can show that \eqref{eq:MVST_pdf} has the following hierarchical representation:
\begin{equation}
	\begin{aligned}
		\mathbf{Y}_i \mid W_i=w_i \sim \operatorname{MVN}_{n \times p} \left(\mathbf{M}_i + w_i \mathbf{A}_i, w_i \boldsymbol{\Sigma}_i, \boldsymbol{\Psi}\right), \quad
		W_i \sim \text{Inverse-Gamma} \left(\nu/2, \nu/2\right),
	\end{aligned}
	\label{eq:MVST_hierarchical}
\end{equation}
where $\Mb_i = \Xb_i \betab$ for the REGMVST model. Additionally, the conditional distribution of $W_i$ given $\mathbf{Y}_i$ is
\begin{equation}
	W_i \mid \mathbf{Y}_i \sim \operatorname{GIG} \left(\rho\left(\mathbf{A}_i, \boldsymbol{\Sigma}_i, \boldsymbol{\Psi}\right), \delta \left(\mathbf{Y}_i;\mathbf{M}_i,\boldsymbol{\Sigma}_i,\boldsymbol{\Psi}\right) +\nu, \lambda_{i} \right),
	\label{eq:W_given_Y}
\end{equation}
where $\lambda_{i} = -\left(\nu + n_i p\right)/2$, $\operatorname{GIG} \left(\rho\left(\mathbf{A}, \boldsymbol{\Sigma}, \boldsymbol{\Psi}\right), \delta \left(\mathbf{Y};\mathbf{M},\boldsymbol{\Sigma},\boldsymbol{\Psi}\right) +\nu, \lambda \right)$ denotes the generalized inverse Gaussian distribution, and the density of $\operatorname{GIG}\left(a,b,\lambda\right)$ distribution is
\begin{equation*}
	f(x ; a, b, \lambda)=\frac{\left(\frac{a}{b}\right)^{\frac{\lambda}{2}} x^{\lambda-1}}{2 K_\lambda(\sqrt{a b})} \exp \left\{-\frac{a x+\frac{b}{x}}{2}\right\}.
\end{equation*}

The remainder of this section is structured as follows. We first introduce the ECME algorithm and explain why it is unsuitable for big data settings. We then describe a parallelized version of the ECME algorithm (PECME) and explain why simple parallelization is insufficient for big data. Finally, we introduce our asynchronous distributed ECME algorithm (ADECME), and explain its key differences from the other two methods. Theoretical properties, including the speed of convergence of the ADECME are also presented.

\subsection{ECME Algorithm}\label{sec:ECME}

Like other EM-variant algorithms, the ECME algorithm begins with three standard steps. These steps involve defining the complete data log-likelihood, calculating the expectation of the complete data log-likelihood with respect to the conditional density of the latent variables given the observed data, and finally deriving the updating formulas for each parameter of interest. In the context of the REGMVST model, the complete data are $\mathcal{D}_{\text{com}} = (\Yb_i, \Xb_i, \mathbf{t}_{i}, W_i: i = 1, \ldots, N)$, and the complete data log-likelihood is
\begin{equation}
	\begin{aligned}
		\ell_C(\boldsymbol{\vartheta}) = &\sum_{i=1}^{N} \left[ \log p\left(\mathbf{Y}_i \mid W_i\right) + \log p\left(W_i\right) \right] \\
		= &\ C + N \left[ \frac{\nu}{2} \log \left(\frac{\nu}{2}\right) - \log \Gamma\left(\frac{\nu}{2}\right) \right] 
		- \frac{1}{2} \sum_{i = 1}^{N} p_{i} \log |\boldsymbol{\Sigma}_{i}| 
		- \frac{1}{2} \left( \sum_{i=1}^{N} n_{i} \right) \log |\boldsymbol{\Psi}| \\
		& - \frac{\nu}{2} \sum_{i=1}^N \log W_i 
		- \frac{1}{2} \sum_{i=1}^N W_i \operatorname{tr}\left(\boldsymbol{\Sigma}_{i}^{-1} \mathbf{A}_{i} \boldsymbol{\Psi}^{-1} \mathbf{A}_{i}^{\top}\right) \\
		& + \frac{1}{2} \sum_{i=1}^N \left[ \operatorname{tr}\left(\boldsymbol{\Sigma}_i^{-1}(\mathbf{Y}_i-\mathbf{M}_i) \boldsymbol{\Psi}^{-1} \mathbf{A}_{i}^{\top}\right)
		+ \operatorname{tr}\left(\boldsymbol{\Sigma}_i^{-1} \mathbf{A}_i \boldsymbol{\Psi}^{-1}(\mathbf{Y}_i-\mathbf{M}_i)^{\top}\right) \right] \\
		& - \frac{1}{2} \sum_{i=1}^N \frac{1}{W_i} \left[ \operatorname{tr}\left(\boldsymbol{\Sigma}_{i}^{-1}(\mathbf{Y}_i-\mathbf{M}_i) \boldsymbol{\Psi}^{-1}(\mathbf{Y}_i-\mathbf{M}_i)^{\top}\right) + \nu \right].
	\end{aligned}
	\label{eq:complete_data_loglikelihood}
\end{equation}
where $C$ does not depend on $\boldsymbol{\vartheta}$.

The E step of the ECME algorithm computes the expectation of the complete data log-likelihood in \eqref{eq:complete_data_loglikelihood} with respect to the conditional density of $W_{i}$ given $\Yb_i$ in \eqref{eq:W_given_Y}.  For iteration $t+1$, we require $\mathbb{E}\left(W_i \mid \mathbf{Y}_i, \boldsymbol{\vartheta}^{(t)}\right), \mathbb{E} \left(\ln W_i \mid \mathbf{Y}_i, \boldsymbol{\vartheta}^{(t)}\right)$, and $ \mathbb{E}\left(1/W_{i} \mid \mathbf{Y}_i, \boldsymbol{\vartheta}^{(t)}\right)$, where $\boldsymbol{\vartheta}^{(t)}$ is the vector of estimated parameters from iteration $t$. Specifically, the calculation of conditional expectation of the complete data log-likelihood in the E step is defined as:
\begin{equation}
	\begin{split}
		Q\big(\boldsymbol{\vartheta} \mid \boldsymbol{\vartheta}^{(t)}\big) 
		&= \mathbb{E}_{\boldsymbol{W} \mid \mathbf{Y}, \boldsymbol{\vartheta}^{(t)}} \left( \ell_C(\boldsymbol{\vartheta}) \right) \\
		&= C - N \log \Gamma\left(\frac{\nu}{2}\right) + \frac{N \nu}{2} \log\left(\frac{\nu}{2}\right) - \frac{\nu}{2} \sum_{i=1}^N c_{i}^{(t+1)} \\
		&\quad - \frac{1}{2} \sum_{i = 1}^{N} p_{i} \log |\boldsymbol{\Sigma}_{i}| - \frac{1}{2}\left(\sum_{i=1}^{N} n_{i}\right) \log |\boldsymbol{\Psi}| \\
		&\quad + \frac{1}{2} \sum_{i=1}^N \operatorname{tr}\left( \boldsymbol{\Sigma}_i^{-1} (\mathbf{Y}_i-\mathbf{M}_i) \boldsymbol{\Psi}^{-1} \mathbf{A}_{i}^{\top} \right) 
		+ \frac{1}{2} \sum_{i=1}^N \operatorname{tr}\left( \boldsymbol{\Sigma}_i^{-1} \mathbf{A}_i \boldsymbol{\Psi}^{-1} (\mathbf{Y}_i-\mathbf{M}_i)^{\top} \right) \\
		&\quad - \frac{1}{2} \sum_{i=1}^N a_{i}^{(t+1)} \operatorname{tr}\left( \boldsymbol{\Sigma}_{i}^{-1} \mathbf{A}_{i} \boldsymbol{\Psi}^{-1} \mathbf{A}_{i}^{\top} \right) \\
		&\quad - \frac{1}{2} \sum_{i=1}^N b_{i}^{(t+1)} \left[ \operatorname{tr}\left( \boldsymbol{\Sigma}_{i}^{-1} (\mathbf{Y}_i-\mathbf{M}_i) \boldsymbol{\Psi}^{-1} (\mathbf{Y}_i-\mathbf{M}_i)^{\top} \right) + \nu \right].
	\end{split}
	\label{eq:regression_Q_function}
\end{equation}
where 
$$\boldsymbol{W} = \left[W_{1},\dots,W_{N}\right]^{\top},$$ 
$$\mathbf{Y} = \left(\mathbf{Y}_{1},\dots,\mathbf{Y}_{N}\right),$$

\begin{equation*}
	\begin{aligned}
		a_i^{(t+1)} & =\mathbb{E}\left(W_i \mid \mathbf{Y}_i, \hat{\boldsymbol{\vartheta}}^{(t)}\right) \\
		& =\sqrt{\frac{\delta\left(\mathbf{Y}_{i} ; \hat{\mathbf{M}}_{i}^{(t)}, \hat{\boldsymbol{\Sigma}}_{i}^{(t)}, \hat{\boldsymbol{\Psi}}^{(t)}\right)+\hat{\nu}^{(t)}}{\rho\left(\hat{\mathbf{A}}_{i}^{(t)}, \hat{\boldsymbol{\Sigma}}_{i}^{(t)}, \hat{\boldsymbol{\Psi}}^{(t)}\right)}} \frac{K_{\lambda^{(t)}_{i}+1}\left(\kappa_{i}^{(t)}\right)}{K_{\lambda_{i}^{(t)}}\left(\kappa_{i}^{(t)}\right)},
	\end{aligned}
	\label{estepa}
\end{equation*}

\begin{equation*}
	\begin{aligned}
		b_i^{(t+1)} = & \mathbb{E}\left(\frac{1}{W_i} \mid \mathbf{Y}_i, \hat{\boldsymbol{\vartheta}}^{(t)}\right) \\
		= & \sqrt{\frac{\rho\left(\hat{\mathbf{A}}^{(t)}_{i}, \hat{\boldsymbol{\Sigma}}^{(t)}_{i}, \hat{\boldsymbol{\Psi}}^{(t)}\right)}{\delta\left(\mathbf{Y}_i ; \hat{\mathbf{M}}^{(t)}_{i}, \hat{\boldsymbol{\Sigma}}^{(t)}_{i}, \hat{\boldsymbol{\Psi}}^{(t)}\right)+\hat{\nu}^{(t)}}}\frac{K_{\lambda_{i}^{(t)}+1}\left(\kappa_{i}^{(t)}\right)}{K_{\lambda_{i}^{(t)}}\left(\kappa_{i}^{(t)}\right)} \\
		& + \frac{\hat{\nu}^{(t)}+n_{i} p}{\delta\left(\mathbf{Y}_{i} ; \hat{\mathbf{M}}^{(t)}_{i}, \hat{\boldsymbol{\Sigma}}^{(t)}_{i}, \hat{\boldsymbol{\Psi}}\right)+\hat{\nu}^{(t)}},
	\end{aligned}
\end{equation*}

\begin{equation*}
	\label{estepc}
	\begin{aligned}
		c_i^{(t+1)} = & \mathbb{E}\left(\log \left(W_i\right) \mid \mathbf{Y}_i, \hat{\boldsymbol{\vartheta}}^{(t)}\right) \\
		= & \log \left(\sqrt{\frac{\delta\left(\mathbf{Y}_i ; \hat{\mathbf{M}}^{(t)}_{i}, \hat{\boldsymbol{\Sigma}}^{(t)}_{i}, \hat{\boldsymbol{\Psi}}^{(t)}\right)+\hat{\nu}^{(t)}}{\rho\left(\hat{\mathbf{A}}_{i}^{(t)}, \hat{\boldsymbol{\Sigma}}_{i}^{(t)}, \hat{\boldsymbol{\Psi}}^{(t)}\right)}}\right) \\
		& +\left.\frac{1}{K_{\lambda_{i}^{(t)}}\left(\kappa_{i}^{(t)}\right)} \frac{\partial}{\partial \lambda} K_\lambda\left(\kappa_{i}^{(t)}\right)\right|_{\lambda=\lambda_{i}^{(t)}},
	\end{aligned}
\end{equation*}
where
\begin{equation*}
	\kappa_{i}^{(t)}=\sqrt{\left[\rho\left(\hat{\mathbf{A}}_{i}^{(t)}, \hat{\boldsymbol{\Sigma}_{i}}^{(t)}, \hat{\boldsymbol{\Psi}}^{(t)}\right)\right]\left[\delta\left(\mathbf{Y}_i ; \hat{\mathbf{M}}^{(t)}_{i}, \hat{\boldsymbol{\Sigma}}^{(t)}_{i}, \hat{\boldsymbol{\Psi}}^{(t)}\right)+\hat{\nu}^{(t)}\right]},
\end{equation*}
and
\begin{equation*}
	\lambda^{(t)}_{i}=-\frac{v^{(t)}+n_{i} p}{2}.
\end{equation*}

After the E step, the series of conditional M (CM) estimate $\betab, \nu, \Psib, \Acal, \phi$ as follows: 

\begin{enumerate}[(1)]
	\item We update $\boldsymbol{\beta}$ as
	\begin{equation*}
		\hat{\boldsymbol{\beta}}^{(t+1)} = \left(\sum_{i=1}^{N} b_{i}^{(t+1)} \mathbf{X}_{i}^{\top} 
		\hat{\boldsymbol{\Sigma}}_{i}^{(t)^{-1}} \mathbf{X}_{i}\right)^{-1} \left(\sum_{i=1}^{N} - \mathbf{X}_{i}^{\top} \hat{\boldsymbol{\Sigma}}_{i}^{(t)^{-1}} \hat{\mathbf{A}}_{i}^{(t)} + b_{i}^{(t+1)} \mathbf{X}_{i}^{\top} \hat{\boldsymbol{\Sigma}}_{i}^{(t)^{-1}} \mathbf{Y}_{i}\right).
	\end{equation*}
	
	\item We update $\nu$ as the solution to 
	\begin{equation*}
		\log \left(\frac{\nu}{2}\right)+1-\varphi\left(\frac{\nu}{2}\right)-\frac{1}{N} \sum_{i=1}^N\left(b_i^{(t+1)}+c_i^{(t+1)}\right)=0,
	\end{equation*}
	where $\varphi(\cdot)$ is the digamma function.
	
	\item  An update of the skewness parameter can be performed as
	\begin{equation*}
		\hat{\mathcal{A}}^{(t+1)}=\frac{\sum_{i=1}^N \boldsymbol{1}_{n_{i}}^{\top} \hat{\boldsymbol{\Sigma}}_i^{(t)^{-1}}\left(\mathbf{Y}_i-\mathbf{X}_i \hat{\boldsymbol{\beta}}^{(t+1)}\right)}{\sum_{i=1}^N a_{i}^{(t+1)} \boldsymbol{1}_{n_{i}}^{\top} \hat{\boldsymbol{\Sigma}}_i^{(t)^{-1}} \boldsymbol{1}_{n_{i}}}.
	\end{equation*}
	
	\item We update $\boldsymbol{\Psi}$ as
	\begin{equation*}
		\begin{aligned}
			\hat{\boldsymbol{\Psi}}^{(t+1)}= & \left[\sum _ { i = 1 } ^ { N } \left(b_i^{(t+1)}\left(\mathbf{Y}_i-\hat{\mathbf{M}}_{i}^{(t+1)}\right)^{\top} \hat{\boldsymbol{\Sigma}}_{i}^{(t)^{-1}}\left(\mathbf{Y}_i-\hat{\mathbf{M}}_{i}^{(t+1)}\right)\right.\right. \\
			& -\hat{\mathbf{A}}_{i}^{(t+1)^{\top}} \hat{\boldsymbol{\Sigma}}_{i}^{(t)^{-1}}\left(\mathbf{Y}_i-\hat{\mathbf{M}}_{i}^{(t+1)}\right) \\
			& -\left(\mathbf{Y}_i-\hat{\mathbf{M}}_{i}^{(t+1)}\right)^{\top} \hat{\boldsymbol{\Sigma}}_{i}^{(t+1)^{-1}} \hat{\mathbf{A}}_{i}^{(t+1)} \\
			& \left.\left.+a_i^{(t+1)} \hat{\mathbf{A}}_{i}^{(t+1)^{\top}} \hat{\boldsymbol{\Sigma}}_{i}^{(t)^{-1}} \hat{\mathbf{A}}_{i}^{(t+1)}\right)\right]/\left[\sum_{i=1}^{N} n_{i}\right].
		\end{aligned}
	\end{equation*}
	
	\item We update two parameters $\rho_{1}$ and $\rho_{2}$ from the DEC structure, sequentially, using the grid search algorithm.
	
	First, for $\rho_{1}$, we construct a vector $\rho_{1} \in \left(10^{-5}, 0.1, 0.2, \dots, 0.9, 1-10^{-5}\right)$ and evaluate the log-transformed observed likelihood in \eqref{eq:regression_likelihood} for each value, using $\hat{\boldsymbol{\beta}}^{(t+1)}$, $\hat{\nu}^{(t+1)}$, $\hat{\boldsymbol{\Psi}}^{(t+1)}$, $\hat{\mathcal{A}}^{(t+1)}$, and $\hat{\rho}_{2}^{(t)}$. The value maximizing the likelihood yields the updated estimate $\hat{\rho}_{1}^{(t+1)}$. The same procedure applies to $\rho_{2}$, where we evaluate the likelihood with $\hat{\rho}_{1}^{(t+1)}$ instead. While the Newton–Raphson or Nelder–Mead method could directly maximize $\rho_{1}$ and $\rho_{2}$ using \eqref{eq:regression_likelihood} as the objective function, the computational cost grows prohibitively high. Parallelization might mitigate this, but communication overhead often renders such approaches inefficient.
	
\end{enumerate}

However, the ECME algorithm is not well-suited for big data applications due to two primary computational bottlenecks. First, the algorithm has a slow E step, which requires calculating the conditional expectation of the complete data log-likelihood -- an operation that must be performed for every single observation in the dataset. This process becomes computationally prohibitive as the sample size grows very large. Second, ECME features a slow updating mechanism for the DEC parameters. Specifically, updating each of the parameters $\rho_1$ and $\rho_2$ requires a full evaluation of the observed data log-likelihood for the entire dataset. Since this evaluation must be performed separately for each parameter, the update cycle demands two complete passes through all observations, further escalating the computational burden for large-scale data.

\subsection{PECME Algorithm} \label{sec:PECME}

In this section, we introduce the PECME algorithm, which represents the parallelized version of the ECME algorithm. While the ECME algorithm operates using a single CPU core, the PECME algorithm leverages parallel processing to enhance efficiency. Effective implementation of the PECME algorithm requires access to multiple CPU cores on a single computer, or the use of multiple nodes within a high-performance computing cluster. The PECME algorithm employs two distinct groups of computing processes, referred to as \emph{workers} and a \emph{manager}. Specifically, PECME reserves $(k + 1)$ processes for computation, consisting of $k$ workers and one manager. Before the PECME algorithm begins, the complete dataset is divided into smaller $k$ disjoint subsets and allocated to the $k$ worker processes. Let $N_j$ denote the number of samples in the $j$-th subset, $\left(\mathbf{Y}_{ji},\mathbf{X}_{ji},\mathbf{t}_{ji}\right)$ represent the $i$-th sample within the $j$-th subset $\left(j = 1,\dots,k; i = 1,\dots,N_{j}\right)$, and $n_{ji}$ denote the number of rows of $\mathbf{Y}_{ji}$. Consequently, the sum of all samples across subsets equals the total sample size, expressed as, $N_{1} + \cdots + N_{k} = N$. The union of all subset samples corresponds to the original complete dataset, $\cup_{j=1}^{k}\cup_{i=1}^{N_{j}}\left(\mathbf{Y}_{ji},\mathbf{X}_{ji},\mathbf{t}_{ji}\right) = \left\{\left(\mathbf{Y}_{1},\mathbf{X}_{1},\mathbf{t}_{1}\right),\dots,\left(\mathbf{Y}_{N},\mathbf{X}_{N},\mathbf{t}_{N}\right)\right\}$. Within the PECME algorithm, each worker computes sufficient statistics from its assigned data subset and then transmits these results to the manager for further processing.

\subsubsection{E Step - PECME} \label{sec:PECME_E_Step}

The manager starts with some initial values $\boldsymbol{\vartheta}^{(0)}$ at $t = 0$ and sends $\boldsymbol{\vartheta}^{(0)}$ to \emph{all} workers. For each of $ t= 0,1,\dots,\infty$, the manager waits to receive \emph{all} sufficient statistics from \emph{all} workers before proceeding to the CM step.

\begin{equation*}
	a_{ji}^{(t+1)} = \sqrt{\frac{\delta\left(\mathbf{Y}_{ji} ; \hat{\mathbf{M}}_{ji}^{(t)}, \hat{\boldsymbol{\Sigma}}_{ji}^{(t)}, \hat{\boldsymbol{\Psi}}^{(t)}\right)+\hat{\nu}^{(t)}}{\rho\left(\hat{\mathbf{A}}_{ji}^{(t)}, \hat{\boldsymbol{\Sigma}}_{ji}^{(t)}, \hat{\boldsymbol{\Psi}}^{(t)}\right)}} \frac{K_{\lambda^{(t)}_{ji}+1}\left(\kappa_{ji}^{(t)}\right)}{K_{\lambda_{ji}^{(t)}}\left(\kappa_{ji}^{(t)}\right)},
\end{equation*}
\begin{equation*}
	\begin{aligned}
		b_{ji}^{(t+1)} = & \sqrt{\frac{\rho\left(\hat{\mathbf{A}}^{(t)}_{ji}, \hat{\boldsymbol{\Sigma}}^{(t)}_{ji}, \hat{\boldsymbol{\Psi}}^{(t)}\right)}{\delta\left(\mathbf{Y}_i ; \hat{\mathbf{M}}^{(t)}_{ji}, \hat{\boldsymbol{\Sigma}}^{(t)}_{ji}, \hat{\boldsymbol{\Psi}}^{(t)}\right)+\hat{\nu}^{(t)}}}\frac{K_{\lambda_{ji}^{(t)}+1}\left(\kappa^{(t)}\right)}{K_{\lambda_{ji}^{(t)}}\left(\kappa^{(t)}\right)} \\
		& +\frac{\hat{\nu}^{(t)}+n_{ji} p}{\delta\left(\mathbf{Y}_{ji} ; \hat{\mathbf{M}}^{(t)}_{ji}, \hat{\boldsymbol{\Sigma}}^{(t)}_{ji}, \hat{\boldsymbol{\Psi}}\right)+\hat{\nu}^{(t)}},
	\end{aligned}
\end{equation*}
\begin{equation*}
	\begin{aligned}
		c_{ji}^{(t+1)} = & \log \left(\sqrt{\frac{\delta\left(\mathbf{Y}_{ji} ; \hat{\mathbf{M}}^{(t)}_{ji}, \hat{\boldsymbol{\Sigma}}^{(t)}_{ji}, \hat{\boldsymbol{\Psi}}^{(t)}\right)+\hat{\nu}^{(t)}}{\rho\left(\hat{\mathbf{A}}_{ji}^{(t)}, \hat{\boldsymbol{\Sigma}}_{ji}^{(t)}, \hat{\boldsymbol{\Psi}}^{(t)}\right)}}\right) \\
		& +\left.\frac{1}{K_{\lambda_{ji}^{(t)}}\left(\kappa_{ji}^{(t)}\right)} \frac{\partial}{\partial \lambda} K_\lambda\left(\kappa_{ji}^{(t)}\right)\right|_{\lambda=\lambda_{ji}^{(t)}}.
		\label{eq:ADECM_c}
	\end{aligned}
\end{equation*}
\begin{equation*}
	\mathbf{S}_{\boldsymbol{\beta}1,ji}^{(t+1)} = b_{ji}^{(t+1)} \mathbf{X}_{ji}^{\top} \hat{\boldsymbol{\Sigma}}_{ji}^{(t)^{-1}} \mathbf{X}_{ji},
\end{equation*}
\begin{equation*}
	\mathbf{S}_{\boldsymbol{\beta}2,ji}^{(t+1)} = -\mathbf{X}_{ji}^{\top} \hat{\boldsymbol{\Sigma}}_{ji}^{(t)^{-1}} \hat{\mathbf{A}}_{ji}^{(t)}+b_{ji}^{(t+1)} \mathbf{X}_{ji}^{\top} \hat{\boldsymbol{\Sigma}}_{ji}^{(t)^{-1}} \mathbf{Y}_{ji},
\end{equation*}
\begin{equation*}
	\mathbf{S}_{\boldsymbol{\beta}1,j}^{(t+1)} = \sum_{i=1}^{N_{j}} \mathbf{S}_{\beta 1, ji}^{(t+1)},
\end{equation*}
\begin{equation*}
	\mathbf{S}_{\boldsymbol{\beta}2,j}^{(t+1)} = \sum_{i=1}^{N_{j}} \mathbf{S}_{\boldsymbol{\beta}2,ji}^{(t+1)}.
\end{equation*}
\begin{equation*}
	\mathbf{S}_{\nu, ji}^{(t+1)} = b_{ji}^{(t+1)} + c_{ji}^{(t+1)},
\end{equation*}
\begin{equation*}
	\mathbf{S}_{\nu, j}^{(t+1)} = \sum_{i=1}^{N_{j}}  \mathbf{S}_{\nu, ji}^{(t+1)}.
\end{equation*}

\subsubsection{CM Step - PECME} \label{sec:PECME_M_Step}

After the manager receives \emph{all} sufficient statistics described in Section~\ref{sec:PECME_E_Step} from \emph{all} workers, it updates $\boldsymbol{\vartheta}^{(t+1)}$ in the following order:

\begin{enumerate}[(1)]
	
	\item Update $\boldsymbol{\beta}$.
	
	The manager updates the estimation of $\boldsymbol{\beta}$ as
	\begin{equation}
		\hat{\boldsymbol{\beta}}^{(t+1)} = \left(\sum_{j=1}^{b} \mathbf{S}_{\boldsymbol{\beta} 1, j}^{(t+1)}\right)^{-1} \left(\sum_{j=1}^{b} \mathbf{S}_{\boldsymbol{\beta} 2, j}^{(t+1)}\right).
		\label{eq:PECME_beta}
	\end{equation}
	
	\item Update $\nu$.
	
	The manager updates the estimation of $\nu$ as the solution to 
	\begin{equation*}
		\log \left(\frac{\nu}{2}\right)+1-\varphi\left(\frac{\nu}{2}\right)-\frac{1}{N} \sum_{j=1}^{b}\left(\mathbf{S}_{\nu, j}^{(t+1)}\right)=0.
	\end{equation*}
	
	\item Update $\mathcal{A}$.
	
	The manager sends the most recently updated estimated value of $\boldsymbol{\beta}$, $\hat{\boldsymbol{\beta}}^{(t+1)}$, to \emph{all} workers to calculate the sufficient statistics of $\mathcal{A}$.
	
	\begin{equation*}
		\mathbf{S}_{\mathcal{A} 1, ji}^{(t+1)} = \mathbf{1}_{n_{ji}}^{\top} \hat{\boldsymbol{\Sigma}}_{ji}^{(t)^{-1}}\left(\mathbf{Y}_{ji}-\mathbf{X}_{ji} \hat{\boldsymbol{\beta}}^{(t+1)}\right),
	\end{equation*}
	\begin{equation*}
		\mathbf{S}_{\mathcal{A} 2, ji}^{(t+1)} = a_{ji}^{(t+1)} \mathbf{1}_{n_{ji}}^{\top} \hat{\boldsymbol{\Sigma}}_{ji}^{(t)^{-1}} \mathbf{1}_{n_{ji}}.
	\end{equation*}
	
	Once the calculation of $\mathbf{S}_{\mathcal{A} 1, ji}^{(t+1)}$ and $\mathbf{S}_{\mathcal{A} 2, ji}^{(t+1)}$ is completed, \emph{all} workers transfer these statistics back to the manager. The manager aggregates these statistics as follows:
	\begin{equation*}
		\mathbf{S}_{\mathcal{A} 1, j}^{(t+1)} = \sum_{i=1}^{N_{j}}  \mathbf{S}_{\mathcal{A} 1, ji}^{(t+1)},
	\end{equation*}
	\begin{equation*}
		\mathbf{S}_{\mathcal{A} 2, j}^{(t+1)} = \sum_{i=1}^{N_{j}}  \mathbf{S}_{\mathcal{A} 2, ji}^{(t+1)}.
	\end{equation*}
	
	After the aggregation, the manager updates the estimation of $\mathcal{A}$ as:
	\begin{equation*}
		\hat{\mathcal{A}}^{(t+1)} = \frac{\sum_{j=1}^{b} \mathbf{S}_{\mathcal{A} 1, j}^{(t+1)}}{\sum_{j=1}^{b} \mathbf{S}_{\mathcal{A} 2, j}^{(t+1)}}.
	\end{equation*}
	
	\item Update $\boldsymbol{\Psi}$.
	
	The manager sends $\hat{\mathcal{A}}^{(t+1)}$ to \emph{all} workers who calculate the sufficient statistics of $\boldsymbol{\Psi}$.
	
	\begin{equation*}
		\begin{aligned}
			\mathbf{S}_{\boldsymbol{\Psi},ji}^{(t+1)} = & 
			\left[b_{ji}^{(t+1)}\left(\mathbf{Y}_{ji}-\hat{\mathbf{M}}_{ji}^{(t+1)}\right)^{\top} \hat{\boldsymbol{\Sigma}}_{ji}^{(t)^{-1}}\left(\mathbf{Y}_{ji}-\hat{\mathbf{M}}_{ji}^{(t+1)}\right)\right. \\
			& -\hat{\mathbf{A}}_{ji}^{(t+1)^{\top}} \hat{\boldsymbol{\Sigma}}_{ji}^{(t)^{-1}}\left(\mathbf{Y}_{ji}-\hat{\mathbf{M}}_{ji}^{(t+1)}\right) \\
			& -\left(\mathbf{Y}_{ji}-\hat{\mathbf{M}}_{ji}^{(t+1)}\right)^{\top} \hat{\boldsymbol{\Sigma}}_{ji}^{(t)^{-1}} \hat{\mathbf{A}}_{ji}^{(t+1)} \\
			& \left.+a_{ji}^{(t+1)} \hat{\mathbf{A}}_{ji}^{(t+1)^{\top}} \hat{\boldsymbol{\Sigma}}_{ji}^{(t)^{-1}} \hat{\mathbf{A}}_{ji}^{(t+1)}\right].
		\end{aligned}
	\end{equation*}
	
	After the calculation is completed, \emph{all} workers transfer $\mathbf{S}_{\boldsymbol{\Psi},ji}^{(t+1)}$ back to the manager. Then, the manager aggregates these statistics as:
	\begin{equation*}
		\mathbf{S}_{\boldsymbol{\Psi},j}^{(t+1)} = \sum_{i=1}^{N_{j}}  \mathbf{S}_{\boldsymbol{\Psi},ji}^{(t+1)}.
	\end{equation*}
	
	After the aggregation, the manager updates the estimation of $\boldsymbol{\Psi}$ as:
	\begin{equation*}
		\hat{\boldsymbol{\Psi}}^{(t+1)} = \frac{\sum_{j=1}^{b}\mathbf{S}_{\boldsymbol{\Psi},j}^{(t+1)}}{\sum_{j=1}^{b}\sum_{i=1}^{N_{j}}  n_{ji}}.
	\end{equation*}
	
	\item Update $\rho_{1}$ and $\rho_{2}$ from the DEC structure using grid search. 
	
	The manager updates $\rho_{1}$ and $\rho_{2}$ sequentially. For $\rho_{1}$, the manager distributes a vector $\rho_{1} \in \left(10^{-5}, 0.1, \dots, 1-10^{-5}\right)$ to all workers, along with $\hat{\boldsymbol{\beta}}^{(t+1)}$, $\hat{\nu}^{(t+1)}$, $\hat{\mathcal{A}}^{(t+1)}$, $\hat{\boldsymbol{\Psi}}^{(t+1)}$, and $\rho_{2}^{(t)}$, requesting evaluation of the observed log-likelihood in \eqref{eq:regression_likelihood}. Workers compute their assigned subsets and return the results; the manager then aggregates these and selects the $\rho_{1}$ value maximizing the log-likelihood as $\hat{\rho}_{1}^{(t+1)}$. The same procedure follows for $\rho_{2}$, using $\rho_{1}^{(t+1)}$ and the corresponding vector $\rho_{2} \in \left(10^{-5}, 0.1, \dots, 1-10^{-5}\right)$ to determine $\hat{\rho}_{2}^{(t+1)}$.
	
\end{enumerate}

It is important to note that each PECME iteration requires five manager-worker communications: during the distributed E step (Section~\ref{sec:PECME_E_Step}), and when updating $\mathcal{A}$, $\boldsymbol{\Psi}$, $\rho_{1}$, and $\rho_{2}$ from the DEC structure. As demonstrated by our simulation studies (Section~\ref{sec:simulation_study}) and real data application (Section~\ref{sec:real_data_application}), this communication overhead incurs significant computational costs, substantially slowing the PECME algorithm.

\subsection{ADECME Algorithm} \label{sec:ADECME}

The ADECME and PECME algorithms differ in both the distributed E step and the CM step. In ADECME, the manager waits for only a fraction $\gamma \in (0,1)$ of workers to finish the distributed E step, thereby improving efficiency (e.g., with 8 workers and $\gamma=0.5$, the manager waits for 4 workers; with $\gamma=0.8$, for 7). To further reduce communication, ADECME computes the sufficient statistics of $\mathcal{A}$ and $\boldsymbol{\Psi}$ during the distributed E step using parameter estimates from the previous iteration rather than the current one, eliminating the need for manager–worker exchanges in the CM step. ADECME also moves the grid search for $\rho_{1}$ and $\rho_{2}$ into the E step, again using previous-iteration estimates ($\hat{\boldsymbol{\beta}}^{(t)}, \hat{\nu}^{(t)}, \hat{\mathcal{A}}^{(t)}, \hat{\boldsymbol{\Psi}}^{(t)}, \hat{\rho}_{2}^{(t)}$ for $\rho_{1}$ and $\hat{\rho}_{1}^{(t)}$ for $\rho_{2}$), whereas PECME performs this search in the CM step with current estimates from iteration $t+1$. These design choices collectively make ADECME more communication-efficient than PECME. In what follows, we detail the modifications to each computational step, beginning with the distributed E step.

\subsubsection{ADECME: The Distributed E Step} \label{sec:ADECME_E_Step}

In addition to computing $a_{ji}^{(t+1)}, b_{ji}^{(t+1)}, c_{ji}^{(t+1)}$, and the sufficient statistics for $\boldsymbol{\beta}$ and $\nu$, all of which have been described in Section~\ref{sec:PECME_E_Step}, the distributed E step of ADECME also involves computing the sufficient statistics for $\mathcal{A}$ and $\boldsymbol{\Psi}$. The details of the calculation of the sufficient statistics for $\mathcal{A}$ and $\boldsymbol{\Psi}$ are as follows:
\begin{equation*}
	\mathbf{S}_{\mathcal{A} 1, ji}^{(t+1)} = \mathbf{1}_{n_{ji}}^{\top} \hat{\boldsymbol{\Sigma}}_{ji}^{(t)^{-1}}\left(\mathbf{Y}_{ji}-\mathbf{X}_{ji} \hat{\boldsymbol{\beta}}^{(t)}\right),
\end{equation*}
\begin{equation*}
	\mathbf{S}_{\mathcal{A} 2, ji}^{(t+1)} = a_{ji}^{(t+1)} \mathbf{1}_{n_{ji}}^{\top} \hat{\boldsymbol{\Sigma}}_{ji}^{(t)^{-1}} \mathbf{1}_{n_{ji}},
\end{equation*}
and
\begin{equation*}
	\begin{aligned}
		\mathbf{S}_{\boldsymbol{\Psi},ji}^{(t+1)} = & 
		\left[b_{ji}^{(t+1)}\left(\mathbf{Y}_{ji}-\hat{\mathbf{M}}_{ji}^{(t)}\right)^{\top} \hat{\boldsymbol{\Sigma}}_{ji}^{(t)^{-1}}\left(\mathbf{Y}_{ji}-\hat{\mathbf{M}}_{ji}^{(t)}\right)\right. \\
		& -\hat{\mathbf{A}}_{ji}^{(t)^{\top}} \hat{\boldsymbol{\Sigma}}_{ji}^{(t)^{-1}}\left(\mathbf{Y}_{ji}-\hat{\mathbf{M}}_{ji}^{(t)}\right) \\
		& -\left(\mathbf{Y}_{ji}-\hat{\mathbf{M}}_{ji}^{(t)}\right)^{\top} \hat{\boldsymbol{\Sigma}}_{ji}^{(t)^{-1}} \hat{\mathbf{A}}_{ji}^{(t)} \\
		& \left.+a_{ji}^{(t)} \hat{\mathbf{A}}_{ji}^{(t)^{\top}} \hat{\boldsymbol{\Sigma}}_{ji}^{(t)^{-1}} \hat{\mathbf{A}}_{ji}^{(t)}\right].
	\end{aligned}
\end{equation*}

Furthermore, the grid search algorithm described in Step (5) of Section \ref{sec:PECME_M_Step} is incorporated into the distributed E step of ADECME. During the grid search, the workers utilize $\hat{\boldsymbol{\beta}}^{(t)}, \hat{\nu}^{(t)}, \hat{\mathcal{A}}^{(t)}, \hat{\boldsymbol{\Psi}}^{(t)}$, and $\hat{\rho}_{2}^{(t)}$ to evaluate the observed log-likelihood for the update of $\rho_{1}$, and they use $\hat{\boldsymbol{\beta}}^{(t)}, \hat{\nu}^{(t)}, \hat{\mathcal{A}}^{(t)}, \hat{\boldsymbol{\Psi}}^{(t)}$, and $\hat{\rho}_{1}^{(t)}$ to evaluate the observed log-likelihood for the update of $\rho_{2}$.

\subsubsection{ADECME: The Distributed CM Step} \label{sec:ADECME_M_step}

Once the manager receives all sufficient statistics from the workers at the end of the distributed E step, \emph{no further communication} between the manager and workers is required for the remainder of the iteration. All parameter updates in the CM step are performed solely by the manager using the aggregated sufficient statistics, as detailed below: 

\begin{enumerate}[(1)]
	
	\item Update $\boldsymbol{\beta}$.
	
	The manager updates the estimation of $\boldsymbol{\beta}$ as
	\begin{equation*}
		\hat{\boldsymbol{\beta}}^{(t+1)} = \left(\sum_{j=1}^{b} \mathbf{S}_{\boldsymbol{\beta} 1, j}^{(t+1)}\right)^{-1} \left(\sum_{j=1}^{b} \mathbf{S}_{\boldsymbol{\beta} 2, j}^{(t+1)}\right).
	\end{equation*}
	
	\item Update $\nu$.
	
	The manager updates the estimation of $\nu$ as the solution to 
	\begin{equation*}
		\log \left(\frac{\nu}{2}\right)+1-\varphi\left(\frac{\nu}{2}\right)-\frac{1}{N} \sum_{j=1}^{b}\left(\mathbf{S}_{\nu, j}^{(t+1)}\right)=0.
	\end{equation*}
	
	\item Update $\mathcal{A}$.
	
	The manager aggregates $\mathbf{S}_{\mathcal{A} 1, ji}^{(t+1)}$ and $\mathbf{S}_{\mathcal{A} 2, ji}^{(t+1)}$ as follows:
	\begin{equation*}
		\mathbf{S}_{\mathcal{A} 1, j}^{(t+1)} = \sum_{i=1}^{N_{j}}  \mathbf{S}_{\mathcal{A} 1, ji}^{(t+1)},
	\end{equation*}
	\begin{equation*}
		\mathbf{S}_{\mathcal{A} 2, j}^{(t+1)} = \sum_{i=1}^{N_{j}}  \mathbf{S}_{\mathcal{A} 2, ji}^{(t+1)}.
	\end{equation*}
	
	After the aggregation, the manager updates the estimation of $\mathcal{A}$ as:
	\begin{equation*}
		\hat{\mathcal{A}}^{(t+1)} = \frac{\sum_{j=1}^{b} \mathbf{S}_{\mathcal{A} 1, j}^{(t+1)}}{\sum_{j=1}^{b} \mathbf{S}_{\mathcal{A} 2, j}^{(t+1)}}.
	\end{equation*}
	
	\item Update $\boldsymbol{\Psi}$.
	
	The manager aggregates $\mathbf{S}_{\boldsymbol{\Psi},ji}^{(t+1)}$ as:
	\begin{equation*}
		\mathbf{S}_{\boldsymbol{\Psi},j}^{(t+1)} = \sum_{i=1}^{N_{j}}  \mathbf{S}_{\boldsymbol{\Psi},ji}^{(t+1)}.
	\end{equation*}
	
	After the aggregation, the manager updates the estimation of $\boldsymbol{\Psi}$ as:
	\begin{equation*}
		\hat{\boldsymbol{\Psi}}^{(t+1)} = \frac{\sum_{j=1}^{b}\mathbf{S}_{\boldsymbol{\Psi},j}^{(t+1)}}{\sum_{j=1}^{b}\sum_{i=1}^{N_{j}}  n_{ji}}.
	\end{equation*}
	
	\item Update $\rho_{1}$ and $\rho_{2}$ from the DEC structure using grid search. 
	
	The manager aggregates the calculated values of the log-likelihood in the distributed E step in Section~\ref{sec:ADECME_E_Step} and then selects the values of $\rho_{1}$ and $\rho_{2}$ that maximize the observed log-likelihood, resulting in $\hat{\rho}_{1}^{(t+1)}$ and $\hat{\rho}_{2}^{(t+1)}$.
	
\end{enumerate}

\subsection{Convergence Criteria}

For all three algorithms, ECME, PECME, and ADECME, we employ the same stopping criterion: 
\begin{equation}
	\max_{i} \left|\hat{\boldsymbol{\vartheta}}_{i}^{(t+1)}-\hat{\boldsymbol{\vartheta}}_{i}^{(t)}\right|<\epsilon,
	\label{eq:stopping criterion}
\end{equation}
where $\hat{\boldsymbol{\vartheta}}_{i}^{(t+1)}$ denotes the $i$-th element of the vector of parameters of interest at the current iteration, and $\epsilon$ is a small positive number, such as $1 \times 10^{-7}$. We did not use the change of the observed log-likelihood, which is another commonly used stopping criterion. This is due to our large sample setting, where the evaluation of observed log-likelihood is very time-consuming and eventually slows down all three algorithms. As suggested by \citet{wu1983convergence}, multiple random initial values should be used to avoid the proposed algorithms halt at a local stationary point. Additionally, we suggest imposing a cap on the maximum number of iterations, set to $1000$, to prevent situations where the random initial values are too distant from the true values, potentially leading to excessively long computation times.

\subsection{Comparative architecture of the three algorithms}

In this section, we delineate the differences between the ECME, PECME, and ADECME algorithms, as further illustrated by their respective pseudo-codes. The ECME algorithm (Algorithm \ref{alg:ecme}) provides the foundational framework for parameter estimation but is computationally prohibitive for large datasets due to its serial E step calculations and the need for the observed-data likelihood evaluations to update the DEC parameters. 

\begin{algorithm}[ht]
	\caption{ECME Algorithm (Details in Section~\ref{sec:ECME})}
	\label{alg:ecme}
	\begin{algorithmic}[1]
		\State \textbf{Input:} observed data $\Dcal_{\text{obs}}$, initial parameter $\boldsymbol{\vartheta}^{(0)}$.
		\State \textbf{Set:} $t \gets 0$.
		\Repeat
		\State \textbf{E Step:} Compute statistics $a_i^{(t+1)}, b_i^{(t+1)}, c_i^{(t+1)}$ using on $\boldsymbol{\vartheta}^{(t)}$ for all $i = 1,\dots,N$.
		\State \textbf{CM Step 1:} Update $\boldsymbol{\beta}^{(t+1)}$ given $\mathcal{A}^{(t)},\rho_{1}^{(t)},\rho_{2}^{(t)}$ and statistics from E step.
		\State \textbf{CM Step 2:} Update $\nu^{(t+1)}$ and statistics from E step.
		\State \textbf{CM Step 3:} Update $\mathcal{A}^{(t+1)}$ given $\boldsymbol{\beta}^{(t+1)},\rho_{1}^{(t)},\rho_{2}^{(t)}$ and statistics from E step.
		\State \textbf{CM Step 4:} Update $\boldsymbol{\Psi}^{(t+1)}$ given $\boldsymbol{\beta}^{(t+1)},\mathcal{A}^{(t+1)},\rho_{1}^{(t)},\rho_{2}^{(t)}$ and statistics from E step.
		\State \textbf{CM Step 5:} Update $\rho_{1}^{(t+1)}$ using a grid search, given $\boldsymbol{\beta}^{(t+1)},\mathcal{A}^{(t+1)},\boldsymbol{\Psi}^{(t+1)},\nu^{(t+1)},\rho_{2}^{(t)}$.
		\State \textbf{CM Step 6:} Update $\rho_{2}^{(t+1)}$ using a grid search, given $\boldsymbol{\beta}^{(t+1)},\mathcal{A}^{(t+1)},\boldsymbol{\Psi}^{(t+1)},\nu^{(t+1)},\rho_{1}^{(t+1)}$.
		\State  \textbf{Set:} $t \gets t + 1$.
		\State \textbf{Convergence Check.}
		\Until{stopping criterion \eqref{eq:stopping criterion} is met.}
		\State \textbf{Output:} $\hat{\boldsymbol{\vartheta}} =\boldsymbol{\vartheta}^{(t)}$
	\end{algorithmic}
\end{algorithm}

The PECME algorithm (Algorithm \ref{alg:pecme}) addresses this bottleneck by parallelizing the E step across multiple workers, distributing the computational load. However, in addition to the distributed E step, its design necessitates four more synchronous manager-worker communications per iteration for updating parameters like $\mathcal{A}$, $\boldsymbol{\Psi}$, $\rho_1$, and $\rho_2$, which introduces significant synchronization overhead and limits its scalability. 

\begin{algorithm}[ht]
	\caption{PECME Algorithm (Details in Section~\ref{sec:PECME})}
	\label{alg:pecme}
	\begin{algorithmic}[1]
		\State \textbf{Input:} observed data $\Dcal_{\text{obs}}$, initial parameter $\boldsymbol{\vartheta}^{(0)}$, the number of workers $k$.
		\State \textbf{Split Data:} Split $\Dcal_{\text{obs}}$ into $k$ disjoint subsets.
		\State \textbf{Set:} $t \gets 0$.
		\Repeat
		\State \textbf{E Step:} Compute $a_{ji}^{(t+1)}, b_{ji}^{(t+1)}, c_{ji}^{(t+1)}$ and sufficient statistics for $\boldsymbol{\beta}, \nu$ in  parallel with $k$ workers.
		\State \textbf{CM Step 1:} Update $\boldsymbol{\beta}^{(t+1)}$ given sufficient statistics from E step.
		\State \textbf{CM Step 2:} Update $\nu^{(t+1)}$ given sufficient statistics from E step.
		\State \textbf{CM Step 3a:} Update sufficient statistics for $\mathcal{A}$ in  parallel with $k$ workers. 
		\State \textbf{CM Step 3b:} Update $\mathcal{A}^{(t+1)}$ with the updated sufficient statistics from CM Step 3a.
		\State \textbf{CM Step 4a:} Update sufficient statistics for $\boldsymbol{\Psi}$ in  parallel with $k$ workers. 
		\State \textbf{CM Step 4b:} Update $\boldsymbol{\Psi}^{(t+1)}$ with the updated sufficient statistics from CM Step 4a..
		\State \textbf{CM Step 5:} Update $\rho_{1}^{(t+1)}$ using a grid search, given $\boldsymbol{\beta}^{(t+1)},\mathcal{A}^{(t+1)},\boldsymbol{\Psi}^{(t+1)},\nu^{(t+1)},\rho_{2}^{(t)}$. The evaluation of the observed log-likelihood for this update is parallelized across $k$ workers.
		\State \textbf{CM Step 6:} Update $\rho_{2}^{(t+1)}$ using a grid search, given $\boldsymbol{\beta}^{(t+1)},\mathcal{A}^{(t+1)},\boldsymbol{\Psi}^{(t+1)},\nu^{(t+1)},\rho_{1}^{(t+1)}$. The evaluation of the observed log-likelihood for this update is parallelized across $k$ workers.
		\State  \textbf{Set:} $t \gets t + 1$.
		\State \textbf{Convergence check.}
		\Until{stopping criterion \eqref{eq:stopping criterion} is met.}
		\State \textbf{Output:} $\hat{\boldsymbol{\vartheta}} =\boldsymbol{\vartheta}^{(t)}$
	\end{algorithmic}
\end{algorithm}

In contrast, the ADECME algorithm (Algorithm \ref{alg:adecme}) is designed for superior computational efficiency. It employs an asynchronous E step, proceeding once a predefined fraction of workers report their results, and crucially computes all sufficient statistics for the CM step, including those for the DEC parameters via grid search using previous-iteration values, concurrently within this single, reduced-communication step. This integrated approach, where the manager performs all subsequent updates without further communication, minimizes idle time and synchronization delays, making ADECME the most communication-efficient and scalable variant for large-scale inference.

\begin{algorithm}[ht]
	\caption{ADECME Algorithm (Details in Section~\ref{sec:ADECME})}
	\label{alg:adecme}
	\begin{algorithmic}[1]
		\State \textbf{Input:} observed data $\Dcal_{\text{obs}}$, initial parameter $\boldsymbol{\vartheta}^{(0)}$, the number of workers $k$, fraction $\gamma$.
		\State \textbf{Split Data:} Split $\Dcal_{\text{obs}}$ into $k$ disjoint subsets.
		\State \textbf{Set:} $t \gets 0$.
		\State \textbf{E Step:} Compute $a_{ji}^{(t+1)}, b_{ji}^{(t+1)}, c_{ji}^{(t+1)}$, sufficient statistics for $\boldsymbol{\beta}, \nu, \mathcal{A}, \boldsymbol{\Psi}$ and observed log-likelihood for $\rho_1,\rho_2$ in parallel with $k$ workers.
		\State \textbf{CM Step 1:} Update $\boldsymbol{\beta}^{(t+1)}$ given sufficient statistics from E step.
		\State \textbf{CM Step 2:} Update $\nu^{(t+1)}$ given sufficient statistics from E step.
		\State \textbf{CM Step 3:} Update $\mathcal{A}^{(t+1)}$ given sufficient statistics from E step.
		\State \textbf{CM Step 4:} Update $\boldsymbol{\Psi}^{(t+1)}$ with the sufficient statistics from E step.
		\State \textbf{CM Step 5:} Update $\rho_{1}^{(t+1)}$ using a grid search. The observed log-likelihood has already been evaluated in E step.
		\State \textbf{CM Step 6:} Update $\rho_{2}^{(t+1)}$ using a grid search. The observed log-likelihood has already been evaluated in E step.
		\State \textbf{Set:} $t \gets 1$.
		\Repeat
		\State \textbf{Asynchronous E Step:} Compute $a_{ji}^{(t+1)}, b_{ji}^{(t+1)}, c_{ji}^{(t+1)}$, sufficient statistics for $\boldsymbol{\beta}, \nu, \mathcal{A}, \boldsymbol{\Psi}$ and observed log-likelihood for $\rho_1,\rho_2$ using asynchronous parallel algorithm. Proceed to the CM Steps once a proportion $\gamma$ of $k$ workers have completed their calculations.
		\State \textbf{CM Steps:} Update $\boldsymbol{\vartheta}^{(t+1)}$ as Lines 5 - 10.
		\State  \textbf{Set:} $t \gets t + 1$.
		\State \textbf{Convergence check.}
		\Until{stopping criterion \eqref{eq:stopping criterion} is met.}
		\State \textbf{Output:} $\hat{\boldsymbol{\vartheta}} =\boldsymbol{\vartheta}^{(t)}$
	\end{algorithmic}
\end{algorithm}

To further demonstrate the operational differences between ADECME and PECME, we present architectural overviews in Appendix~\ref{sec:Architectural Overview}. As shown in Figure~\ref{fig:architectural overview of the PECME algorithm}, PECME requires five synchronous manager-worker communications per iteration and updates all sufficient statistics in every distributed E step. In contrast, Figure~\ref{fig:architectural overview of the ADECME algorithm} illustrates that ADECME uses an asynchronous approach, where only a fraction of workers contribute updated statistics in each iteration, with stale values from slower workers being reused. Critically, after the asynchronous distributed E step, no further communication occurs between the manager and workers during the CM steps. This fundamental difference in synchronization and communication patterns underlies ADECME's superior scalability to large-scale inference problems.

\subsection{Convergence assessment of ADECME}\label{sec:convergence_theorem_ADECME}

We derive a lower bound for the matrix rate and speed of convergence for our ADECME algorithm. \citet{dempster1977maximum} and \citet{Men94} show that the convergence rate and speed of EM-type algorithms depend on the observed and complete data information matrices. Their approach is inapplicable to our setting due to the partial updates of the ADECME algorithm, where only a $\gamma$ fraction of the sufficient statistics are updated in every iteration. \citet{NeaHin98} develop an EM extension that uses a fraction of the samples in an iteration. This extension is an instance of the class of online EMs \citep{CapMou09} that use stochastic approximation for enhancing the efficiency of EM-type algorithms. 

Our ADECME algorithm is based on the Distributed EM framework, which uses the full data but updates only a fraction of the sufficient statistics in every iteration \citep{srivastava2019asynchronous,zhou2023asynchronous}. It is the distributed extension of the parent ECM algorithm for parameter estimation in a matrix-variate $t$ setting. Due to the partial ADECME updates, the likelihood sequence obtained from ADECME  is not guaranteed to increase in every iteration; however, the ADECME likelihood sequence still converges as shown in the following proposition, based on Theorem 1 in \citet{NeaHin98}.

\begin{proposition}\label{ad-prop1}
	Let $\tilde p$ be a probability density on the space of missing data $\wb = (W_1, \ldots, W_N)$,  $\ell_C(\varthetab)$ and $ \ell (\varthetab)$ be the complete and observed data log likelihood in \eqref{eq:complete_data_loglikelihood}, and $\EE_{\wb}$ be the expectation with respect to density of $\wb$. Define the following objective function of ($\tilde p, \varthetab$):
	\begin{align*}
		\Fcal(\tilde p, \varthetab) = \EE_{\wb} \left\{ \ell_C(\varthetab)\right\} - \EE_{\wb} \left\{ \log \tilde p(\wb) \right\}, \quad \tilde p(\wb) = \prod_{j=1}^k \prod_{i=1}^{N_j} p(w_{ji} \mid \Yb_{ji}, \Xb_{ji}, \tb_{ji}, \varthetab_j) \equiv \prod_{j=1}^k p_j, 
	\end{align*}
	where worker $j$ performs its local E step using $p_j$ by setting $\varthetab = \varthetab_j$. Let $\{ \varthetab^{(t)} \}$ be the $\varthetab$ estimate sequence generated by ADECME and $\tilde p^{(t)} = \prod_{j_1 \in \Rcal_{t}} \tilde p_{j_1}^{(t-1)} \prod_{j_0 \in \Rcal_{t}^c} \tilde p_{j_0}^{(t-1)}$, where $\Rcal_{t}$ includes the indices of workers that returned their results to the manager at the end of $t$th ADECME iteration, $p_{j_1}^{(t-1)}$ equals $p_{j_1}$ evaluated with $\varthetab = \varthetab^{(t-1)}$, and $p_{j_0}^{(t-1)}$ equals $p_{j_0}$ evaluated with $\varthetab = \varthetab^{(t_{j_0})}$ for some $t_{j_0} < t-1$.
	Then, ADECME iterations do not decrease the $\{\Fcal(\tilde p^{(t)}, \varthetab^{(t)})\}$ sequence. 
	Furthermore, if the $\{\Fcal(\tilde p^{(t)}, \varthetab^{(t)})\}$ sequence converges to a stationary point $\hat \Fcal = \Fcal( \widehat {\tilde p }, \hat \varthetab)$, then the observed data likelihood sequence $\ell(\varthetab^{(t)})$ converges to $\ell(\hat \varthetab)$.
	%Then, at the $(t+1)$-th iteration ($t=0,1, 2, \ldots$), the E step in ADECME maximizes $\Fcal(\tilde p, \varthetab^{(t)}$ with respect to $\tilde p$ to obtain $\tilde p^{(t+1)}$ and CM step in ADECME maximizes $\Fcal(\tilde p^{(t+1)}, \varthetab)$ with respect to $\varthetab$ to obtain $\varthetab^{(t+1)}$.  
\end{proposition}

Proposition \ref{ad-prop1} guarantees that the $\Fcal(\tilde p^{(t)}, \varthetab^{(t)})$ is monotonic but not the $\ell(\varthetab^{(t)})$ sequence. Unlike the ECM algorithm in \citet{gallaugher2017matrix}, the ADECME likelihood sequence is not monotonic, but the convergence of $\{\ell(\varthetab^{(t)})\}$ sequence is guaranteed via the convergence of $\{\Fcal(\tilde p^{(t)}, \varthetab^{(t)})\}$ sequence. \cite{wu1983convergence} shows that the convergence of $\{\ell(\varthetab^{(t)})\}$ does not imply convergence of the $\{\varthetab^{(t)}\}$ sequence. 
To guarantee the convergence of ADECME sequence $\{\varthetab^{(t)}\}$, we require the following two assumptions:
\begin{enumerate}
	\item[A1.] With a small probability $\zeta > 0$, we wait for all the workers to return their results to the manager. The manager waits to hear from a $\gamma$ fraction of workers with a large probability $1 - \zeta$.
	\item[A2.] The stationary points $(\widehat {\tilde p}, \hat \varthetab )$ lie in the interior of $\tilde P \otimes \Thetab$, where $\tilde P$ and $\Thetab $ are space of all probability measures on $\wb$ and parameter space of the MVST distribution, respectively.
\end{enumerate}

Assumption A1 is a technical condition that guarantees the manager receives results from every worker as the ADECME progresses, thereby preventing artifacts caused by computational or communication load imbalance \citep{zhou2023asynchronous}. Assumption A2 is used to show that the $\{\varthetab^{(t)}\}$ sequence converges if the $\{\ell(\varthetab^{(t)})\}$ sequence converges. With these assumptions, we have the following proposition guaranteeing the convergence of the ADECME sequence $\{\varthetab^{(t)}\}$. 

\begin{proposition}\label{ad-prop2}
	If the previous two assumptions A1 and A2 hold, then the ADECME sequence $\{\varthetab^{(t)}\}$ converges to $\hat \varthetab$, which is either a stationary point or a maximizer of $\ell(\varthetab)$.
\end{proposition}

Our next result is about the rate of convergence of the ADECME sequence $\{ \varthetab^{(t)} \}$. The previous two propositions identify conditions that guarantee the convergence of $\{ \varthetab^{(t)} \}$ to a stationary point. The convergence rate defines the speed at which $\| \varthetab^{(t)} - \hat \varthetab\| $ decays with $t$. \citet{dempster1977maximum}
and \citet{Men94} show that the rate and speed of convergence depends on the complete and observed data information matrices. 
For simplicity, we assume that $\Sigmab_i$'s equal $\Sigmab$, an $n \times n$ positive definite matrix, and we treat $\Sigmab$ as a parameter. Our derivation of these matrices depend on the relationship between the matrix and vector variate Skew $t$ distributions. Specifically, 
\begin{align} \label{vecmvst}
	\Yb_i \sim \text{MVST}(\Xb_i \betab, \Ab_i, \Sigmab, \Psib, \nu) \iff
	\yb_i \sim \text{ST}(\Ib \otimes \Xb_i \vv(\betab), \vv(\Ab_i), \Psib \otimes \Sigmab, \nu);
\end{align}
see Eq.~(9) in \citetalias{gallaugher2017matrix}. Using the equivalence in \eqref{vecmvst}, we derive the analytic form of the complete and observed data information matrices in the Appendix; see Theorems  \ref{thm1} and \ref{thm2}.

We now derive a lower bound for the matrix rate of convergence of ADECME algorithm. Let $\hat \varthetab$ be the stationary point of the ADECME sequence $\{\varthetab^{(t)} \}$,  $N$ be the sample size, $\Rb$ be the matrix rate of convergence, $\Sbb$ be the matrix speed of convergence, $\Ib_{c,i}$ and $\Ib_{o,i}$ be the complete data and observed data information matrix for the $i$the sample, $i=1, \ldots, N$. Then, \citet{Men94} shows that $\Rb$ and $\Sbb$ are defined as follows:
\begin{align}\label{eq:rate-em1}
	\Ib_{c_N}  = \sum_{i=1}^N \Ib_{c,i}, \quad \Ib_{o_N} = \sum_{i=1}^N \Ib_{o,i}, \quad \Sbb = \Ib_{c_N}^{-1} \Ib_{o_N}, \quad \Rb = \Ib -  \Ib_{c_N}^{-1} \Ib_{o_N},\quad \Rb  = \Ib - \Sbb ,
\end{align}
where $\Ib$ is a $d \times d$ identity matrix, $\Sbb$ and $\Rb$ are $d \times d$ positive definite matrices,  and \eqref{vecmvst} implies that $d = {pq + p + n(n+1)/2 + p(p+1)/2 + 1}$.  Theorems \ref{thm1} and \ref{thm2} in the Appendix define the analytic forms of  $\Ib_{c,i}$ and $\Ib_{o,i}$ for each $i$.  The rate and speed of convergence equal $r_{\max} = \lambda_{\text{max}}(\Rb)$ and $s_{\min} = \lambda_{\text{min}}(\Sbb) = 1 - r_{\max}$. The following proposition derives the analytic forms for $\Rb$ and $\Sbb$.

\begin{proposition}
	Let $\hat \varthetab$ be the stationary point of the ADECME algorithm for estimating $\varthetab= (\text{vec}(\betab ), \ab,  \text{vech}(\Sigmab), \text{vech}(\Psib), \nu)$ in the MVST regression model in \eqref{vecmvst} using the complete data model based on \eqref{eq:complete_data_loglikelihood}. Denote the rate  of convergence of the ADECME algorithm for parameter estimation as $r_{\max}$. Assume that 
	
	\begin{enumerate}
		\item The parameter space $\Thetab$ is a compact subset of $\RR^d$  and $\nu > 4$.
		\item In a small neighborhood around the stationary point $\hat \varthetab$, the gradient and Hessian of $\Qcal(\cdot \mid \cdot)$ are regular in the sense that for any $\varthetab, \varthetab'$ in a small neighborhood around $\hat \varthetab$, 
		\begin{align}\label{eq:q-regular}
			\mathrm{D}^{10} \Qcal(\varthetab \mid \varthetab') = \mathrm{D}^{10} \Qcal (\varthetab \mid \varthetab) + o(1), \quad 
			\mathrm{D}^{20} \Qcal (\varthetab \mid \varthetab') = \mathrm{D}^{20} \Qcal(\varthetab \mid \varthetab) - \Deltab, 
		\end{align}
		where $o(1)$ is a $d$-dimensional vector whose norm goes to zero as the neighborhood radius shrinks to 0 and $\Deltab$ is a $d \times d$ positive definite matrix with bounded eigen values. 
	\end{enumerate} 
	Then, for a sufficiently large $t$, $r_{\max} \leq \lambda_{\max}(\Rb  + \tilde \Deltab_{\gamma})$, where $\Rb$ is the rate of convergence matrix defined in \eqref{eq:rate-em1} for the EM that use the full data and  $\tilde \Deltab_{\gamma} = (1 - \gamma) \Sbb \{  \Ib  + (1 - \gamma) \Ib_{c_N}^{-1}\Deltab \}^{-1}  \Ib_{c_N}^{-1}\Deltab$.
\end{proposition}

The proof of this proposition is provided in Appendix~\ref{rate-proof}. The term $\lambda_{\max}(\Rb  + \tilde \Deltab_{\gamma})$
characterizes the convergence rate of the standard EM algorithm without acceleration; thus, its largest eigenvalue serves as an upper bound for $r_{\max}$. The matrix $\tilde \Deltab_{\gamma}$
is positive definite, and its eigenvalues are scaled by the factor $(1 - \gamma)$, representing the proportion of samples excluded in each iteration of the ADECME algorithm. This correction term  quantifies the impact of asynchronous and distributed updates: by omitting an $(1-\gamma)$-fraction of samples, the algorithm exhibits a slower theoretical convergence rate; however, each iteration is substantially faster, as computations involve only a $\gamma$-fraction of the data, resulting in significant overall efficiency gains in real time. Finally, $s_{\min} = 1 - r_{\max} \geq 1 - \lambda_{\max}(\Rb  + \tilde \Deltab_{\gamma})$.

\section{Simulation Study} \label{sec:simulation_study}

In this section, we conducted extensive simulation studies using synthetic data generated from the following three schemes to compare the ECME, PECME, and ADECME algorithms in light of achieving accuracy of point estimation, computational time, and robustness to distributional assumptions of the noise, for various sample sizes. 

In Schemes I and II, we generated samples $\left\{\left(\mathbf{Y}_1, \mathbf{X}_1, \boldsymbol{t}_1\right), \ldots,\left(\mathbf{Y}_N, \mathbf{X}_N, \boldsymbol{t}_N\right)\right\}$ from the REGMVST model: $\mathbf{Y}_{i} = \mathbf{X}_{i} \boldsymbol{\beta} + \boldsymbol{e}_{i}$, where, for each subject $i = 1, \dots, N$, the number of observations is $n_{i} = z_{i} + 2$, with $z_{i}$ following a Poisson distribution with a mean of $8$, ensuring that each subject has at least two observations. The first column of $\mathbf{X}_{i}$ consists of samples from an exponential distribution with a mean of $1$, the second column is generated from a standard normal distribution, and the third column is drawn from a Bernoulli distribution with a mean of $2\Phi\left(\left|t_i\right| - 1\right)$, where $\Phi(\cdot)$ is the cumulative density function of the standard normal distribution. Here, $|t_{i}|$, representing the time of each observation, follows a zero-truncated standard normal distribution, making the third column of $\mathbf{X}_{i}$ time-dependent, with its mean drawn from a standard uniform distribution. The noise term $\boldsymbol{e}_{i}$ was generated from a matrix variate skew-t distribution $\operatorname{MVST}\left(\boldsymbol{0}_{i}, \boldsymbol{1}_{i} \mathcal{A}, \boldsymbol{\Sigma}_{i}, \boldsymbol{\Psi}, \nu\right)$, where $\boldsymbol{0}_{i}$ is an $n_{i}$ by $2$ matrix of zeros, $\boldsymbol{1}_{i}$ is a vector of ones of length $n_{i}$, and $\boldsymbol{\Sigma}_{i}$ is a correlation matrix following the DEC structure, as defined in \eqref{eq:DEC_definition}. The true values of the model parameters are:
$$
\boldsymbol{\beta} = 
\begin{bmatrix}
	0.5 & 0.5 \\
	1.5 & 1.5 \\
	-0.5 & -0.5
\end{bmatrix},
$$
$$
\mathcal{A} = \begin{bmatrix} 2.0 & -2.0 \end{bmatrix},
$$
$$
\boldsymbol{\Psi} = 
\begin{bmatrix}
	1.0 & -0.5 \\
	-0.5 & 1.0
\end{bmatrix},
$$
with $\nu = 5$, $\rho_{1} = 0.9$, and $\rho_{2} = 0.8$. However, in Scheme III, we tested the robustness of the REGMVST model by altering the noise term $\boldsymbol{e}_{i}$ to follow a matrix-variate generalized hyperbolic distribution \citep{thabane2004matrix}. In this case, the latent variable $W_{i}$ has no degrees of freedom, but two other associated parameters are present, while all other parameters remain unchanged.

\subsection{Scheme I}

Here, we aim to demonstrate that the ADECME, PECME, and ECME algorithms lead to identical point estimation at a finite sample size of $N = 250$, and that the ADECME algorithm is faster than the other two, even under a moderate sample sizes. For ADECME, we reserved one core as the manager and the other eight cores as the workers, within the cluster computing environment at the Dept. of Biostatistics, Virginia Commonwealth University. We explored the combinations of $\gamma = \left\{0.625,0.75,0.875\right\}$. This implies the manager waits for $8\times0.625 = 5$, $8\times0.75 = 6$, and $8\times0.875 = 7$ workers, respectively, to complete the computation in the distributed E step described in Section~\ref{sec:ADECME_E_Step}. For the PECME algorithm, we also reserved one core as the manager and the other eight cores as the workers. As discussed before, in the PECME algorithm, the manager waits for \emph{all} workers to complete the computation in the distributed E step described in Section~\ref{sec:PECME_E_Step}. For ECME, we only reserved one core, as the ECME algorithm does not benefit from reserving multiple cores. We repeated the simulation study 50 times.

In Figure \ref{fig:scheme1_computation_time}, we plot the total computational time (in minutes) for the ADECME algorithm, with $\gamma \in \{0.625, 0.750, 0.875\}$, the PECME algorithm, and the ECME algorithm. The boxplot clearly shows that the ADECME algorithm with three different $\gamma$ values is faster than both PECME and ECME algorithms, with the ADECME algorithm achieving the fastest performance when $\gamma = 0.875$. Unsurprisingly, the ECME algorithm is observed to be the slowest. 

\begin{figure}[ht]
	\centering
	\includegraphics[width = 0.8\textwidth]{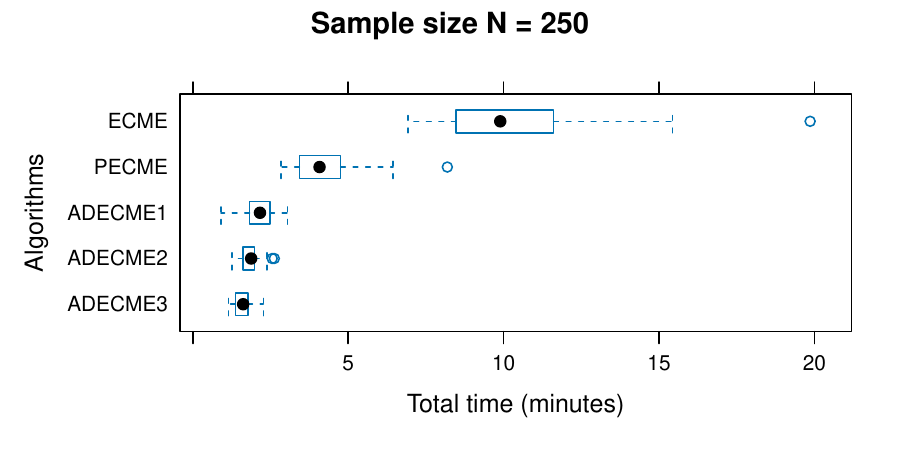}
	\caption{\label{fig:scheme1_computation_time} Total computation time (in minutes) across 50 replicates from fitting the REGMVST model using the proposed ADECME and the competing PECME and ECME algorithms to simulated data generated from Scheme I, with sample size $N = 250$. ADECME1, ADECME2 and ADECME3 represent the ADECME algorithm with $\gamma = 0.625, 0.750$ and $0.875$ respectively.}
\end{figure}

\begin{table}[ht]
	\centering
	\scalebox{1.0}{
		\begin{tabular}[t]{lrrrrr}
			\toprule
			& ADECME1 & ADECME2 & ADECME3 & PECME & ECME\\
			\midrule
			TT & 2.146 (0.436) & 1.836 (0.304) & 1.612 (0.264) & 4.297 (1.092) & 10.462 (2.609)\\
			E step & 2.136 (0.434) & 1.827 (0.303) & 1.605 (0.263) & 0.559 (0.141) & 0.760 (0.188)\\
			DEC & 0.007 (0.001) & 0.005 (0.001) & 0.005 (0.001) & 3.651 (0.930) & 9.656 (2.409)\\
			$\boldsymbol{\Psi}$ & 0.001 (0.000) & 0.001 (0.000) & 0.000 (0.000) & 0.063 (0.015) & 0.036 (0.009)\\
			$\mathcal{A}$ & 0.000 (0.000) & 0.000 (0.000) & 0.000 (0.000) & 0.019 (0.005) & 0.007 (0.002)\\
			$\boldsymbol{\beta}$ & 0.001 (0.000) & 0.001 (0.000) & 0.001 (0.000) & 0.001 (0.000) & 0.001 (0.000)\\
			$\nu$ & 0.002 (0.000) & 0.002 (0.000) & 0.001 (0.000) & 0.003 (0.001) & 0.002 (0.001)\\
			TNI & 253.240 (52.395) & 206.320 (34.468) & 172.920 (28.670) & 281.480 (70.941) & 281.480 (70.941)\\
			\bottomrule
		\end{tabular}
	}
	\caption{\label{tab:scheme1_time_N250} Table entries are the average computational time (standard deviations) in minutes from using the proposed ADECME algorithm, and the PECME and ECME algorithms to fit the MVST regression model to simulated data from Scheme 1, for sample size $N = 250$ across $50$ replicates. TT denotes the average total time, while TNI represents the average total number of iterations. Also, ADECME1--3 represent the ADECME algorithm with $\gamma = 0.625, 0.750$ and $0.875$, respectively.}
\end{table}

Table \ref{tab:scheme1_time_N250} reveals ADECME's computational advantages: while its distributed E step is the most time-consuming, the PECME and ECME spend more time updating DEC parameters ($\rho_1$, $\rho_2$). ECME (without parallelization) averages 9.656 minutes for DEC updates versus PECME's 3.651 minutes (full parallelization). ADECME's asynchronous E step requires only one manager-worker communication round compared to two in PECME/ECME, significantly improving efficiency. Crucially, ADECME's E step time is shorter than PECME's DEC update time per iteration, and it converges in fewer iterations overall. This efficiency stems from ADECME's partial-update nature, which resembles stochastic approximation methods that can accelerate ECME convergence \citep{Toulis2015}. For $\gamma \in {0.625, 0.750, 0.875}$, higher $\gamma$ values reduce iteration counts but increase E step duration, as discussed in \citet{srivastava2019asynchronous}. Empirically, $\gamma = 0.875$ optimally balances E step efficiency and convergence speed.

Finally, we demonstrate that the point estimates from the ADECME, PECME, and ECME algorithms are identical, even under small to moderate sample size setting, as shown in Table \ref{tab:scheme1_pe_N250}. This is evident from the fact that, for all three algorithms, the averages of the point estimates differ only in the third decimal place, with the standard deviations across 50 replicates also appearing nearly identical.

\begin{table}[ht]
	\begin{flushleft}
		\scalebox{1.00}{
			\begin{tabular}[t]{lccc}
				\toprule
				& ADECME1 & ADECME2 & ADECME3 \\
				\midrule
				$\hat{\boldsymbol{\beta}}$ &
				$
				\begin{bmatrix*}[r]
					0.500 (1.926) & 0.500 (2.240)\\
					1.500 (2.514) & 1.499 (2.052)\\
					-0.500 (2.255) & -0.500 (3.098)
				\end{bmatrix*}
				$ &
				$
				\begin{bmatrix*}[r]
					0.500 (1.926) & 0.500 (2.240)\\
					1.500 (2.514) & 1.499 (2.052)\\
					-0.500 (2.255) & -0.500 (3.098)
				\end{bmatrix*}
				$ &
				$
				\begin{bmatrix*}[r]
					0.500 (1.926) & 0.500 (2.240)\\
					1.500 (2.514) & 1.499 (2.052)\\
					-0.500 (2.255) & -0.500 (3.098)
				\end{bmatrix*}
				$ \\
				\addlinespace
				$\hat{\mathcal{A}}$ &
				$
				\begin{bmatrix*}[r]
					2.008 (92.825) & -2.006 (108.180)
				\end{bmatrix*}
				$ &
				$
				\begin{bmatrix*}[r]
					2.007 (92.787) & -2.006 (108.232)
				\end{bmatrix*}
				$ &
				$
				\begin{bmatrix*}[r]
					2.007 (92.780) & -2.006 (108.227)
				\end{bmatrix*}
				$ \\
				\addlinespace
				$\hat{\boldsymbol{\Psi}}$ &
				$
				\begin{bmatrix*}[r]
					0.997 (51.611) & -0.501 (31.593) \\
					-0.501 (31.593) & 1.005 (44.841)
				\end{bmatrix*}
				$ &
				$
				\begin{bmatrix*}[r]
					0.997 (51.660) & -0.501 (31.618) \\
					-0.501 (31.618) & 1.005 (44.828)
				\end{bmatrix*}
				$ &
				$
				\begin{bmatrix*}[r]
					0.997 (51.660) & -0.501 (31.618) \\
					-0.501 (31.618) & 1.005 (44.827)
				\end{bmatrix*}
				$ \\
				\addlinespace
				$\hat{\rho}_{1}$ &
				$0.900 (0.000)$ &
				$0.900 (0.000)$ &
				$0.900 (0.000)$ \\
				$\hat{\rho}_{2}$ &
				$0.800 (0.000)$ &
				$0.800 (0.000)$ &
				$0.800 (0.000)$ \\
				\addlinespace
				$\hat{\nu}$ &
				$5.190 (510.942)$ &
				$5.190 (510.680)$ &
				$5.190 (510.674)$ \\
				\bottomrule
				\vspace*{0.1 in}
		\end{tabular}}
		\scalebox{1.00}{
			\begin{tabular}[t]{lcc}
				\toprule
				& PECME & ECME \\
				\midrule
				$\hat{\boldsymbol{\beta}}$ &
				$
				\begin{bmatrix*}[r]
					0.500 (1.926) & 0.500 (2.240)\\
					1.500 (2.514) & 1.499 (2.052)\\
					-0.500 (2.255) & -0.500 (3.098)
				\end{bmatrix*}
				$ &
				$
				\begin{bmatrix*}[r]
					0.500 (1.926) & 0.500 (2.240)\\
					1.500 (2.514) & 1.499 (2.052)\\
					-0.500 (2.255) & -0.500 (3.098)
				\end{bmatrix*}
				$ \\
				\addlinespace
				$\hat{\mathcal{A}}$ &
				$
				\begin{bmatrix*}[r]
					2.007 (92.780) & -2.006 (108.227)
				\end{bmatrix*}
				$ &
				$
				\begin{bmatrix*}[r]
					2.007 (92.780) & -2.006 (108.227)
				\end{bmatrix*}
				$ \\
				\addlinespace
				$\hat{\boldsymbol{\Psi}}$ &
				$
				\begin{bmatrix*}[r]
					0.997 (51.660) & -0.501 (31.618) \\
					-0.501 (31.618) & 1.005 (44.827)
				\end{bmatrix*}
				$ &
				$
				\begin{bmatrix*}[r]
					0.997 (51.660) & -0.501 (31.618) \\
					-0.501 (31.618) & 1.005 (44.827)
				\end{bmatrix*}
				$ \\
				\addlinespace
				$\hat{\rho}_{1}$ &
				$0.900 (0.000)$ &
				$0.900 (0.000)$ \\
				$\hat{\rho}_{2}$ &
				$0.800 (0.000)$ &
				$0.800 (0.000)$ \\
				\addlinespace
				$\hat{\nu}$ &
				$5.190 (510.675)$ &
				$5.190 (510.675)$ \\
				\bottomrule
		\end{tabular}}
	\end{flushleft}
	\caption{\label{tab:scheme1_pe_N250} Table entries are the point estimates (100 $\times$ standard deviation) of the fixed effects parameters, averaged across the 50 replicates, obtained from fitting the MVST model to simulated data generated under Scheme I, via the ADECME, PECME and the standard ECME algorithms, for subject size $N = 250$. Here, ADECME1, ADECME2, and ADECME3 denote the ADECME algorithm with $\gamma = 0.625$, $0.750$, and $0.875$, respectively.}
\end{table}

\subsection{Scheme 2}\label{sec:simulation_scheme2}

Here, we compare the performance of the ADECME and PECME algorithms under larger sample sizes. First, we aim to show that the ECME algorithm becomes impractical at this big data setting by comparing the computational time of the ADECME, PECME, and ECME algorithms for one simulated data with size $N = 25,000$ -- reflective of the sample size in the real data. Second, we aim to demonstrate that the ADECME algorithm yields identical point estimations compared to the PECME algorithm, while maintaining its computational advantage for large sample sizes of $N = 25,000$ and $N = 100,000$ with $10$ Monte-Carlo replicates. In this scheme, we requested 65 cores of one CPU and assigned one core as the manager and the remaining 64 cores as the workers.

\begin{table}[ht]
	\centering
	\begin{tabular}[t]{lccc}
		\toprule
		& ADECME & PECME & ECME\\
		\midrule
		Time & $11.229$ & $36.930$ & $901.441$\\
		\addlinespace
		$\hat{\boldsymbol{\beta}}$ &
		$
		\begin{bmatrix*}[r]
			0.500 & 0.500 \\
			1.500 & 1.500 \\
			-0.500 & -0.500
		\end{bmatrix*}
		$ &
		$
		\begin{bmatrix*}[r]
			0.500 & 0.500 \\
			1.500 & 1.500 \\
			-0.500 & -0.500
		\end{bmatrix*}
		$ &
		$
		\begin{bmatrix*}[r]
			0.500 & 0.500 \\
			1.500 & 1.500 \\
			-0.500 & -0.500
		\end{bmatrix*}
		$\\
		\addlinespace
		$\hat{\mathcal{A}}$ & 
		$
		\begin{bmatrix*}[r]
			2.006 & -2.006
		\end{bmatrix*}
		$ & 
		$
		\begin{bmatrix*}[r]
			2.006 & -2.006
		\end{bmatrix*}
		$ & 
		$
		\begin{bmatrix*}[r]
			2.006 & -2.006
		\end{bmatrix*}
		$\\
		
		\addlinespace
		
		$\hat{\boldsymbol{\Psi}}$ & 
		$
		\begin{bmatrix*}[r]
			1.002 & -0.502 \\
			-0.502 & 0.999
		\end{bmatrix*}
		$ & 
		$
		\begin{bmatrix*}[r]
			1.002 & -0.502 \\
			-0.502 & 0.999
		\end{bmatrix*}
		$ & 
		$
		\begin{bmatrix*}[r]
			1.002 & -0.502 \\
			-0.502 & 0.999
		\end{bmatrix*}
		$ \\
		\addlinespace
		$\hat{\rho}_{1}$ & 0.900 & 0.900 & 0.900\\
		$\hat{\rho}_{2}$ & 0.800 & 0.800 & 0.800\\
		\addlinespace
		$\hat{\nu}$ & 5.035 & 5.035 & 5.035\\
		\bottomrule
	\end{tabular}
	\caption{\label{tab:simu2_RNECM} Computation time (in minutes) and point estimates obtained from fitting the REGMVST model to one simulated data with $N = 25,000$ from Scheme II, using the ADECME ($\gamma = 0.875$), the PECME and the ECME algorithms.}
\end{table}

In Table~\ref{tab:simu2_RNECM}, we present the computational time (in minutes) and the point estimates from the ADECME algorithm with $\gamma= 0.875$, the PECME algorithm, and the ECME algorithm, under similar settings as in Scheme I, now with a sample size of $N = 25,000$. We only conducted this simulation once, as the ECME algorithm took more than half a day to converge. This single run is sufficient to demonstrate that the ECME algorithm is impractical for large data settings. All three algorithms yielded identical point estimations when rounded to 3 decimal places.

\begin{figure}[ht]
	\centering
	\includegraphics[width = 0.8\textwidth]{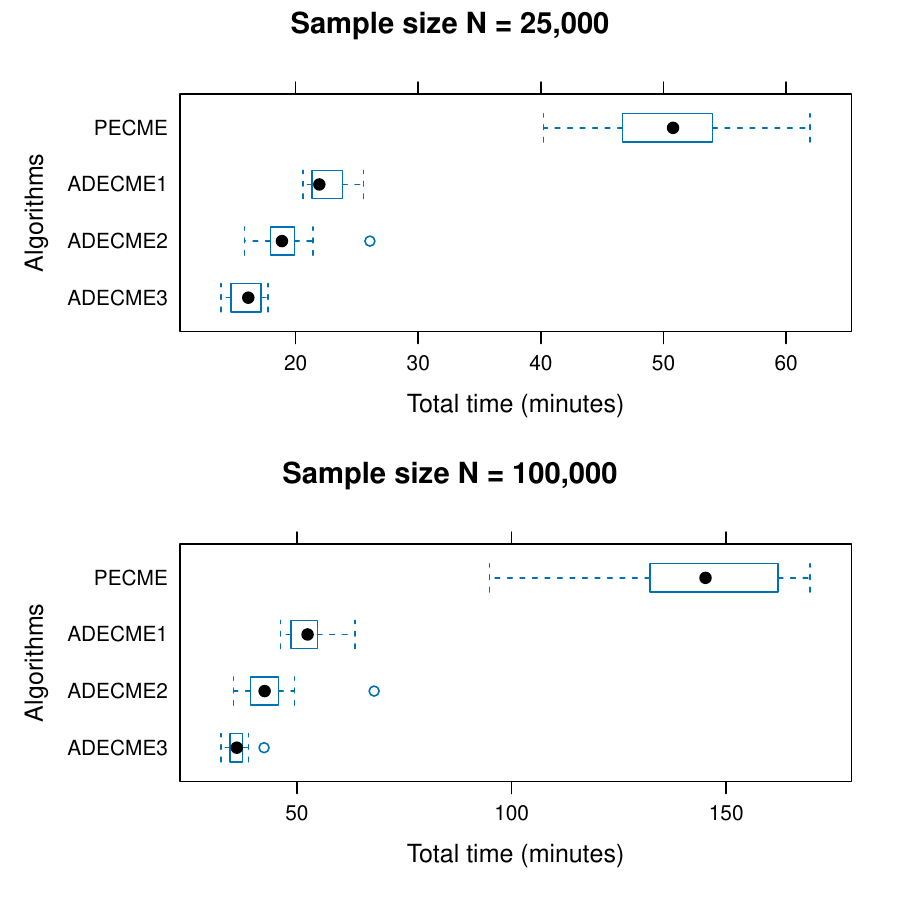}
	\caption{\label{fig:scheme2_computation_time} Total computation time (in minutes) across 10 replicates from fitting the REGMVST model using the proposed ADECME and the competing PECME and ECME algorithms to simulated data generated from Scheme II, with sample sizes $N = 25,000$ and $N = 100,000$. ADECME1, ADECME2 and ADECME3 represent the ADECME algorithm with $\gamma = 0.625, 0.750$ and $0.875$ respectively}
\end{figure}

\begin{table}[ht]
	\centering
	\begin{tabular}[t]{lrrrr}
		\toprule
		$N = 25,000$ & ADECME1 & ADECME2 & ADECME3 & PECME\\
		\midrule
		TT & 22.555 (1.675) & 19.278 (2.913) & 16.006 (1.317) & 50.810 (6.577)\\
		E step & 22.546 (1.674) & 19.270 (2.912) & 15.999 (1.316) & 2.563 (0.273)\\
		DEC & 0.005 (0.000) & 0.004 (0.001)& 0.004 (0.000) & 44.197 (5.884)\\
		$\boldsymbol{\Psi}$ & 0.001 (0.000) & 0.001 (0.000) & 0.001 (0.000) & 3.146 (0.406)\\
		$\mathcal{A}$ & 0.001 (0.000) & 0.000 (0.000) & 0.000 (0.000) & 0.900 (0.139)\\
		$\boldsymbol{\beta}$ & 0.001 (0.000) & 0.001 (0.000) & 0.001 (0.000) & 0.002 (0.000)\\
		$\nu$ & 0.002 (0.000) & 0.001 (0.000) & 0.001 (0.000) & 0.002 (0.000)\\
		TNI & 233.000 (17.404) & 196.200 (29.907) & 160.000 (13.325) & 255.900 (26.409)\\
		\bottomrule
	\end{tabular}
	
	\bigskip  % Add some vertical space between the tables
	
	\begin{tabular}[t]{lrrrr}
		\toprule
		$N = 100,000$ & ADECME1 & ADECME2 & ADECME3 & PECME\\
		\midrule
		TT & 52.777 (5.046) & 44.654 (9.074) & 36.257 (2.861) & 143.414 (23.458)\\
		E step & 52.765 (5.045) & 44.643 (9.071) & 36.248 (2.860) & 7.111 (1.308)\\
		DEC & 0.006 (0.001) & 0.006 (0.001) & 0.005 (0.001) & 121.996 (19.762)\\
		$\boldsymbol{\Psi}$ & 0.001 (0.000) & 0.001 (0.001) & 0.001 (0.000) & 10.748 (1.737)\\
		$\mathcal{A}$ & 0.001 (0.000) & 0.001 (0.000) & 0.001 (0.000) & 3.554 (1.107)\\
		$\boldsymbol{\beta}$ & 0.002 (0.000) & 0.002 (0.000) & 0.001 (0.000) & 0.002 (0.000)\\
		$\nu$ & 0.002 (0.000) & 0.002 (0.000) & 0.002 (0.000) & 0.002 (0.000)\\
		TNI & 253.500 (24.236) & 209.100 (42.686) & 166.100 (13.102) & 307.800 (56.942)\\
		\bottomrule
	\end{tabular}
	\caption{\label{tab:scheme2_combined} Combined results from the simulation study fitting the REGMVST model to data generated under Scheme 2 using 10 replicates, where the upper and lower panels represent sample sizes $N = 25,000$, and  $N = 100,000$, respectively. Here, TT: the average total time, TNI: the average total number of iterations, values in parentheses denoting the standard deviation across the 10 replicates, and ADECME1, ADECME2 and ADECME3 denote the ADECME algorithm with $\gamma = 0.625, 0.750$ and $0.875$, respectively.}
\end{table}

We provide details of the computational time for the ADECME and PECME algorithms in Figure \ref{fig:scheme2_computation_time}, and in Table \ref{tab:scheme2_combined}. The ADECME algorithm with the three different $\gamma$ values was approximately 2-4 times faster than the PECME algorithm for both $N = 25,000$ and $N = 100,000$. Among the ADECME options, $\gamma = 0.875$ appeared to be the most efficient choice under both sample sizes. Additionally, all runs using the ADECME algorithm had smaller total computational times than these using the PECME algorithm, and required fewer iterations to reach convergence. Notably, the ADECME algorithm with $\gamma = 0.875$ required the fewest iterations, and the longest E step per iteration among the three $\gamma$ values. Once again, we observed that the ADECME algorithm with $\gamma = 0.875$ took the least time to complete the study in Scheme II among all algorithms we tried.
Finally, when comparing the most time-consuming steps of the ADECME and PECME algorithms —- the distributional E-step and the DEC parameter updates, respectively, we observe that, owing to fewer communication rounds and the innovative asynchronous parallel mechanism, the distributional E-step in the ADECME consistently required less time per iteration than the DEC parameter updates in PECME.

In Table~\ref{tab:schem2_combined_pe}, we summarize the point estimates from the the ADECME algorithm with $\gamma = 0.625, 0.75$, and $0.875$, as well as the PECME algorithm, for large sample sizes of $N = 25,000$ and $N = 100,000$. With $64$ workers, $\gamma = 0.625, 0.75$, and $0.875$ imply that the manager waits for 40, 48, and 56 workers, respectively, to complete the computation in the distributional E step. The ADECME algorithm with the three different $\gamma$ values and the PECME algorithm yielded identical point estimates, with all absolute biases close to zero, and identical associated standard deviations across the 10 replicates.

\begin{table}[ht]
	\centering
	\scalebox{0.75}{
		\begin{tabular}[t]{lcccc}
			\toprule
			$N = 25,000$ & ADECME1 & ADECME2 & ADECME3 & PECME\\
			\midrule
			$\hat{\boldsymbol{\beta}}$ & 
			$
			\begin{bmatrix*}[r]
				0.500 (0.228) & 0.500 (0.228)\\
				1.500 (0.234) & 1.500 (0.173)\\
				-0.500 (0.220) & -0.500 (0.368)
			\end{bmatrix*}
			$ & 
			$
			\begin{bmatrix*}[r]
				0.500 (0.228) & 0.500 (0.228)\\
				1.500 (0.234) & 1.500 (0.173)\\
				-0.500 (0.220) & -0.500 (0.368)
			\end{bmatrix*}
			$ &
			$
			\begin{bmatrix*}[r]
				0.500 (0.228) & 0.500 (0.228)\\
				1.500 (0.234) & 1.500 (0.173)\\
				-0.500 (0.220) & -0.500 (0.368)
			\end{bmatrix*}
			$ &
			$
			\begin{bmatrix*}[r]
				0.500 (0.228) & 0.500 (0.228)\\
				1.500 (0.234) & 1.500 (0.173)\\
				-0.500 (0.220) & -0.500 (0.368)
			\end{bmatrix*}
			$
			\\
			\addlinespace
			$\hat{\mathcal{A}}$ &
			$
			\begin{bmatrix*}[r]
				1.997 (8.602) & -1.994 (7.000)
			\end{bmatrix*}
			$ & 
			$
			\begin{bmatrix*}[r]
				1.997 (8.602) & -1.994 (7.000)
			\end{bmatrix*}
			$ &
			$
			\begin{bmatrix*}[r]
				1.997 (8.602) & -1.994 (7.000)
			\end{bmatrix*}
			$ &
			$
			\begin{bmatrix*}[r]
				1.997 (8.602) & -1.994 (7.001)
			\end{bmatrix*}
			$\\
			\addlinespace
			$\hat{\boldsymbol{\Psi}}$ & 
			$
			\begin{bmatrix*}[r]
				1.000 (5.606) & -0.500 (2.771) \\
				-0.500 (2.771) & 1.001 (2.914)
			\end{bmatrix*}
			$ & 
			$
			\begin{bmatrix*}[r]
				1.000 (5.606) & -0.500 (2.771) \\
				-0.500 (2.771) & 1.001 (2.914)
			\end{bmatrix*}
			$ &
			$
			\begin{bmatrix*}[r]
				1.000 (5.606) & -0.500 (2.771) \\
				-0.500 (2.771) & 1.001 (2.914)
			\end{bmatrix*}
			$ &
			$
			\begin{bmatrix*}[r]
				1.000 (5.606) & -0.500 (2.771) \\
				-0.500 (2.771) & 1.001 (2.914)
			\end{bmatrix*}
			$ \\
			\addlinespace
			$\hat{\rho}_{1}$ &
			$0.900 (0.000)$ &
			$0.900 (0.000)$ &
			$0.900 (0.000)$ &
			$0.900 (0.000)$ \\
			$\hat{\rho}_{2}$ & 
			$0.800 (0.000)$ &
			$0.800 (0.000)$ &
			$0.800 (0.000)$ &
			$0.800 (0.000)$ \\
			\addlinespace
			$\hat{\nu}$ &
			$5.004 (39.331)$ & 
			$5.004 (39.331)$ &
			$5.004 (39.331)$ &
			$5.004 (39.331)$ \\
			\bottomrule
		\end{tabular}
	}
	
	\bigskip
	
	\scalebox{0.75}{
		\begin{tabular}[t]{lcccc}
			\toprule
			$N = 100,000$ & ADECME1 & ADECME2 & ADECME3 & PECME\\
			\midrule
			$\hat{\boldsymbol{\beta}}$ & 
			$
			\begin{bmatrix*}[r]
				0.500 (0.128) & 0.500 (0.171)\\
				1.500 (0.118) & 1.500 (0.091)\\
				-0.500 (0.096) & -0.500 (0.103)
			\end{bmatrix*}
			$ & 
			$
			\begin{bmatrix*}[r]
				0.500 (0.128) & 0.500 (0.171)\\
				1.500 (0.118) & 1.500 (0.091)\\
				-0.500 (0.096) & -0.500 (0.103)
			\end{bmatrix*}
			$ &
			$
			\begin{bmatrix*}[r]
				0.500 (0.128) & 0.500 (0.171)\\
				1.500 (0.118) & 1.500 (0.091)\\
				-0.500 (0.096) & -0.500 (0.103)
			\end{bmatrix*}
			$ &
			$
			\begin{bmatrix*}[r]
				0.500 (0.128) & 0.500 (0.171)\\
				1.500 (0.118) & 1.500 (0.091)\\
				-0.500 (0.096) & -0.500 (0.103)
			\end{bmatrix*}
			$
			\\
			\addlinespace
			$\hat{\mathcal{A}}$ &
			$
			\begin{bmatrix*}[r]
				1.999 (4.484) & -2.000 (2.696)
			\end{bmatrix*}
			$ & 
			$
			\begin{bmatrix*}[r]
				1.999 (4.484) & -2.000 (2.697)
			\end{bmatrix*}
			$ &
			$
			\begin{bmatrix*}[r]
				1.999 (4.484) & -2.000 (2.697)
			\end{bmatrix*}
			$ &
			$
			\begin{bmatrix*}[r]
				1.999 (4.484) & -2.000 (2.697)
			\end{bmatrix*}
			$\\
			\addlinespace
			$\hat{\boldsymbol{\Psi}}$ & 
			$
			\begin{bmatrix*}[r]
				1.000 (1.811) & -0.500 (0.955) \\
				-0.500 (0.955) & 0.999 (2.297)
			\end{bmatrix*}
			$ & 
			$
			\begin{bmatrix*}[r]
				1.000 (1.811) & -0.500 (0.955) \\
				-0.500 (0.955) & 0.999 (2.297)
			\end{bmatrix*}
			$ &
			$
			\begin{bmatrix*}[r]
				1.000 (1.811) & -0.500 (0.955) \\
				-0.500 (0.955) & 0.999 (2.297)
			\end{bmatrix*}
			$ &
			$
			\begin{bmatrix*}[r]
				1.000 (1.811) & -0.500 (0.955) \\
				-0.500 (0.955) & 0.999 (2.297)
			\end{bmatrix*}
			$ \\
			\addlinespace
			$\hat{\rho}_{1}$ &
			$0.900 (0.000)$ &
			$0.900 (0.000)$ &
			$0.900 (0.000)$ &
			$0.900 (0.000)$ \\
			$\hat{\rho}_{2}$ & 
			$0.800 (0.000)$ &
			$0.800 (0.000)$ &
			$0.800 (0.000)$ &
			$0.800 (0.000)$ \\
			\addlinespace
			$\hat{\nu}$ &
			$5.007 (34.615)$ & 
			$5.007 (34.615)$ & 
			$5.007 (34.615)$ & 
			$5.007 (34.615)$ \\
			\bottomrule
		\end{tabular}
	}
	\caption{\label{tab:schem2_combined_pe} Table entries are the point estimates (100 $\times$ standard deviation) of the model parameters obtained from fitting the REGMVST model (via the ADECME and PECME algorithms) to simulated data generated under Scheme II with 10 replicates, where the upper and lower panels represent $N = 25,000$ and $N = 100,000$, respectively. ADECME1, ADECME2, and ADECME3 represent the ADECME algorithm with $\gamma = 0.625$, $0.750$, and $0.875$, respectively.}
\end{table}

\begin{table}[ht]
	\centering
	\begin{tabular}{lcc}
		\toprule
		& $N = 25,000$ & $N = 100,000$\\
		\midrule
		$\hat{\boldsymbol{\beta}}$ & 
		$
		\begin{bmatrix*}[r]
			0.500(0.257) & 0.500(0.397) \\
			1.500(0.305) & 1.500(0.231) \\
			-0.500(0.260) & -0.500(0.330)
		\end{bmatrix*}
		$ & 
		$
		\begin{bmatrix*}[r]
			0.500(0.134) & 0.500(0.177) \\
			1.500(0.119) & 1.500(0.128) \\
			-0.500(0.232) & -0.500(0.169)
		\end{bmatrix*}
		$\\
		\addlinespace
		$\hat{\mathcal{A}}$ & 
		$
		\begin{bmatrix*}[r]
			3.730 (20.049) & -3.727 (15.460)
		\end{bmatrix*}
		$&
		$
		\begin{bmatrix*}[r]
			3.737 (9.469) & -3.736 (10.118)
		\end{bmatrix*}
		$\\
		\addlinespace
		$\hat{\boldsymbol{\Psi}}$ &
		$
		\begin{bmatrix*}[r]
			1.862 (8.334) & -0.932 (5.188) \\
			-0.932 (5.188) & 1.865 (10.222)
		\end{bmatrix*}
		$&
		$
		\begin{bmatrix*}[r]
			1.867 (6.505) & -0.934 (4.122) \\
			-0.934 (4.122) & 1.866 (7.163)
		\end{bmatrix*}
		$\\
		\addlinespace
		$\hat{\rho}_{1}$ & $0.900 (0.000)$ & $0.900 (0.000)$ \\
		$\hat{\rho}_{2}$ & $0.800 (0.000)$ & $0.800 (0.000)$ \\
		\addlinespace
		$\hat{\nu}$ & $6.734(37.000)$ & $6.738(44.675)$\\
		\bottomrule
	\end{tabular}
	\caption{\label{tab:simu3} Table entries are the point estimates (100 $\times$ standard deviation) of model parameters obtained from fitting the REGMVST model via the ADECME to simulated data generated under Scheme III with 10 replicates, for sample sizes of $N \in \{25000, 100000\}$. The ADECME algorithm with $\gamma = 0.875$ was used to calculate the MLE.}
\end{table}

\subsection{Scheme III}\label{sec: scheme_3}

In this final scheme, our objective is to showcase the robustness of the REGMVST model. Instead of generating noise from the MVST distribution, we utilize a matrix variate generalized hyperbolic distribution \citepalias{gallaugher2019three}, with $\lambda = \omega = 1$. The parameters $\boldsymbol{\beta}$, $\mathbf{A}_{i} = \boldsymbol{1}_{i} \mathcal{A}$, $\boldsymbol{\Psi}$, and $\boldsymbol{\Sigma}_{i}$ remain consistent with Schemes I and II. Our aim is to investigate the performance of the REGMVST model under model misspecification with large sample sizes of $N = 25,000$ and $N = 100,000$. We summarize the inference results from the REGMVST model in Table \ref{tab:simu3}. It is noteworthy that, even with the mis-specified distributional assumption, the REGMVST model still yields point estimations of $\boldsymbol{\beta}$, $\rho_{1}$, and $\rho_{2}$ with an average absolute bias of $0$ when rounded to 3 decimal places. The so-called “correct" estimation values of the skewness parameters $\mathcal{A}$ and column covariance matrix $\boldsymbol{\Psi}$ are unknown for our proposed model, as data were generated from a mis-specified distribution rather than the MVST distribution.

\section{Application: Periodontal Disease EHR Data} \label{sec:real_data_application}

In this section, we illustrate the potential of our proposed ADECME algorithm via application to a longitudinal periodontal disease EHR dataset derived from the HealthPartners Institute of Minnesota, between 2007-2014. During a subject's clinic visit, the clinical attachment level (CAL) and pocket depth (PD) are two important biomarkers (measured in millimeters, mm) assessed by dental hygienists using a periodontal probe at pre-specified tooth-sites to monitor periodontal disease progression \citep{bandyopadhyay2010linear}. Because subjects undergo repeated examinations/clinic visits over time, the response at each visit is treated as a bivariate vector comprising the subject-level mean pocket depth (PD) and mean clinical attachment level (CAL), each averaged across all tooth sites in the mouth. In addition, a variety of covariates, such as gender, race, standardized age (subtracting the mean and dividing by the standard deviation), diabetes status, smoking status, brushing and flossing habits, and insurance status were collected from the EHR. 

\begin{figure}[ht]
	\centering
	\includegraphics[width = 0.8\textwidth]{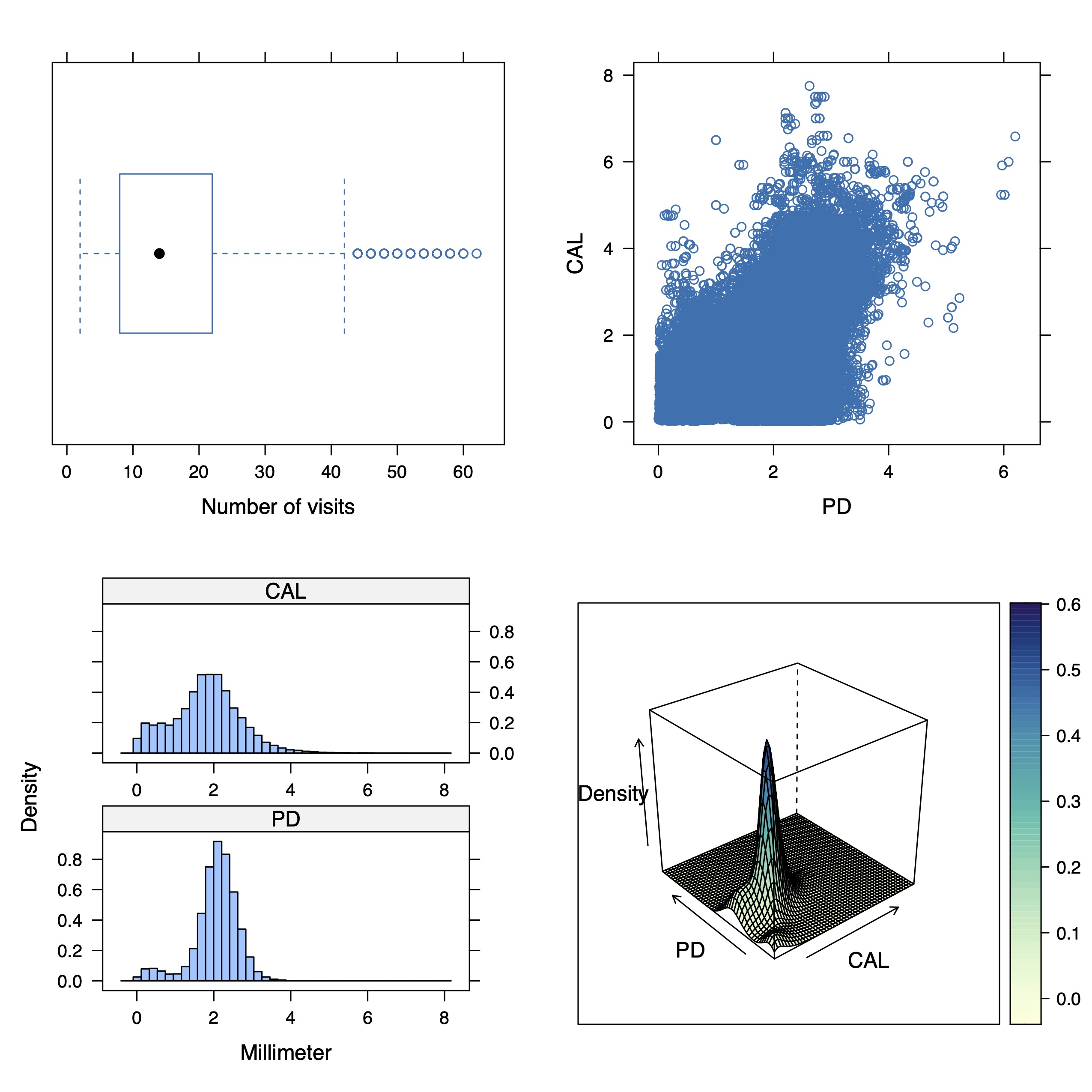}	
	\caption{\label{fig:real_data_all_figures} Exploratory analysis of the Periodontal disease EHR data. The top left and right panels present the box-plot of the number of visits per subject, and the scatterplot of the PD and CAL responses, respectively. The bottom left and right panels depict the density histograms of PD and CAL, and the 2D kernel density contour plot of PD versus CAL, respectively.}
\end{figure}

The data exhibits various features, which makes the the REGMVST model particularly amenable. First, the number and spacing of the CAL and PD measurements from each subject across time appears highly irregular (varying number of measurements $n_i$ per subject), as reflected in the non-uniform row dimension of $\mathbf{Y}_i$ -- our matrix-variate response (see, Figure \ref{fig:real_data_all_figures}, top-left panel), while the temporal effect between measurements corresponds to the DEC structure in $\boldsymbol{\Sigma}_i$. Second, CAL and PD show a strong correlation (Pearson's coefficient = 0.55, Figure \ref{fig:real_data_all_figures}, top-right panel), which is accounted for by the row covariance matrix $\boldsymbol{\Psi}$. Finally, both biomarkers exhibit heavy tails, with most observations centered around 2 mm, a notable concentration around 0 mm, with outliers observed near 6-8 mm (Figure \ref{fig:real_data_all_figures}, bottom-left panel). As shown in the 2D kernel-density contour plot of PD and CAL (Figure \ref{fig:real_data_all_figures}, bottom-right panel), the joint distribution exhibits both asymmetry and heavy tails. These features naturally motivate an MVST error model, where the skewness parameters $\mathcal{A}$ characterize the asymmetry, while the degrees-of-freedom parameter $\nu$ effectively captures the heavy-tailed behavior.

\begin{table}[ht]
	\centering
	\scalebox{1.00}{
		\begin{tabular}[t]{lcccc}
			\toprule
			\multicolumn{1}{c}{ } & \multicolumn{2}{c}{ADECME} & \multicolumn{2}{c}{PECME} \\
			\cmidrule(l{3pt}r{3pt}){2-3} \cmidrule(l{3pt}r{3pt}){4-5}
			Covariate & CAL & PD & CAL & PD\\
			\midrule
			Intercept & \phantom{-}1.913 (\phantom{-}1.886, \phantom{-}1.941) & \phantom{-}2.106 (\phantom{-}2.089, \phantom{-}2.125) & \phantom{-}1.913 (\phantom{-}1.886, \phantom{-}1.941) & \phantom{-}2.106 (\phantom{-}2.089, \phantom{-}2.125)\\
			Male & \phantom{-}0.210 (\phantom{-}0.191, \phantom{-}0.227) & \phantom{-}0.151 (\phantom{-}0.139, \phantom{-}0.163) & \phantom{-}0.210 (\phantom{-}0.191, \phantom{-}0.227) & \phantom{-}0.151 (\phantom{-}0.139, \phantom{-}0.163)\\
			Race: black & \phantom{-}0.103 (\phantom{-}0.059, \phantom{-}0.139) & \phantom{-}0.173 (\phantom{-}0.142, \phantom{-}0.205) & \phantom{-}0.103 (\phantom{-}0.059, \phantom{-}0.139) & \phantom{-}0.173 (\phantom{-}0.142, \phantom{-}0.205)\\
			Race: white & -0.133 (-0.165, -0.112) & -0.097 (-0.115, -0.082) & -0.133 (-0.165, -0.112) & -0.097 (-0.115, -0.082)\\
			Standardized age & \phantom{-}0.129 (\phantom{-}0.123, \phantom{-}0.134) & \phantom{-}0.040 (\phantom{-}0.037, \phantom{-}0.044) & \phantom{-}0.129 (\phantom{-}0.123, \phantom{-}0.134) & \phantom{-}0.040 (\phantom{-}0.037, \phantom{-}0.044)\\
			\addlinespace
			Diabetes & -0.001 (-0.006, \phantom{-}0.005) & \phantom{-}0.001 (-0.003, \phantom{-}0.005) & -0.001 (-0.006, \phantom{-}0.005) & \phantom{-}0.001 (-0.003, \phantom{-}0.005)\\
			Smoker & \phantom{-}0.018 (\phantom{-}0.013, \phantom{-}0.022) & \phantom{-}0.013 (\phantom{-}0.011, \phantom{-}0.016) & \phantom{-}0.018 (\phantom{-}0.013, \phantom{-}0.022) & \phantom{-}0.013 (\phantom{-}0.011, \phantom{-}0.016)\\
			Daily brushing & -0.004 (-0.007, -0.001) & \phantom{-}0.007 (\phantom{-}0.005, \phantom{-}0.010) & -0.004 (-0.007, -0.001) & \phantom{-}0.007 (\phantom{-}0.005, \phantom{-}0.010)\\
			Daily flossing & \phantom{-}0.009 (\phantom{-}0.005, \phantom{-}0.012) & -0.006 (-0.009, -0.004) & \phantom{-}0.009 (\phantom{-}0.005, \phantom{-}0.012) & -0.006 (-0.009, -0.004)\\
			Insurance & -0.009 (-0.013, -0.004) & -0.003 (-0.006, \phantom{-}0.001) & -0.009 (-0.013, -0.004) & -0.003 (-0.006, \phantom{-}0.001)\\
			\bottomrule
			\vspace*{0.1 in}
		\end{tabular}
	}
	
	\scalebox{1.00}{
		\begin{tabular}[t]{lcc}
			\toprule
			Parameter & ADECME & RPECME\\
			\midrule
			$A_1$ & -0.009 (-0.010, -0.009) & -0.009 (-0.010, -0.009)\\
			$A_2$ & -0.002 (-0.003, -0.002) & -0.002 (-0.003, -0.002)\\
			$\Psi_{1,1}$ & \phantom{-}0.044 (\phantom{-}0.043, \phantom{-}0.046) & \phantom{-}0.044 (\phantom{-}0.043, \phantom{-}0.046)\\
			$\Psi_{1,2}$ & \phantom{-}0.016 (\phantom{-}0.022, \phantom{-}0.023) & \phantom{-}0.016 (\phantom{-}0.022, \phantom{-}0.023)\\
			$\Psi_{2,2}$ & \phantom{-}0.023 (\phantom{-}0.016, \phantom{-}0.017) & \phantom{-}0.023 (\phantom{-}0.016, \phantom{-}0.017)\\
			\addlinespace
			$\rho_1$ & \phantom{-}0.900 (\phantom{-}0.900, \phantom{-}0.900) & \phantom{-}0.900 (\phantom{-}0.900, \phantom{-}0.900)\\
			$\rho_2$ & \phantom{-}0.100 (\phantom{-}0.100, \phantom{-}0.100) & \phantom{-}0.100 (\phantom{-}0.100, \phantom{-}0.100)\\
			$\nu$ & \phantom{-}1.069 (\phantom{-}1.050, \phantom{-}1.085) & \phantom{-}1.069 (\phantom{-}1.050, \phantom{-}1.085)\\
			\bottomrule
		\end{tabular}
	}
	\caption{\label{tab:real_data_application_inference} Table entries are the point estimates of the mean and associated 90\% confidence intervals (in parentheses) of fixed-effects co-efficients and other parameters, obtained from fitting the REGMVST model to the periodontal disease EHR data, using the ADECME and PECME algorithms. The reference levels of the fixed-effects parameters are Gender: Female, Race: Other, Diabetes: No, Smoker: No, Brushing: less than daily, Flossing: less than daily, Insurance: No.}
\end{table}

In this empirical application, we aim to demonstrate the real-world applicability of our proposed regression model, and to underscore the effectiveness and practical advantages of the ADECME algorithm. It is noteworthy that the number of subjects in this study is $24,416$, which is quite large. For our REGMVST formulation, CAL and PD were treated as matrix-variate responses, with the subject-level observation times incorporated into the DEC structure. To verify that ADECME and PECME produce identical MLE and $90\%$ confidence intervals, we utilized ADECME with $\gamma = 0.875$ and PECME. We employed a classic nonparametric bootstrap method, resampled at the subject level, to construct confidence intervals for all parameters of interest. Specifically, for each bootstrap iteration, we randomly sampled subjects with replacement from the original dataset to construct a bootstrap sample, from which we obtained a point estimate. We repeated this procedure 100 times to obtain 100 point estimates of all parameters of interest, from which we constructed $90\%$ quantile-based confidence intervals. Inference, in terms of mean and 90\% confidence intervals of the model parameters, are presented in Table \ref{tab:real_data_application_inference}.

Remarkably, we observed that the ADECME algorithm with $\gamma = 0.875$ and the PECME algorithm yield \emph{exactly same} point estimates and confidence intervals, when rounded to 3 decimal places. We observe that younger subjects exhibit better periodontal health than older subjects, and that non-smokers tend to have better periodontal health than smokers -- findings which align to those reported in previous studies \citep{Borojevic2012,Clark2021}. The model also indicates that male subjects have higher CAL and PD values than females and that racial disparities exist, with Black subjects showing higher values and White subjects showing lower values compared to other races. The results for oral hygiene covariates were mixed, and discordant. Daily brushing was associated with a statistically significant decrease in CAL but a significant increase in PD. Conversely, daily flossing was associated with a significant increase in CAL but a significant decrease in PD. This pattern is biologically not implausible, since CAL is a cumulative historical measure of periodontal destruction that integrates both gingival recession and pocketing, whereas PD primarily captures the current inflammatory status of the gum, and hence do not represent identical disease progression processes. Flossing-induced interproximal injury/clefts and bone loss, and brushing associated recession, leading to higher CAL, are not unprecedented \citep{walters2003periodontal, rajapakse2007does}. However, these associations should be interpreted cautiously, in light of self-reported behaviors (brushing and flossing), potential indication/reverse-causation bias (people with aggressive periodontal disease may brush/floss more), and residual confounding in observational EHR data.

For insurance status, subjects with coverage was associated with a statistically significant decrease in CAL, while its effect on PD was not statistically significant. Furthermore, both estimated skewness parameters, $A_{1}$ and $A_{2}$ are negative, and their associated confidence intervals do not include zero. This is supported by the exploratory data analysis steps, which showed a notable concentration of measurements close to 0 mm. Moreover, the estimated degree of freedom is approximately 1.07, indicating very heavy-tailed features and confirming the presence of the few larger outliers near 6-8 mm, as observed in the exploratory Figure \ref{fig:real_data_all_figures}. The estimated correlation parameter $\rho_1 = 0.9$ (within the DEC specification) suggests a strong positive autocorrelation, implying strong temporal dependence between the previous and current CAL and PD measures. The parameter $\rho_2 = 0.1$ suggests that irregular visit times of subjects also contribute to the longitudinal association. Furthermore, the positive estimate for $\Psi_{1,2}$, with a credible interval excluding zero, indicates a positive association between the two biomarkers, i.e., higher CAL is associated with higher PD.

\begin{table}[ht]
	\centering
	\begin{tabular}[t]{lrr}
		\toprule
		& ADECME & PECME\\
		\midrule
		TT & 14.608 (2.874) & 22.378 (7.330)\\
		E step time & 14.601 (2.873) & 1.050 (0.224)\\
		DEC & 0.004 (0.001) & 19.739 (6.765)\\
		$\boldsymbol{\Psi}$ & 0.001 (0.000) & 1.258 (0.333)\\
		$\mathcal{A}$ & 0.000 (0.000) & 0.329 (0.073)\\
		$\boldsymbol{\beta}$ & 0.001 (0.000) & 0.001 (0.000)\\
		$\nu$ & 0.001 (0.000) & 0.001 (0.000)\\
		TNI & 144.750 (27.886) & 127.010 (24.579)\\
		\bottomrule\end{tabular}\caption{\label{tab:real_data_application_bootstrap_time} Table entries are the computation time (in minutes), and associated standard deviation (in parentheses) obtained from 100 bootstrap iterations, for various quantities while fitting the REGMVST model to the periodontal disease EHR data using the ADECME and PECME algorithms. Here, TT denotes the average total time, while TNI represents the average total number of iterations. }
\end{table}

As shown in Table \ref{tab:real_data_application_bootstrap_time}, the ADECME algorithm with $\gamma = 0.875$ required, on average, only 65\% of the computational time required by PECME. The most computationally stubborn steps within the ADECME and PECME algorithms are the distributional E step, and the DEC parameters ($\rho_1$, and $\rho_2$) updating, respectively. Furthermore, due to a reduced number of communications and an innovative asynchronous parallel mechanism, the distributional E step in ADECME was, on average, faster per iteration than updating the DEC parameters in PECME. These computational patterns align with those observed in the simulation studies detailed in Section \ref{sec:simulation_study}, although the number of iterations until convergence was slightly higher for ADECME.

\begin{figure}[ht]
	\centering
	\includegraphics[width = \textwidth]{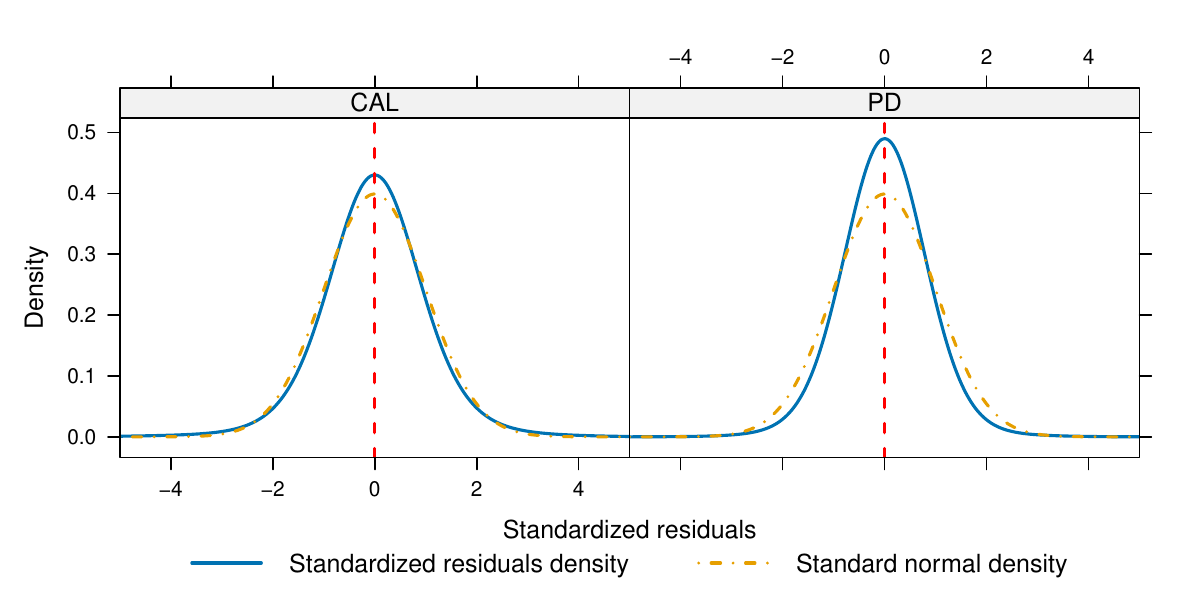}
	\caption{\label{fig:real_data_density_residual} Density plots of standardized residuals (in blue), with the standard normal density overlayed (in dash-dotted yellow), obtained from fitting the REGMVST model to the periodontal disease EHR data via the ADECME algorithm.}
\end{figure}

Utilizing Equation~\eqref{eq:MVST_stochastic_definition} and properties of the MVN distribution, we define $\Sigma^{-1/2}_{i} (\mathbf{Y}_{i} - \mathbf{X}_{i}\boldsymbol{\beta} - W_{i} \mathbf{A}_{i})/\sqrt{W_i}$ as the standardized residuals for subject $i$, where each column independently and identically follows the standard normal distribution. It is important to note that this standardization implies independence across time points, but not across biomarkers. We compute $\Sigma^{-1/2}_{i}$ using the Cholesky decomposition and plug in the point estimates of the parameters, along with the conditional expectation of $W_{i}$ given the data, as specified in Equation~\eqref{eq:W_given_Y}. These standardized residuals facilitate model diagnosis, as illustrated in Figure~\ref{fig:real_data_density_residual}, where we compare their densities to the standard normal distribution. The residuals for both CAL and PD are centered around zero as expected. However, the standardized residuals for CAL approximately follow the standard normal distribution, but exhibit a higher peak near zero, suggesting that the model may be capturing heavy-tailed behavior somewhat more aggressively than necessary. A similar but more pronounced pattern is observed for PD. While these deviations are modest (and maybe linked to the unexpected inference on the brushing and flossing habits), they highlight areas where the model fit could potentially be refined. Nevertheless, while recognizing the inherent limitations of all statistical models, we maintain that the REGMVST model provides clinically relevant insights into periodontal disease progression, and constitutes a methodologically sound approach for modeling the characteristically skewed and heavy-tailed distribution of periodontal biomarkers, as observed in EHRs.

\section{Conclusion} \label{sec:conclusion}

In this work, we introduced REGMVST -- a flexible, matrix-variate, skew-t regression framework tailored for irregular longitudinal data with symmetric/skewed, with/without heavy-tailed responses, and proposed ADECME -- an asynchronous and distributed ECME algorithm that enables scalable inference for large and complex datasets. The REGMVST model synthesizes several desirable features: robust handling of skewness and heavy-tailed errors through the matrix-variate skew-t distribution, explicit modeling of irregular temporal dependence through the DEC structure, and preservation of the natural matrix form of multivariate longitudinal responses, thus allowing the response matrix dimension to vary with subjects. Together, these components provide a principled statistical machinery for analyzing real-world EHR, and related applications where non-Gaussian noise, irregular follow-up, and high dimensionality (in number of subjects) coexist. Our methodological contributions are complemented by the introduction of three tailored ECME-type algorithms (the ECME, PECME, and ADECME), with the ADECME algorithm emerging as the most  efficient in mitigating the computational bottlenecks of the standard ECME and its parallelized variant (PECME), observed in fitting large data. We provide the convergence guarantees of the ADECME, and offer extensive simulation studies and real EHR data illustration to demonstrate the computational advantages and practical viability of the ADECME under large data situations.

Our current study leads to several extensions that merit exploration. The REGMVST model can be further generalized by replacing the MVST distribution with other matrix-variate distributions \citep{gallaugher2019three}, or the skew-normal independent family \citep{bandyopadhyay2012skew}. Incorporating variable selection or regularization into the REGMVST framework may enhance its applicability in high-dimensional (covariate/features) settings, such as in including genomics and multi-omics oral-health research. Moreover, the linearity assumption between the location matrix and the response matrix can be relaxed, leading to more flexible regression forms, including nonlinear or functional predictors. Additional modeling enhancements, such as dynamic correlation structures or hierarchical random-effects formulations, could further expand the model’s applicability to complex longitudinal designs. On the computational side, integrating ADECME with GPU-based acceleration \citep{tagare2010adaptive}, streaming data architectures \citep{cappe2009line}, or federated learning frameworks \citep{dieuleveut2021federated} could substantially expand scalability to national-scale EHR networks. Together, these avenues reinforce the broader potential of matrix-variate skew-t modeling, and asynchronous EM-type algorithms in modern statistical learning and biomedical data science.

\section*{Acknowledgements}
The authors thank the HealthPartners Institute of Minnesota for providing the motivating data and the context of this work. They also acknowledge Dr. Reuben Retnam for assisting in an earlier version of the work. Bandyopadhyay acknowledges partial research support from grants R21DE031879 and R01DE031134 awarded by the United States National Institutes of Health. Srivastava acknowledges partial research support from the National Science Foundation (DMS-1854667 and DMS-2506058). Additionally, the authors express their gratitude to the High-Performance Research Computing core facility at Virginia Commonwealth University.

\section*{Declaration of generative AI in scientific writing}

While preparing this work, the authors used the generative pre-trained transformer models to check grammar. After using this tool/service, the authors reviewed and edited the content as needed and take full responsibility for the publication's content.

\newpage

\begin{appendices}
	
	\section{Proof of Propositions \ref{ad-prop1} and \ref{ad-prop2}}
	
	\subsection{Proof of Proposition \ref{ad-prop1}}
	
	We adapt the proof of Theorems 1 and 2 in \citet{NeaHin98} to our setup. 
	At the end of $t$th iteration of ADECME, $\varthetab^{(t)}$ is the parameter estimate obtained from the distributed CM step. In the distributed E step of this iteration, for $j=1, \ldots, k$,  $\tilde p_{j}^{(t-1)} = \prod_{i=1}^{N_j} p(w_{ji} \mid \Yb_{ji}, \Xb_{ji}, \tb_{ji}, \varthetab^{(t-1)})$ is the conditional density of the missing data $\wb_j = (w_{j1}, \ldots, w_{j N_j})$ given the observed data on subset $j$ if this worker returned its sufficient statistics to the manager. Otherwise, the conditional density of $\wb_j $ given the observed data on subset $j$ is $\tilde p_{j}^{(t_j)} = \prod_{i=1}^{N_j} p(w_{ji} \mid \Yb_{ji}, \Xb_{ji}, \tb_{ji}, \varthetab^{(t_j)})$ for some $t_j < t -1$. If $\Rcal_{t} \subset \{1, \ldots, k\}$ includes the indices of workers who returned their sufficient statistics to the manager in the $t$th iteration, then define $\tilde p^{(t)} = \prod_{j_1 \in \Rcal_{t}} \tilde p_{j_1}^{(t-1)} \prod_{j_0 \in \Rcal_{t}^c} \tilde p_{j_0}^{(t-1)}$, where $p_{j_0}^{(t-1)}$ equals $p_{j_0}^{(t_{j_0})}$ for some $t_{j_0} < t-1$.
	
	The distributed E step in the $(t+1)$th iteration of ADECME computes the conditional expectations of the complete sufficient statistics locally on all the $k$ subsets with $\varthetab = \varthetab^{(t)}$. It ends after the manager has heard from a $\gamma$-fraction of workers. If worker $j$ returned the sufficient statistics, then $\tilde p_{j}^{(t_j)}$ or $\tilde p_{j}^{(t-1)}$ is updated to $\tilde p_{j}^{(t)}$  after setting $\varthetab^{(t-1)}$ or $\varthetab^{(t_j)} $ to $ \varthetab^{(t)}$, otherwise $\tilde p_{j}^{(t_j)}$ or $\tilde p_{j}^{(t-1)}$ remains unchanged. Define $\tilde p^{(t+1)} = \prod_{j_1 \in \Rcal_{t+1}} \tilde p_{j_1}^{(t)} \prod_{j_0 \in \Rcal_{t+1}^c} \tilde p_{j_0}^{(t)}$, where $\Rcal_{t+1} $ includes indices of the workers who returned their sufficient statistics to the manager in the $(t+1)$th iteration and $\tilde p_{j_0}^{(t)}$ equals either $\tilde p_{j_0}^{(t_{j_0})}$ or $\tilde p_{j_0}^{(t-1)}$. 
	Theorem 1 in \citet{NeaHin98} implies that  $\Fcal(\tilde p^{(t)}, \varthetab^{(t)}) \leq \Fcal(\tilde p^{(t+1)}, \varthetab^{(t)})$. 
	
	The distributed CM step in the  $(t+1)$th iteration of ADECME updates $\varthetab^{(t)}$ to $\varthetab^{(t+1)}$. Theorem 1  in \citet{NeaHin98} again implies that $\Fcal(\tilde p^{(t+1)}, \varthetab^{(t)}) \leq \Fcal(\tilde p^{(t+1)}, \varthetab^{(t+1)})$. Using the last inequality from the previous paragraph, at the end of $(t+1)$th iteration of ADECME, $\Fcal(\tilde p^{(t)}, \varthetab^{(t)})$ from the $t$th iteration of ADECME increase to $\Fcal(\tilde p^{(t+1)}, \varthetab^{(t+1)})$ because 
	$\Fcal(\tilde p^{(t)}, \varthetab^{(t)}) \leq \Fcal(\tilde p^{(t+1)}, \varthetab^{(t)}) \leq \Fcal(\tilde p^{(t+1)}, \varthetab^{(t+1)})$; therefore, for every $\gamma$, the ADECME algorithm maintains the monotone ascent of $\Fcal(\widetilde p, \varthetab)$ at every iteration . 
	
	Finally, we have assumed that $\varthetab$ belongs to a compact parameter space such that all the densities are bounded on this space. This implies that the $\{\Fcal(\tilde p^{(t)}, \varthetab^{(t)}) \}$ sequence converges. Theorem 2  in \citet{NeaHin98}  implies that if $(\widehat {\tilde p}, \widehat \varthetab)$ is a fixed point of the $\{\Fcal(\tilde p^{(t)}, \varthetab^{(t)}) \}$ sequence, then  $\hat \ell = \ell(\hat \varthetab)$ is a fixed point of the $\ell(\varthetab^{(t)})$ sequence. 
	
	\subsection{Proof of Proposition \ref{ad-prop2}}
	
	The distributed CM step in Section \ref{sec:ADECME_M_step} implies that the ADECME map $\varthetab^{(t)} \mapsto \varthetab^{(t+1)}$ is closed and continuous. Furthermore, we declare convergence when $\|\varthetab^{(t)} - \varthetab^{(t+1)}\|_{\infty} \leq \epsilon$ for sufficiently small $\epsilon>0$ and  $\|\varthetab^{(t)} - \varthetab^{(t+1)}\|_{\infty}  \to 0$ as $t \to \infty$ because $\Qcal(\varthetab^{(t+1)} \mid \varthetab^{(t)}) - \Qcal(\varthetab^{(t)} \mid \varthetab^{(t)}) \geq c \|\varthetab^{(t+1)} - \varthetab^{(t)}\|$ for a universal constant $c$. The function $\Qcal ( \cdot \mid \cdot)$ in \eqref{eq:regression_Q_function} is continuously differentiable in both arguments. This implies that the $\Qcal(\cdot \mid \cdot)$ function obtained from the distributed E step is also continuously differentiable in both arguments. Assumption A2 implies that the stationary points of $\Thetab$ are also assumed to belong to a compact set. Using these three conditions, Theorem 6 in \citet{wu1983convergence} implies that the $\{\varthetab^{(t)}\}$ sequence either converges to a stationary point or maximizer of $\ell(\varthetab)$.

	\section{Multivariate (Vector Variate) Skew \texorpdfstring{$t$}{t} Distribution}
	
	Assume that $\yb \in \RR^{d \times 1}$ follows a multivariate Skew $t$ distribution with parameters ($\mub, \gammab, \Sigmab, \nu$). Let 
	\begin{align}\label{eq:s-rho}
		s (\yb)= \left[ \{ \nu + \rho(\yb) \} \gammab^\top \Sigmab^{-1} \gammab \right]^{\frac{1}{2}}, \quad
		\rho(\yb) = (\yb - \mub)^\top \Sigmab^{-1} (\yb - \mub).
	\end{align}
	Then, the joint density function of $\yb$ and its log are 
	\begin{align}\label{eq:skewt-dens}
		f(\yb) &= \frac{2^{1 - \frac{\nu + d}{2}}}{\Gamma(\frac{\nu}{2}) (\pi \nu)^{\frac{d}{2}} |\Sigmab|^{\frac{1}{2}} }\frac{ K_{\frac{\nu + d}{2}} \left( s (\yb) \right) e^{(\yb - \mub)^\top \Sigmab^{-1} \gammab} }
		{ s (\yb)^{-\frac{\nu + d}{2}} \left( 1 + \frac{\rho(\yb)} {\nu}\right)^{\frac{\nu + d}{2}}}, \\
		\log f(\yb) &= \left(1 - \frac{\nu + d}{2} \right) \log 2 - \log \Gamma(\frac{\nu}{2}) - \frac{d}{2} \log (\pi \nu) - \frac{1}{2} \log |\Sigmab| + \log K_{\frac{\nu + d}{2}} \left( s (\yb) \right) + \nonumber \\
		&\qquad (\yb - \mub)^\top \Sigmab^{-1} \gammab + \frac{\nu + d}{2} \log s(\yb) - \frac{\nu + d}{2} \log \left( 1 + \frac{\rho(\yb)}{\nu} \right), \nonumber
	\end{align}
	where $K_{\lambda}(x) = \frac{1}{2} \int_{0}^{\infty} y^{\lambda - 1} e^{- \frac{x}{2} (y + y^{-1})} dy$ for $x > 0$ is the modified Bessel function of the third kind; see Proposition 2.4 in \citet{WenAle06} for a derivation of the density using a multivariate normal mean-variance mixture model.

	Our first result obtains an analytic form for the information matrix of $\yb$ with density $f(\yb)$ in \eqref{eq:skewt-dens}. For notational convenience, the partial derivatives are denoted as $\di$.
	\begin{proposition}\label{prop1}
		Let $\ell(\thetab) = \log f(\yb)$ be the log likelihood function of $\thetab$, where  $\thetab = (\mub, \gammab, \Sigmab, \nu) \in \RR^{\frac{d^2 + 5d + 2}{2}}$ and $\yb$ follows a multivariate Skew $t$($\mub, \gammab, \Sigmab, \nu$) distribution. Then, the first derivative of the log likelihood of $\thetab$ and the information matrix of $\yb$ are
		\begin{align}%\label{eq:obs-inf}
			\frac{\di \ell(\thetab)}{\di \thetab} &= \left( \frac{\di \ell(\thetab)}{\di \mub}, \frac{\di \ell(\thetab)}{\di \gammab}, \frac{\di \ell(\thetab)}{\di \vech(\Sigmab)}, \frac{\di \ell(\thetab)}{\di \nu}\right ) \in \RR^{1 \times \frac{d^2 + 5d + 2}{2}},\nonumber\\
			\Ib_{\text{obs}}(\thetab) &= \EE \left(\frac{\di \ell(\thetab)}{\di \thetab^\top} \frac{\di \ell(\thetab)}{\di \thetab}\right), 
		\end{align}
		where the expectation is with respect to the distribution of $\yb$ and $\Ib_{\text{obs}}(\thetab)$ exists if $\nu > 4 $. The analytic forms of the blocks in $\frac{\di \ell(\thetab)}{\di \thetab}$ are as follows:
		\begin{align*}
			\frac{\di \ell(\theta)}{\di \mub} &=  \left\{ c_{\mu}(\yb) (\mub - \yb)^\top - \gammab^\top  \right\} \Sigmab^{-1} \in \RR^{1 \times d}, \\ 
			\frac{\di \ell(\theta)}{\di \gammab} &=  \left\{ c_{\gammab}(\yb) \gammab^\top - (\mub - \yb)^\top  \right\} \Sigmab^{-1} \in \RR^{1 \times d},\\
			\frac{\di \ell(\theta)}{\di \Sigmab} &=  \Sigmab^{-1} \Cb_{\Sigma}(\yb)  \Sigmab^{-1} \in \RR^{d \times d},\\ 
			\frac{\di \ell(\theta)}{\di \vech (\Sigmab)}&= \vv \left\{ \Cb_{\Sigma}(\yb)\right\}^\top (\Sigmab^{-1} \otimes \Sigmab^{-1}) \Db_d \in \RR^{1 \times d(d+1)/2},\\ 
			\frac{\di \ell(\theta)}{\di \nu} &= c_{1 \nu}(\yb) + c_{2 \nu}(\yb) \in \RR,
		\end{align*}
		where
		\begin{align*}
			c_{\mu}(\yb) &= \left\{ \frac{K_{\frac{\nu + d}{2}}' \left( s (\yb) \right)}{K_{\frac{\nu + d}{2}} \left( s (\yb) \right)} + \frac{\nu + d}{2s(\yb)}\right\} \frac{\gammab^\top \Sigmab^{-1} \gammab } {s(\yb)} - \frac{\nu + d}{ \nu + \rho(\yb) },\\ 
			c_{\gamma}(\yb) &= \left\{ \frac{K_{\frac{\nu + d}{2}}' \left( s (\yb) \right)}{K_{\frac{\nu + d}{2}} \left( s (\yb) \right)} + \frac{\nu + d}{2 s(\yb)}\right\} \frac{\nu + \rho(\yb)} {s(\yb)},\\
			\Cb_{\Sigma}(\yb) &=   c_{\mu \mu}(\yb) (\mub - \yb)(\mub - \yb)^\top   + c_{\gamma \gamma}(\yb) \gammab \gammab^\top + \frac{1}{2}\left\{\gammab (\mub - \yb)^\top +  (\mub - \yb) \gammab^\top \right\} - \frac{1}{2}\Sigmab,  \\
			c_{\mu \mu}(\yb) &= \frac{\nu + d}{ 2(\nu + \rho(\yb)) } - \left[ \frac{K_{\frac{\nu + d}{2}}' \left( s (\yb) \right)}{K_{\frac{\nu + d}{2}} \left( s (\yb) \right)} + \frac{\nu + d}{2s(\yb)}\right] \frac{\gammab^\top \Sigmab^{-1} \gammab } {2s(\yb)} ,\\
			c_{\gamma \gamma}(\yb) &=  - \left[ \frac{K_{\frac{\nu + d}{2}}' \left( s (\yb) \right)}{K_{\frac{\nu + d}{2}} \left( s (\yb) \right)} + \frac{\nu + d}{2s(\yb)}\right] \frac{\nu + \rho(\yb)} {2s(\yb)},\\
			c_{1 \nu}(\yb) &= -\frac{1}{2} \left\{\nu \log 2 + \psi\left(\frac{\nu}{2}\right) + \frac{d}{\nu} - \frac{(\nu + d) \rho(\yb)}{\nu(\nu + \rho(\yb))} + \log \left( 1 + \frac{\rho(\yb)}{\nu} \right) - \log s(\yb)\right\},\\
			c_{2 \nu}(\yb) &= \left\{ \frac{\partial K_{\frac{\nu + d}{2}} \left( s (\yb) \right)}{K_{\frac{\nu + d}{2}} \left( s (\yb) \right)} + \frac{\nu + d}{2 s(\yb)}\right\} \frac{\gammab^\top \Sigmab^{-1}\gammab}{2 s(\yb)},
		\end{align*}
		$\vv$, $\vech$ are vectorization and symmetric vectorizations of a (symmetric) matrix, $\Db_d$ is the duplication matrix that satisfies $\vv(\di \Sigmab) = \Db_d \vech(\di \Sigmab)$, $K'_{\lambda} (x)= \frac{\di K_{\lambda}(x)}{\di x}$, $\psi(\cdot)$ is the digamma function, and $\partial K_{\lambda} (x)= \frac{\di K_{\lambda}(x)}{\di \lambda}$. Similarly, if $\nu^* > 4$ and 
		\begin{align*}
			\Vb^*_{c_{\mu} y} &= \EE 
			\left[  \{c_{\mu}(\yb)\}^2 (\yb - \mub)  (\yb - \mub)^\top \right],\quad
			\cb^*_{c_{\mu}y} = \EE 
			\left\{  c_{\mu}(\yb) (\yb - \mub)   \right\}, \nonumber\\
			v^*_{c_{\gamma} } &= \EE 
			\left[  \{c_{\gamma}(\yb)\}^2 \right],\quad
			\cb^*_{c_{\gamma}y} = \EE 
			\left\{  c_{\gamma}(\yb) (\yb - \mub)   \right\}, \quad
			\Vb^*_{y} = 
			\EE\{(\yb - \mub)  (\yb - \mub)^\top \}, \nonumber\\
			\cb_{\Sigma}(\yb) &= \vv \{\Cb_{\Sigma}(\yb) \}, \quad
			\Vb^*_{c_{\Sigma} } = \EE  \{ \cb_{\Sigma}(\yb) \cb_{\Sigma}(\yb)^\top\}.
		\end{align*}
		Then, \eqref{eq:obs-inf} implies that the four diagonal blocks in $\Ib_{\text{obs}}({\thetab})$ for the four parameter blocks are 
		\begin{align*}
			[\Ib_{\text{obs}}(\thetab)]_{\mu \mu} &=  \Sigmab^{-1}(\Vb^*_{c_{\mu} y} + 2 \cb^*_{c_{\mu}y} \gammab^\top + \gammab \gammab^\top) \Sigmab^{-1},\\
			[\Ib_{\text{obs}}(\thetab)]_{\gamma \gamma} &= \Sigmab^{-1}(\Vb^*_{ y} + 2 \cb^*_{c_{\gamma}y} \gammab^\top + v^*_{c_{\gamma} }  \gammab \gammab^\top) \Sigmab^{-1},\\
			[\Ib_{\text{obs}}(\thetab)]_{\vech \Sigma \vech \Sigma} &= \Db_d^\top  (\Sigmab^{-1} \otimes \Sigmab^{-1})  \Vb^*_{c_{\Sigma} }  (\Sigmab^{-1} \otimes \Sigmab^{-1}) \Db_d,\\
			[\Ib_{\text{obs}}(\thetab)]_{\nu \nu} &= \EE(c_{1 \nu}^2) + 
			\EE(c_{2 \nu}^2) + 2 \EE(c_{1 \nu} c_{2 \nu}).
		\end{align*}
		
	\end{proposition}

	\begin{proof}
		We find the differentials of $\rho(\yb)$ and $s(\yb)$. Using the definitions of $\rho(\yb)$ and $s(\yb)$ in \eqref{eq:s-rho},
		\begin{align*}
			\di \rho(\yb) &= \tr \left\{ 2 (\mub - \yb)^\top \Sigmab^{-1} \di \mub - \Sigmab^{-1}(\mub - \yb) (\mub - \yb)^\top \Sigmab^{-1} \di \Sigmab \right\},\\
			\frac{\di \rho(\yb)}{\di \mub} &= 2 (\mub - \yb)^\top \Sigmab^{-1}, \quad  
			\frac{\di \rho(\yb)}{\di \Sigmab} = - \Sigmab^{-1}(\mub - \yb) (\mub - \yb)^\top \Sigmab^{-1} , \\
			s (\yb)^2 & = \{ \nu + \rho(\yb) \} \gammab^\top \Sigmab^{-1} \gammab , \quad  2 s  \di s = \di \rho(\yb)  \gammab^\top \Sigmab^{-1} \gammab + \{ \nu + \rho(\yb) \} \di  (\gammab^\top \Sigmab^{-1} \gammab),
		\end{align*}
		where we have suppressed the dependence of $s$ on $\yb$ for notational simplicity. The first differential of $s(\yb)$ depends on $\di \rho(\yb)$, which is defined in the previous display, and
		\begin{align*}
			\di  (\gammab^\top \Sigmab^{-1} \gammab) &= \tr \left( 2 \gammab^\top \Sigmab^{-1} \di \gammab - \Sigmab^{-1} \gammab \gammab^\top \Sigmab^{-1} \di \Sigmab  \right),\\
			\frac{\di \gammab^\top \Sigmab^{-1} \gammab}{\di \gammab}  &= 2 \gammab^\top \Sigmab^{-1}, \quad 
			\frac{\di \gammab^\top \Sigmab^{-1} \gammab}{\di \Sigmab} = -\Sigmab^{-1} \gammab \gammab^\top \Sigmab^{-1} , 
		\end{align*}
		and other derivatives are zero. The previous two displays imply that 
		\begin{align*}
			\di s  &=  \tr \left\{ 2 (\mub - \yb)^\top \Sigmab^{-1} \di \mub - \Sigmab^{-1}(\mub - \yb) (\mub - \yb)^\top \Sigmab^{-1} \di \Sigmab \right\} \gammab^\top \Sigmab^{-1} \gammab / (2s) + \\
			&\, \quad \tr \left( 2 \gammab^\top \Sigmab^{-1} \di \gammab - \Sigmab^{-1} \gammab \gammab^\top \Sigmab^{-1} \di \Sigmab  \right) \{ \nu + \rho(\yb) \} / (2s),\\
			&=  \tr \left\{ \frac{\gammab^\top \Sigmab^{-1} \gammab } {s}  (\mub - \yb)^\top \Sigmab^{-1} \di \mub \right\} +  \tr \left( \frac{\nu + \rho(\yb)} {s} \gammab^\top \Sigmab^{-1} \di \gammab  \right\} - \\
			&\, \quad  \tr \left[ \Sigmab^{-1} \left\{ \frac{\gammab^\top \Sigmab^{-1} \gammab } {2s}    (\mub - \yb) (\mub - \yb)^\top  +\frac{\nu + \rho(\yb)} {2s} \gammab \gammab^\top  \right\} \Sigmab^{-1} \di \Sigmab\right] ,\\
			\frac{\di s (\yb)}{\di \mub} &=\frac{\gammab^\top \Sigmab^{-1} \gammab } {s(\yb)} (\mub - \yb)^\top \Sigmab^{-1} , \quad
			\frac{\di s (\yb)}{\di \gammab} = \frac{\nu + \rho(\yb)} {s(\yb)} \gammab^\top \Sigmab^{-1},\\
			\frac{\di s (\yb)}{\di \Sigmab} &=-\Sigmab^{-1} \left\{ \frac{\gammab^\top \Sigmab^{-1} \gammab } {2s(\yb)}    (\mub - \yb) (\mub - \yb)^\top  +\frac{\nu + \rho(\yb)} {2s(\yb)} \gammab \gammab^\top  \right\} \Sigmab^{-1}.
		\end{align*}
		
		Consider the log likelihood of $\thetab$ based on \eqref{eq:skewt-dens}. Specifically, $\ell(\thetab) = \log f(\yb) $ and the analytic form of $\frac{\di \ell(\thetab)}{\di \nu}$ follows from known results. For the non-scalar parameters, the first differential of $\ell(\thetab)$ is
		\begin{align*}
			\ell(\thetab) &=  \left(1 - \frac{\nu + d}{2} \right) \log 2 - \log \Gamma(\frac{\nu}{2}) - \frac{d}{2} \log (\pi \nu) - \frac{1}{2} \log |\Sigmab| + \log K_{\frac{\nu + d}{2}} \left( s (\yb) \right) + \nonumber \\
			&\qquad (\yb - \mub)^\top \Sigmab^{-1} \gammab + \frac{\nu + d}{2} \log s(\yb) - \frac{\nu + d}{2} \log \left( 1 + \frac{\rho(\yb)}{\nu} \right), \nonumber\\
			\di \ell(\thetab) &= - \frac{1}{2} \tr (\Sigmab^{-1} \di \Sigmab) + \frac{K_{\frac{\nu + d}{2}}' \left( s (\yb) \right)}{K_{\frac{\nu + d}{2}} \left( s (\yb) \right)}  \di s(\yb) +\\
			&\; \quad \tr (- \di \mub^\top \Sigmab^{-1} \gammab + (\mub - \yb)^\top \Sigmab^{-1} \di \Sigmab \Sigmab^{-1} \gammab - (\mub - \yb)^\top \Sigmab^{-1} \di \gammab) + \\
			&\; \quad \frac{\nu + d}{2 s(\yb)} \di s(\yb) - \frac{\nu + d}{2 \{\nu + \rho(\yb) \}} \di \rho(\yb)\\
			&= - \frac{1}{2} \tr (\Sigmab^{-1} \di \Sigmab) + \left\{ \frac{K_{\frac{\nu + d}{2}}' \left( s (\yb) \right)}{K_{\frac{\nu + d}{2}} \left( s (\yb) \right)} + \frac{\nu + d}{2 s(\yb)} \right\}  \di s(\yb)  - \frac{\nu + d}{2 \{\nu + \rho(\yb) \}} \di \rho(\yb)\\
			&\quad + \tr (- \gammab^\top \Sigmab^{-1}  \di \mub + \Sigmab^{-1} \gammab  (\mub - \yb)^\top \Sigmab^{-1} \di \Sigmab - (\mub - \yb)^\top \Sigmab^{-1} \di \gammab).
		\end{align*}
		
		Using the first differential of $\ell(\thetab)$, 
		\begin{align*}
			\frac{\di \ell(\thetab)} {\di \mub} &=   
			\left\{ \frac{K_{\frac{\nu + d}{2}}' \left( s (\yb) \right)}{K_{\frac{\nu + d}{2}} \left( s (\yb) \right)} + \frac{\nu + d}{2 s(\yb)} \right\}  \frac{\di s(\yb)}{\di \mub}  - \frac{\nu + d}{2 \{\nu + \rho(\yb) \}} \frac{\di \rho(\yb)} {\di \mub}- \gammab^\top \Sigmab^{-1} \\
			&= \left\{ \frac{K_{\frac{\nu + d}{2}}' \left( s (\yb) \right)}{K_{\frac{\nu + d}{2}} \left( s (\yb) \right)} + \frac{\nu + d}{2 s(\yb)} \right\}  \frac{\gammab^\top \Sigmab^{-1} \gammab } {s(\yb)} (\mub - \yb)^\top \Sigmab^{-1}  - \\
			&\qquad \frac{\nu + d}{2 \{\nu + \rho(\yb) \}} 2 (\mub - \yb)^\top \Sigmab^{-1}- \gammab^\top \Sigmab^{-1} \\
			&= \left[\left\{ \frac{K_{\frac{\nu + d}{2}}' \left( s (\yb) \right)}{K_{\frac{\nu + d}{2}} \left( s (\yb) \right)} + \frac{\nu + d}{2 s(\yb)} \right\}  \frac{\gammab^\top \Sigmab^{-1} \gammab } {s(\yb)} - \frac{\nu + d}{\{\nu + \rho(\yb) \}}  \right](\mub - \yb)^\top \Sigmab^{-1} - \gammab^\top \Sigmab^{-1}\\
			&\equiv \left\{ c_{\mu}(\yb) (\mub - \yb)^\top - \gammab^\top \right\} \Sigmab^{-1}.
		\end{align*}
		Similarly,  noting that $\frac{\di \rho(\yb)}{\di \gammab} = \zero$,  $\di \ell(\thetab)$  implies that
		\begin{align*}
			\frac{\di \ell(\thetab)} {\di \gammab} &=   
			\left\{ \frac{K_{\frac{\nu + d}{2}}' \left( s (\yb) \right)}{K_{\frac{\nu + d}{2}} \left( s (\yb) \right)} + \frac{\nu + d}{2 s(\yb)} \right\}  \frac{\di s(\yb)}{\di \gammab}  - (\mub - \yb)^\top \Sigmab^{-1} \\  
			&=   
			\left\{ \frac{K_{\frac{\nu + d}{2}}' \left( s (\yb) \right)}{K_{\frac{\nu + d}{2}} \left( s (\yb) \right)} + \frac{\nu + d}{2 s(\yb)} \right\}  \frac{\nu + \rho(\yb)} {s(\yb)} \gammab^\top \Sigmab^{-1}  - (\mub - \yb)^\top \Sigmab^{-1} \\  
			&\equiv \left\{ c_{\gamma}(\yb) \gammab^\top - (\mub - \yb)^\top \right\} \Sigmab^{-1}.
		\end{align*}
		
		Finally, the derivative with respect to $\vech(\Sigmab)$ follows by noting that
		\begin{align*}
			\frac{\di \ell(\thetab)}{\di \Sigmab}  
			&= - \frac{1}{2} \Sigmab^{-1}  + \left\{ \frac{K_{\frac{\nu + d}{2}}' \left( s (\yb) \right)}{K_{\frac{\nu + d}{2}} \left( s (\yb) \right)} + \frac{\nu + d}{2 s(\yb)} \right\}  \frac{\di s(\yb)}{\di \Sigmab}  - \frac{\nu + d}{2 \{\nu + \rho(\yb) \}} \frac{\di \rho(\yb)} {\di \Sigmab}\\
			&\quad + \Sigmab^{-1} \gammab  (\mub - \yb)^\top \Sigmab^{-1}\\
			&= - \frac{1}{2} \Sigmab^{-1}  \\
			&\quad - \left\{ \frac{K_{\frac{\nu + d}{2}}' \left( s (\yb) \right)}{K_{\frac{\nu + d}{2}} \left( s (\yb) \right)} + \frac{\nu + d}{2 s(\yb)} \right\} 
			\Sigmab^{-1} \left\{ \frac{\gammab^\top \Sigmab^{-1} \gammab } {2s(\yb)}    (\mub - \yb) (\mub - \yb)^\top  +\frac{\nu + \rho(\yb)} {2s(\yb)} \gammab \gammab^\top  \right\} \Sigmab^{-1}\\
			&\quad + \frac{\nu + d}{2 \{\nu + \rho(\yb) \}} \Sigmab^{-1}(\mub - \yb) (\mub - \yb)^\top \Sigmab^{-1} + \Sigmab^{-1} \frac{1}{2}\{(\mub - \yb) \gammab  ^\top  +  \gammab  (\mub - \yb)^\top \} \Sigmab^{-1}\\
			&\equiv \Sigmab^{-1} \Cb_{\Sigma} \Sigmab^{-1},\\
			\Cb_{\Sigma} &=\left[\frac{\nu + d}{2 \{\nu + \rho(\yb) \}} - \left\{ \frac{K_{\frac{\nu + d}{2}}' \left( s (\yb) \right)}{K_{\frac{\nu + d}{2}} \left( s (\yb) \right)} + \frac{\nu + d}{2 s(\yb)} \right\} \frac{\gammab^\top \Sigmab^{-1} \gammab } {2s(\yb)}  \right] (\mub - \yb) (\mub - \yb)^\top \\
			&\quad - 
			\left\{ \frac{K_{\frac{\nu + d}{2}}' \left( s (\yb) \right)}{K_{\frac{\nu + d}{2}} \left( s (\yb) \right)} + \frac{\nu + d}{2 s(\yb)} \right\} \frac{\nu + \rho(\yb)} {2s(\yb)} \gammab \gammab^\top 
			+ \frac{1}{2}\{(\mub - \yb) \gammab  ^\top  +  \gammab  (\mub - \yb)^\top \} - \frac{1}{2} \Sigmab \\
			&\equiv c_{\mu \mu}(\yb) (\mub - \yb) (\mub - \yb)^\top + c_{\gamma \gamma}(\yb) \gammab \gammab^\top + \frac{1}{2}\{(\mub - \yb) \gammab  ^\top  +  \gammab  (\mub - \yb)^\top \}- \frac{1}{2} \Sigmab.
		\end{align*}
		The last equation is written as
		\begin{align*}
			\di \ell(\thetab) &= \tr(\Sigmab^{-1} \Cb_{\Sigma} \Sigmab^{-1} \di \Sigmab)= \vv(\Sigmab^{-1} \Cb_{\Sigma} \Sigmab^{-1})^\top \vv(d \Sigmab) \\
			&= \{(\Sigmab^{-1} \otimes \Sigmab^{-1}) \vv(\Cb_{\Sigma}) \}^\top \vv(d \Sigmab) = \vv(\Cb_{\Sigma})^{\top} (\Sigmab^{-1} \otimes \Sigmab^{-1})\vv(d \Sigmab) \\
			&= 
			\vv(\Cb_{\Sigma})^{\top} (\Sigmab^{-1} \otimes \Sigmab^{-1}) \Db_d \vech(d \Sigmab),
		\end{align*}
		where $\Db_d$ is the duplication matrix that satisfies $\vv(d \Sigmab) = \Db_d \vech(d \Sigmab)$ \citep{MagNeu19}. The last display implies that 
		$$\frac{\di \ell(\theta)}{\di \vech (\Sigmab)}= \vv \left\{ \Cb_{\Sigma}(\yb)\right\}^\top (\Sigmab^{-1} \otimes \Sigmab^{-1}) \Db_d \equiv \cb_{\Sigmab}(\yb)^\top (\Sigmab^{-1} \otimes \Sigmab^{-1}) \Db_d .$$ 
		
		The form of the information matrix implies the forms of the diagonal blocks for $\mub$, $\gammab$, $\vech(\Sigmab)$, and $\nu$. Define
		\begin{align}\label{eq:vars}
			\Vb^*_{c_{\mu} y} &= \EE 
			\left[  \{c_{\mu}(\yb)\}^2 (\yb - \mub)  (\yb - \mub)^\top \right],\quad
			\cb^*_{c_{\mu}y} = \EE 
			\left\{  c_{\mu}(\yb) (\yb - \mub)   \right\}, \nonumber\\
			v^*_{c_{\gamma} } &= \EE 
			\left[  \{c_{\gamma}(\yb)\}^2 \right],\quad
			\cb^*_{c_{\gamma}y} = \EE 
			\left\{  c_{\gamma}(\yb) (\yb - \mub)   \right\}, \quad
			\Vb^*_{y} = 
			\EE\{(\yb - \mub)  (\yb - \mub)^\top \}, \nonumber\\
			\cb_{\Sigma}(\yb) &= \vv \{\Cb_{\Sigma}(\yb) \}, \quad
			\Vb^*_{c_{\Sigma} } = \EE  \{ \cb_{\Sigma}(\yb) \cb_{\Sigma}(\yb)^\top\},
		\end{align}
		where all the expectations are with respect to the distribution of $\yb$, Skew $t$($\mub^*, \gammab^*, \Sigmab^*, \nu^*$). Applying the Cauchy-Schwartz inequality implies that all expectations in \eqref{eq:vars} exist given $\nu^* > 4$, when the covariance matrix of $\yb$ exists. When $\nu > 4$, 
		\begin{align*}
			[\Ib_{\text{obs}}(\thetab)]_{\mu \mu} &= \EE \left(\frac{\di \ell(\thetab)}{\di \mub^\top} \frac{\di \ell(\thetab)}{\di \mub}\right) = \Sigmab^{-1}(\Vb^*_{c_{\mu} y} + 2 \cb^*_{c_{\mu}y} \gammab^\top + \gammab \gammab^\top) \Sigmab^{-1},\\
			[\Ib_{\text{obs}}(\thetab)]_{\gamma \gamma} &= \EE \left(\frac{\di \ell(\thetab)}{\di \gammab^\top} \frac{\di \ell(\thetab)}{\di \gammab}\right) =\Sigmab^{-1} (\Vb^*_{ y} + 2 \cb^*_{c_{\gamma}y} \gammab^\top + v^*_{c_{\gamma} }  \gammab \gammab^\top) \Sigmab^{-1},\\
			[\Ib_{\text{obs}}(\thetab)]_{\vech \Sigma \vech \Sigma} &= \EE \left(\frac{\di \ell(\thetab)}{\di \vech(\Sigmab)^\top} \frac{\di \ell(\thetab)}{\di \vech(\Sigmab)}\right)  = \Db_d^\top  (\Sigmab^{-1} \otimes \Sigmab^{-1})  \Vb^*_{c_{\Sigma} }  (\Sigmab^{-1} \otimes \Sigmab^{-1}) \Db_d,\\
			[\Ib_{\text{obs}}(\thetab)]_{\nu \nu} &= \EE \left(\frac{\di \ell(\thetab)}{\di \nu} \frac{\di \ell(\thetab)}{\di \nu}\right) = \EE(c_{1 \nu}^2) + 
			\EE(c_{2 \nu}^2) + 2 \EE(c_{1 \nu} c_{2 \nu}).
		\end{align*}

		The off-diagonal blocks, $[\Ib_{\text{obs}}]_{\mub \gammab}$, $[\Ib_{\text{obs}}]_{\mub \vech \Sigmab}$, $[\Ib_{\text{obs}}]_{\mub \nu}$, 
		$[\Ib_{\text{obs}}]_{\gammab \vech \Sigmab}$, $[\Ib_{\text{obs}}]_{\gammab \nu}$,  $[\Ib_{\text{obs}}]_{\vech \Sigmab \nu}$, are found similarly using the following expectations:  
		\begin{align*} 
			[\Ib_{\text{obs}}]_{\mub \gammab} &= \EE \left\{\frac{\di \ell(\theta)}{\di \mub^\top} \frac{\di \ell(\theta)}{\di \gammab} \right\}, \\
			[\Ib_{\text{obs}}]_{\mub \vech \Sigmab} &= \EE \left\{\frac{\di \ell(\theta)}{\di \mub^\top} \frac{\di \ell(\theta)}{\di \vech \Sigmab} \right\},\\
			[\Ib_{\text{obs}}]_{\mub \nu} &= \EE \left\{\frac{\di \ell(\theta)}{\di \mub^\top} \frac{\di \ell(\theta)}{\di \nu} \right\}, \\
			[\Ib_{\text{obs}}]_{\gammab \vech \Sigmab} &= \EE \left\{\frac{\di \ell(\theta)}{\di \gammab^\top} \frac{\di \ell(\theta)}{\di \vech \Sigmab} \right\}, \\
			[\Ib_{\text{obs}}]_{\gammab \nu} &= \EE \left\{\frac{\di \ell(\theta)}{\di \gammab^\top} \frac{\di \ell(\theta)}{\di \nu} \right\},\\
			[\Ib_{\text{obs}}]_{\vech \Sigmab \nu} &= \EE \left\{\frac{\di \ell(\theta)}{\di \vech \Sigmab^\top} \frac{\di \ell(\theta)}{\di \nu} \right\}.
		\end{align*}
		
		The proof is complete.
	\end{proof}
	\citet{Are10} derives the score function (i.e., $\di \ell(\thetab) / \di \thetab$) and the information matrix using a different approach. Their motivation is to study the skew $t$ score function and its relation with skew normal and $t$ distributions. Our motivation is to use it for deriving the rate of convergence of an EM-type algorithm for estimating $\thetab$.

	Using the multivariate normal mean-variance mixture model, our second result obtains an analytic form for the ``complete data'' information matrix of $\yb$. Specifically, if $\yb$ follows a multivariate Skew $t$($\mub, \gammab, \Sigmab, \nu$) distribution, then we obtain this distribution as the marginal of $\yb$ in the following hierarchical model for ``complete data'' $(\yb, w)$:
	\begin{align}\label{eq:skew-mix}
		\yb \mid w \sim \text{Normal}_d (\mub + w \gammab, w \Sigmab), \quad 
		w \sim \text{Inverse Gamma}(\nu/2, \nu/2),
	\end{align}
	where the scale and shape parameters of the Inverse Gamma distribution equal $\nu/2$, $w$ is the ``missing'' data,  and marginalizing over $w$ yields the Skew $t$($\mub, \gammab, \Sigmab, \nu$) distribution of $\yb$. The following proposition uses the complete data model in \eqref{eq:skew-mix} to obtain the analytic form of the complete data information matrix.
	
	\begin{proposition}\label{prop2}
		Let  $g(\yb, w)$ be the joint density of the complete data $(\yb, w)$ defined by the hierarchical model in \eqref{eq:skew-mix},
		$\thetab = (\mub, \gammab, \Sigmab, \nu) \in \RR^{\frac{d^2 + 5d + 2}{2}}$ and $\yb$ follows a multivariate Skew $t$($\mub, \gammab, \Sigmab, \nu$) distribution. Then, the complete data information matrix and its blocks are 
		\begin{align}%\label{eq:comp-inf}
			[\Ib_{\text{com}}(\thetab)]_{\nu  \nu} &= \frac{1}{4} \psi'(\nu/2) -  \frac{1}{2 \nu}, \nonumber\\
			[\Ib_{\text{com}}(\thetab)]_{\mub, \mub} &= \Sigmab^{-1}, \nonumber\\
			[\Ib_{\text{com}}(\thetab)]_{\gammab, \gammab} &= \frac{\nu}{\nu - 2}\Sigmab^{-1}, 
			\nonumber \\
			[\Ib_{\text{com}}(\thetab)]_{\mub, \gammab} &= \Sigmab^{-1}, \nonumber\\
			[\Ib_{\text{com}}(\thetab)]_{\vech \Sigmab, \vech \Sigmab} &=\frac{1}{2} \Db_d^\top (\Sigmab^{-1} \otimes \Sigmab^{-1} ) \Db_d,
		\end{align}
		where  $\Db_d$ is the duplication matrix. The remaining blocks of the completed data information matrix are zero matrices.
		
	\end{proposition}

	\begin{proof}
		The hierarchical model in \eqref{eq:skew-mix} implies that the complete data log likelihood is
		\begin{align*}
			\log g(\yb, w) &= - \frac{1}{2} \log |2 \pi w \Sigmab| - \frac{1}{2} (\yb - \mub - w \gammab)^\top (w\Sigmab)^{-1} (\yb - \mub - w \gammab)\\
			&\quad  + \frac{\nu}{2} \log \frac{\nu}{2} - \left( \frac{\nu}{2} + 1\right) \log w - \frac{\nu}{2 w} - \log \Gamma (\nu  / 2)\\
			&= - \frac{d}{2} \log (2 \pi w) - \frac{1}{2} \log |\Sigmab| - \frac{1}{2w} (\yb - \mub)^\top \Sigmab^{-1} (\yb - \mub) \\
			&\quad - \frac{w}{2} \gammab^\top \Sigmab^{-1} \gammab + (\yb - \mub)^\top \Sigmab^{-1} \gammab\\
			&\quad  + \frac{\nu}{2} \log \frac{\nu}{2} - \left( \frac{\nu}{2} + 1\right) \log w - \frac{\nu}{2 w} - \log \Gamma (\nu  / 2).
		\end{align*}
		The second derivative with respect to $\nu$ follows from standard results:
		\begin{align*}
			\frac{\di^2 \log g(\yb, w)}{\di \nu^2} &= \frac{1}{2 \nu} - \frac{1}{4} \psi'(\nu/2).
		\end{align*}
		Noting that $\frac{\di \log g(\yb, w)}{\di \nu}$ does not depend on $\mub, \gammab, \Sigmab$, we get that 
		\begin{align*}
			\frac{\di^2 \log g(\yb, w) }{\di \nu \di \mub^\top} = \frac{\di^2 \log g(\yb, w) }{\di \nu \di \gammab^\top} = \frac{\di^2 \log g(\yb, w) }{\di \nu \di \vech(\Sigmab)^\top} = \zero, 
		\end{align*}
		where $\zero$ is a row vector of the appropriate dimension.
		
		As a function of the non-scalar parameters $\mub, \gammab, \Sigmab$, 
		\begin{align*} 
			\log g(\yb, w) &\propto - \frac{1}{2} \log |\Sigmab| - \frac{1}{2} (\yb - \mub - w \gammab)^\top (w\Sigmab)^{-1} (\yb - \mub - w \gammab).
		\end{align*}
		The quadratic form of the $\log g(\yb, w)$ in $\mub$ and $ \gammab$ implies that
		\begin{align*}
			\frac{\di^2 \log g(\yb, w) }{\di \mub \di \mub^\top} = -\frac{1}{w} \Sigmab^{-1}, \quad
			\frac{\di^2 \log g(\yb, w) }{\di \gammab \di \gammab^\top} = -w \Sigmab^{-1}, 
			\quad
			\frac{\di^2 \log g(\yb, w) }{\di \mub \di \gammab^\top} = -\Sigmab^{-1}.
		\end{align*}
		Taking expectations of all the three terms gives
		\begin{align*}
			[\Ib_{\text{com}}(\thetab)]_{\mub, \mub} &= -\EE \left( \frac{\di^2 \log g(\yb, w) }{\di \mub \di \mub^\top} \right) = \EE(w^{-1}) \Sigmab^{-1} = \Sigmab^{-1}, \\
			[\Ib_{\text{com}}(\thetab)]_{\gammab, \gammab} &= -\EE \left( \frac{\di^2 \log g(\yb, w) }{\di \gammab \di \gammab^\top} \right) = \EE(w) \Sigmab^{-1} = \frac{\nu}{\nu - 2}\Sigmab^{-1}, 
			\\
			[\Ib_{\text{com}}(\thetab)]_{\mub, \gammab} &= -\EE \left(\frac{\di^2 \log g(\yb, w) }{\di \mub \di \gammab^\top} \right) = \Sigmab^{-1},
		\end{align*}
		where we have used that $W$ follows the Inverse-Gamma($\nu/2$, $\nu/2$) distribution and assumed that $\nu > 2$ for the existence of $\EE(w)$.
		Similarly, the cross terms,
		\begin{align*}
			\frac{\di^2 \log g(\yb, w) }{\di \vech (\Sigmab) \di \mub^\top} = \frac{1}{w} \frac{\di \Sigmab^{-1}}{\di \vech (\Sigmab)} (\yb - \mub - w\gammab), \quad
			\frac{\di^2 \log g(\yb, w) }{\di \vech(\Sigmab) \di \gammab^\top} =  \frac{\di \Sigmab^{-1}}{\di \vech (\Sigmab)} (\yb - \mub - w\gammab).
		\end{align*}
		Because $\EE( \yb - \mub - w\gammab \mid W = w) = 0$,  
		\begin{align*}
			[\Ib_{\text{com}}(\thetab)]_{\vech \Sigmab, \mub} &= -\EE \left( \frac{\di^2 \log g(\yb, w) }{\di \vech(\Sigmab) \di \mub^\top} \right) = \zero, \\
			[\Ib_{\text{com}}(\thetab)]_{\vech \Sigmab, \gammab} &= -\EE \left( \frac{\di^2 \log g(\yb, w) }{\di \vech(\Sigmab) \di \gammab^\top} \right) =  \zero.
		\end{align*}
		
		Finally, we drive the derivative with respect to $\vv(\Sigmab)$ and $\vech(\Sigmab)$. If we retain the terms dependent on $d \Sigmab$ only, then 
		\begin{align*} 
			\di^2 \log g(\yb, w) &=  \frac{1}{2} \tr(\Sigmab^{-1} \di \Sigmab \Sigmab^{-1} d \Sigmab) - \frac{1}{w} \tr \left\{ (\yb - \mub - w \gammab)^\top \Sigmab^{-1} \di \Sigmab  \Sigmab^{-1} \di \Sigmab \Sigmab^{-1} (\yb - \mub - w \gammab) \right\}  \\
			&=\vv(\di \Sigmab)^\top  \frac{1}{2} (\Sigmab^{-1} \otimes \Sigmab^{-1})\vv(\di \Sigmab) - \\
			&\quad \, \vv(\di \Sigmab)^\top \frac{1}{w} \left\{\Sigmab^{-1} \otimes \Sigmab^{-1 } (\yb - \mub - w \gammab) (\yb - \mub - w \gammab)^\top \Sigmab^{-1 }\right\} \vv(\di \Sigmab)  \\
			&= \vv(\di \Sigmab)^\top \Vb_{\mu, \gamma, \Sigma, w, y} \vv(\di \Sigmab),\\
			\Vb_{\mu, \gamma, \Sigma, w, y} &=
			\frac{1}{2} (\Sigmab^{-1} \otimes \Sigmab^{-1}) -  \frac{1}{w} \left\{\Sigmab^{-1} \otimes \Sigmab^{-1 } (\yb - \mub - w \gammab) (\yb - \mub - w \gammab)^\top \Sigmab^{-1 }\right\} \\
			&= \Sigmab^{-1} \otimes 
			\left\{\frac{1}{2} \Sigmab^{-1} -  \frac{1}{w}  \Sigmab^{-1 } (\yb - \mub - w \gammab) (\yb - \mub - w \gammab)^\top \Sigmab^{-1 }\right\}
		\end{align*}
		If $\Db_d$ is the duplication matrix such that $\vv(\di \Sigmab) = \Db_d \vech(\di \Sigmab)$, then the previous display implies that 
		\begin{align}\label{eq:comp-hess}
			\frac{\di^2 \log g(\yb, w) }{\di \vech(\Sigmab) \di \vech(\Sigmab)^\top} &=
			\Db_d^\top \Vb_{\mu, \gamma, \Sigma, w, y} \Db_d.   
		\end{align}
		Using \eqref{eq:skew-mix}, $\EE \left\{ (\yb - \mub - w \gammab) (\yb - \mub - w \gammab)^\top \right\} = w \Sigmab$ and 
		\begin{align*}
			[\Ib_{\text{com}}(\thetab)]_{\vech \Sigmab, \vech \Sigmab} &= 
			-\Db_d^\top \EE \{ \EE ( \Vb_{\mu, \gamma, \Sigma, w, y} \mid W = w) \} \Db_d = \frac{1}{2} \Db_d^\top (\Sigmab^{-1} \otimes \Sigmab^{-1} ) \Db_d.
		\end{align*}
		
		The proof is complete.
	\end{proof}

	\section{Analytic Forms of the Complete and Observed Data Information Matrices}
	
	The next theorem extends Propositions \ref{prop1} and \ref{prop2} to the simplified REGMVST model. To avoid extensive algebra, we assume that
	\begin{align}\label{eq:sim-regmvst}
		\Yb = \Xb \betab + \Eb, \quad \Eb \sim \text{MVST}(\zero, \Ab, \Sigmab, \Psib, \nu), \quad \Yb \in \RR^{n \times p}, \quad \Xb \in \RR^{n \times q}, \quad \betab \in \RR^{q \times p},
	\end{align}
	for the theoretical results, where $\Ab = \one_n \ab^\top$, $\ab$ is a $p \times 1$ vector of skewness, $\Psib$ and $\Sigmab$ are the $p \times p$ and $n \times n$  column and row covariance matrices of $\Eb$, and $\nu$ is the degrees of freedom.
	The vectorized form of \eqref{eq:sim-regmvst} is
	\begin{align}\label{eq:vec-regmvst}
		\yb = \Ib_p \otimes \Xb \bbb + \eb \equiv \tilde \Xb \bbb + \eb, \quad \eb \sim \text{MST}_{np}(\zero, \Ib_p \otimes \one_n \ab, \Psib \otimes \Sigmab, \nu),
	\end{align}
	where $\text{MST}_{np}$ is the $np$-dimensional multivariate skew $t$ distribution. This implies that $\yb$ follows $\text{MST}_{np}(\tilde \Xb \bbb, \Ib_p \otimes \one_n \ab, \Psib \otimes \Sigmab, \nu)$; see (9) in \citet{gallaugher2017matrix} for details. Using \eqref{eq:skew-mix}, the parameter expanded form of $\yb \sim \text{MST}_{np}(\tilde \Xb \bbb, \Ib_p \otimes \one_n \ab, \Psib \otimes \Sigmab, \nu)$ is
	\begin{align}\label{eq:px-sim-regmvst}
		\yb \mid w \sim \text{Normal}_{np} (\tilde \Xb \bbb + w (\Ib_p \otimes \one_n) \ab, w (\Psib \otimes \Sigmab)), \quad 
		w \sim \text{Inverse Gamma}(\nu/2, \nu/2).
	\end{align}
	The next theorem uses Proposition \ref{prop2} to define the complete data information matrix for the vectorized REGMSVT parameter-expanded model in \eqref{eq:px-sim-regmvst}.
	\begin{theorem} \label{thm1}
		Let $\Yb$ follow the REGMVST model in \eqref{eq:sim-regmvst}, $g(\yb, w)$ be the joint density of the complete data $(\yb, w)$ defined by the hierarchical model in \eqref{eq:px-sim-regmvst}, and
		$\thetab = (\bbb, \ab, \vech(\Sigmab), \vech(\Psib), \nu) \in \RR^{pq + p + n(n+1)/2 + p(p+1)/2 + 1}$. Then, the complete data information matrix and its blocks are 
		\begin{align}\label{eq:comp-inf}
			[\Ib_{\text{com}}(\thetab)]_{\nu  \nu} &= \frac{1}{4} \psi'(\nu/2) -  \frac{1}{2 \nu}, \nonumber\\
			[\Ib_{\text{com}}(\thetab)]_{\bbb, \bbb} &= \tilde \Xb^\top (\Psib^{-1} \otimes \Sigmab^{-1}) \tilde \Xb, \nonumber\\
			[\Ib_{\text{com}}(\thetab)]_{\ab, \ab} &=  \frac{\nu}{\nu - 2}   (\one_n^\top\Sigmab^{-1} \one_n) \Psib^{-1}, 
			\nonumber \\
			[\Ib_{\text{com}}(\thetab)]_{\bbb, \ab} &= \tilde \Xb^\top (\Psib^{-1} \otimes \Sigmab^{-1} \one_n) , \nonumber\\
			[\Ib_{\text{com}}(\thetab)]_{\vech \Psib, \vech \Psib} 
			&= \frac{n}{2}\, \Db_p^\top(\Psib^{-1}\otimes \Psib^{-1}) \Db_p
			,\nonumber\\
			[\Ib_{\text{com}}(\thetab)]_{\vech \Sigmab, \vech \Sigmab}
			&= \frac{p}{2}\, \Db_n^\top(\Sigmab^{-1}\otimes \Sigmab^{-1}) \Db_n
			, \nonumber \\
			[\Ib_{\text{com}}(\thetab)]_{\vech \Sigmab, \vech \Psib} &= 
			\Db_n^\top\big(\Psib^{-1} \otimes \Sigmab^{-1} \big) \Db_p,      
		\end{align}
		where $\Db_n$ and $\Db_p$ are the duplication matrices such that $\Db_p \vech(\Psib) = \vv(\Psib)$ and $\Db_n \vech(\Sigmab) = \vv(\Sigmab)$. The remaining blocks of the completed data information matrix are zero matrices.
		
	\end{theorem}
	
	\begin{proof}
		Following the proof of Proposition \ref{prop2}, as  a function of $\bbb, \ab, \Sigmab$, and $ \Psib$, the log-likelihood implied by \eqref{eq:px-sim-regmvst} satisfies
		\begin{align*}
			\log g(\yb, w) &\propto - \frac{1}{2} \log |\Psib  \otimes \Sigmab | - \frac{1}{2w} (\yb - \mub - w \gammab)^\top (\Psib^{-1} \otimes  \Sigmab^{-1} ) (\yb - \mub - w \gammab),
		\end{align*}
		where $\mub = \tilde \Xb \bbb$ and $\gammab = \Ib_p \otimes \one_n \ab$.
		Using the fact that  $|\Psib  \otimes \Sigmab |  =  | \Psib|^n |\Sigmab|^p $, the differential of the first term is
		\begin{align*}
			- \frac{1}{2} \log |\Psib  \otimes \Sigmab |  &=- \frac{p}{2}  \log |\Sigmab| - \frac{n}{2}  \log |\Psib|,\\
			- \frac{1}{2} \di  \log |\Psib  \otimes \Sigmab |  &=  - \frac{p}{2} \tr(\Sigmab^{-1} \di \Sigmab) - \frac{n}{2} \tr(\Psib^{-1} \di \Psib).
		\end{align*}
		For convenience, denote $\rb = \yb - \mub - w \gammab$, then the quadratic form in the second term
		\begin{align*}
			\rb^\top(\Psib^{-1} \otimes  \Sigmab^{-1} ) \rb &= \vv(\Rb)^\top \vv(\Sigmab^{-1} \Rb \Psib^{-1}) = \tr ( \Rb^\top \Sigmab^{-1} \Rb \Psib^{-1}),
		\end{align*}
		where $\vv(\Rb) = \rb$, and its differential as a function of $\Psib$ and $\Sigmab$ is
		\begin{align*}
			\di \rb^\top(\Psib^{-1} \otimes \Sigmab^{-1}) \rb &= \di \tr ( \Rb^\top \Sigmab^{-1} \Rb \Psib^{-1})\\
			&= - \tr ( \Rb^\top \Sigmab^{-1} \di \Sigmab \Sigmab^{-1} \Rb \Psib^{-1}) - \tr ( \Rb^\top \Sigmab^{-1} \Rb \Psib^{-1} \di \Psib  \Psib^{-1}).
		\end{align*}
		Define $\Rb = \Yb - \Xb \betab - w \one_n \ab^\top$ using \eqref{eq:sim-regmvst}, $\Sbb = \Rb^\top \Sigmab^{-1} \Rb$, and $\Tb = \Rb \Psib^{-1} \Rb^\top$, 
		\begin{align*}
			\di \tr ( \Rb^\top \Sigmab^{-1} \di \Sigmab \Sigmab^{-1} \Rb \Psib^{-1}) =
			-& \tr ( \Rb^\top \Sigmab^{-1} \di \Sigmab \Sigmab^{-1} \di \Sigmab \Sigmab^{-1} \Rb \Psib^{-1}) \\
			-& \tr ( \Rb^\top \Sigmab^{-1} \di \Sigmab \Sigmab^{-1} \di \Sigmab \Sigmab^{-1}  \Rb \Psib^{-1}) \\
			-& \tr ( \Rb^\top \Sigmab^{-1} \di \Sigmab \Sigmab^{-1} \Rb \Psib^{-1} \di \Psib \Psib^{-1}) \\
			\di \tr ( \Rb^\top \Sigmab^{-1} \Rb \Psib^{-1} \di \Psib  \Psib^{-1}) =
			-&\tr(\Rb^\top \Sigmab^{-1} \di \Sigmab \Sigmab^{-1} \Rb \Psib^{-1} \di \Psib  \Psib^{-1}) \\
			-& \tr ( \Rb^\top \Sigmab^{-1} \Rb \Psib^{-1} \di \Psib \Psib^{-1} \di \Psib  \Psib^{-1}) \\
			-& \tr ( \Rb^\top \Sigmab^{-1} \Rb \Psib^{-1} \di \Psib  \Psib^{-1} \di \Psib  \Psib^{-1}) \\
			\di^2 \rb^\top(\Psib^{-1} \otimes \Sigmab^{-1}) \rb &= - \di \tr ( \Rb^\top \Sigmab^{-1} \di \Sigmab \Sigmab^{-1} \Rb \Psib^{-1}) - \di   \tr ( \Rb^\top \Sigmab^{-1} \Rb \Psib^{-1} \di \Psib  \Psib^{-1}) \\
			&= 2 \tr ( \Rb^\top \Sigmab^{-1} \di \Sigmab \Sigmab^{-1} \di \Sigmab \Sigmab^{-1} \Rb \Psib^{-1}) +\\
			& \quad \; 2 \tr ( \Rb^\top \Sigmab^{-1} \Rb \Psib^{-1} \di \Psib  \Psib^{-1} \di \Psib  \Psib^{-1}) + \\
			& \quad \;  2 \tr ( \Rb^\top \Sigmab^{-1} \di \Sigmab \Sigmab^{-1} \Rb \Psib^{-1} \di \Psib \Psib^{-1}).
		\end{align*}
		These three expressions imply that 
		\begin{align}\label{eq:psi-sig-der}
			\frac{\di^2\log g(\yb, w) }{ \di \vech(\Psib) \di \vech(\Psib)^\top}
			&= \frac{n}{2}\, \Db_p^\top(\Psib^{-1}\otimes \Psib^{-1}) \Db_p
			-\frac{1}{w}\, \Db_p^\top\big(\Psib^{-1}\otimes \Psib^{-1} \Sbb \Psib^{-1}\big) \Db_p, \nonumber \\
			\frac{\di^2\log g(\yb, w) }{\di \vech(\Sigmab) \di \vech(\Sigmab)^\top }
			&= \frac{p}{2}\, \Db_n^\top(\Sigmab^{-1}\otimes \Sigmab^{-1}) \Db_n
			-\frac{1}{w}\, \Db_n^\top\big(\Sigmab^{-1}\otimes \Sigmab^{-1} \Tb \Sigmab^{-1}\big) \Db_n, \nonumber\\
			\frac{\di^2\log g(\yb, w) }{ \di \vech(\Psib) \di \vech(\Sigmab)^\top } &= 
			-\frac{1}{w}\, \Db_n^\top\big(\Sigmab^{-1} \Rb \Psib^{-1} \otimes \Sigmab^{-1} \Rb \Psib^{-1}\big) \Db_p.
		\end{align}
		Finally, noting that $\EE(\Sbb \mid w) = w n \Psib$, $\EE(\Tb \mid w) =  w p \Sigmab$, and 
		\begin{align*}
			\EE \big(\Sigmab^{-1} \Rb \Psib^{-1} \otimes \Sigmab^{-1} \Rb \Psib^{-1} \mid w \big) &= \EE \left( \vv(\Sigmab^{-1} \Rb \Psib^{-1} ) \vv(\Sigmab^{-1} \Rb \Psib^{-1} )^\top  \mid w \right) \\
			&= (\Psib^{-1} \otimes \Sigmab^{-1} ) \EE (\rb \rb^\top \mid w) (\Psib^{-1} \otimes \Sigmab^{-1})\\
			&= w (\Psib^{-1} \otimes \Sigmab^{-1} ),
		\end{align*}
		the second derivatives in \eqref{eq:psi-sig-der} imply that 
		the complete data information matrix for $\vech(\Psib)$ and $\vech(\Sigmab)$ are 
		\begin{align*}
			- \EE \left( \frac{\di^2\log g(\yb, w) }{ \di \vech(\Psib) \di \vech(\Psib)^\top} \right) 
			&= \frac{n}{2}\, \Db_p^\top(\Psib^{-1}\otimes \Psib^{-1}) \Db_p,  \\
			-\EE \big(\frac{\di^2\log g(\yb, w) }{\di \vech(\Sigmab) \di \vech(\Sigmab)^\top } \big)
			&= \frac{p}{2}\, \Db_n^\top(\Sigmab^{-1}\otimes \Sigmab^{-1}) \Db_n, \\
			-\EE \left(\frac{\di ^2\log g(\yb, w) }{ \di \vech(\Psib) \di \vech(\Sigmab)^\top } \right) &= 
			\Db_n^\top \big( \Psib^{-1} \otimes \Sigmab^{-1} \big) \Db_p.
		\end{align*}

		The blocks for $\bbb$ and $\ab$ are obtained using Proposition \ref{prop2} and the chain rule. Specifically,  $\di \mub = \tilde \Xb \di \bbb$ and $\di \gammab = \Ib_p \otimes \one_n \di \ab$, and the blocks for $\mub$ and $\gammab$ in the complete data information matrices are modified as 
		\begin{align*}
			\frac{\di^2 \log g(\yb, w) }{\di \bbb \di \bbb^\top} &= -\frac{1}{w} \tilde \Xb^\top \left( \Psib^{-1} \otimes \Sigmab^{-1} \right) \tilde \Xb, \\
			\frac{\di^2 \log g(\yb, w) }{\di \ab \di \ab^\top} &= -w \left( \Ib_p \otimes \one_n^\top  \right)  \left( \Psib^{-1} \otimes \Sigmab^{-1} \right) \left(\Ib_p \otimes \one_n \right) = 
			-w  \left( \Psib^{-1} \otimes \one_n^\top \Sigmab^{-1} \one_n \right) , 
			\\
			\frac{\di^2 \log g(\yb, w) }{\di \bbb \di \ab^\top} &= -\tilde \Xb^\top \left( \Psib^{-1} \otimes \Sigmab^{-1} \right) \left(\Ib_p \otimes \one_n \right) = -\tilde \Xb^\top \left( \Psib^{-1} \otimes \Sigmab^{-1} \one_n \right).
		\end{align*}
		Using these three equations,
		\begin{align*}
			-  \EE \left(   \frac{\di^2 \log g(\yb, w) }{\di \bbb \di \bbb^\top}  \right)&= \EE(1/w) \tilde \Xb^\top \left( \Psib^{-1} \otimes \Sigmab^{-1} \right) \tilde \Xb = \tilde \Xb^\top \left( \Psib^{-1} \otimes \Sigmab^{-1} \right) \tilde \Xb, \\
			-  \EE \left(      \frac{\di^2 \log g(\yb, w) }{\di \ab \di \ab^\top} \right) &= 
			\frac{\nu}{\nu - 2}  \left( \one_n^\top \Sigmab^{-1} \one_n \right) \Psib^{-1}, 
			\\
			-  \EE \left(       \frac{\di^2 \log g(\yb, w) }{\di \bbb \di \ab^\top} \right) &= \tilde \Xb^\top \left( \Psib^{-1} \otimes \Sigmab^{-1} \one_n \right).
		\end{align*}
		Finally, the information block for $\nu$ remains unchanged from Proposition \ref{prop2}.  The theorem is proved.
		
	\end{proof}

	The next theorem uses Proposition \ref{prop1} and chain rule to define the observed data information matrix for the vectorized REGMVST model in \eqref{eq:vec-regmvst}. 
	\begin{theorem}\label{thm2}
		Let $ f(\yb)$ be the density of $\yb$ defined by the vectorized REGMVST model in \eqref{eq:vec-regmvst} with parameters $\thetab = (\bbb, \ab, \Sigmab, \Psib, \nu) \in \RR^{pq + p + n(n+1)/2 + p(p+1)/2 + 1}$. Define
		\begin{align*}
			s (\yb) &= \left[ \{ \nu + \rho(\yb) \} \ab^\top (\Ib_p \otimes \one_n^\top) (\Psib^{-1} \otimes \Sigmab^{-1}) (\Ib_p \otimes \one_n) \ab \right]^{\frac{1}{2}}, \\
			\rho(\yb)&= (\yb - \tilde \Xb \bbb)^\top \Sigmab^{-1} (\yb - \tilde \Xb \bbb),\\
			c_{\bbb}(\yb) &= \left\{ \frac{K_{\frac{\nu + np}{2}}' \left( s (\yb) \right)}{K_{\frac{\nu + np}{2}} \left( s (\yb) \right)} + \frac{\nu + np}{2s(\yb)}\right\} \frac{\one_n^\top \Sigmab^{-1} \one_n  \ab^\top \Psib^{-1}  \ab } {s(\yb)} - \frac{\nu + np}{ \nu + \rho(\yb) },\\ 
			c_{\ab}(\yb) &= \left\{ \frac{K_{\frac{\nu + np}{2}}' \left( s (\yb) \right)}{K_{\frac{\nu + np}{2}} \left( s (\yb) \right)} + \frac{\nu + np}{2 s(\yb)}\right\} \frac{\nu + \rho(\yb)} {s(\yb)},\\
			\Cb_{\Omegab}(\yb) &=   c_{\mu \mu}(\yb) (\mub - \yb)(\mub - \yb)^\top   + c_{\gamma \gamma}(\yb) \gammab \gammab^\top + \frac{1}{2}\left\{\gammab (\mub - \yb)^\top +  (\mub - \yb) \gammab^\top \right\} - \frac{1}{2}\Sigmab,  \\
			c_{\bbb \bbb}(\yb) &= \frac{\nu + np}{ 2(\nu + \rho(\yb)) } - \left[ \frac{K_{\frac{\nu + np}{2}}' \left( s (\yb) \right)}{K_{\frac{\nu + np}{2}} \left( s (\yb) \right)} + \frac{\nu + np}{2s(\yb)}\right] \frac{\one_n^\top \Sigmab^{-1} \one_n  \ab^\top \Psib^{-1}  \ab } {2s(\yb)} ,\\
			c_{\ab \ab}(\yb) &=  - \left[ \frac{K_{\frac{\nu + np}{2}}' \left( s (\yb) \right)}{K_{\frac{\nu + np}{2}} \left( s (\yb) \right)} + \frac{\nu + np}{2s(\yb)}\right] \frac{\nu + \rho(\yb)} {2s(\yb)},\\
			c_{1 \nu}(\yb) &= -\frac{1}{2} \left\{\nu \log 2 + \psi\left(\frac{\nu}{2}\right) + \frac{np}{\nu} - \frac{(\nu + np) \rho(\yb)}{\nu(\nu + \rho(\yb))} + \log \left( 1 + \frac{\rho(\yb)}{\nu} \right) - \log s(\yb)\right\},\\
			c_{2 \nu}(\yb) &= \left\{ \frac{\partial K_{\frac{\nu + np}{2}} \left( s (\yb) \right)}{K_{\frac{\nu + np}{2}} \left( s (\yb) \right)} + \frac{\nu + np}{2 s(\yb)}\right\} \frac{\one_n^\top \Sigmab^{-1} \one_n  \ab^\top \Psib^{-1}  \ab}{2 s(\yb)},
		\end{align*}
		where $K'_{\lambda} (x)= \frac{\di K_{\lambda}(x)}{\di x}$, $\psi(\cdot)$ is the digamma function, and $\partial K_{\lambda} (x)= \frac{\di K_{\lambda}(x)}{\di \lambda}$. For $\nu > 4$ , 
		\begin{align*}
			\Vb^*_{c_{\bbb} y} &= \EE 
			\left[  \{c_{\bbb}(\yb)\}^2 (\yb - \tilde \Xb \bbb)  (\yb - \tilde \Xb \bbb)^\top \right],\quad
			\cb^*_{c_{\bbb}y} = \EE 
			\left\{  c_{\bbb}(\yb) (\yb - \tilde \Xb \bbb)   \right\}, \nonumber\\
			v^*_{c_{\ab} } &= \EE 
			\left[  \{c_{\ab}(\yb)\}^2 \right],\quad
			\cb^*_{c_{\ab}y} = \EE 
			\left\{  c_{\ab}(\yb) (\yb - \tilde \Xb \bbb)   \right\}, \quad
			\Vb^*_{y} = 
			\EE\{(\yb - \tilde \Xb \bbb)  (\yb - \tilde \Xb \bbb)^\top \} 
		\end{align*}
		exist. If $\ell(\thetab) = \log f(\yb)$, then the observed data information matrix of $\yb$ is
		\begin{align}\label{eq:obs-inf}
			\Ib_{\text{obs}}(\thetab) &= \EE \left(\frac{\di \ell(\thetab)}{\di \thetab^\top} \frac{\di \ell(\thetab)}{\di \thetab}\right), 
		\end{align}
		where the expectation is with respect to the distribution of $\yb$ and $\Ib_{\text{obs}}(\thetab)$ exists if $\nu > 4 $. The analytic forms of the blocks in $\frac{\di \ell(\thetab)}{\di \thetab}$ are as follows:
		\begin{align*}
			\frac{\di \ell(\theta)}{\di \bbb} &=  \left\{ c_{\bbb}(\yb) (\tilde \Xb \bbb - \yb)^\top - \ab^\top (\Ib_p \otimes \one^\top_n) \right\} (\Psib^{-1} \otimes \Sigmab^{-1}) \tilde \Xb, \\ 
			\frac{\di \ell(\theta)}{\di \ab} &=  \left\{ c_{\ab}(\yb) \ab^\top (\Ib_p \otimes \one^\top_n) - (\tilde \Xb \bbb - \yb)^\top \right\} (\Psib^{-1} \otimes \Sigmab^{-1} \one_n),\\
			\frac{\di \ell(\theta)}{\di \vech(\Psib)} &= \db_{\Psib}^\top \Db_p ,\\ 
			\frac{\di \ell(\theta)}{\di \vech (\Sigmab)}&= \db_{\Sigmab}^\top \Db_n,\\ 
			\frac{\di \ell(\theta)}{\di \nu} &= c_{1 \nu}(\yb) + c_{2 \nu}(\yb),
		\end{align*}
		where $\db_{\Psib} = \vv(\Db_{\Psib})$, $\db_{\Sigmab} = \vv(\Db_{\Sigmab})$, $\Omegab = \Psib \otimes \Sigmab$, $(i,j)$th entry of $p \times p$ matrix $\Db_{\Psib}$ is $\tr \left\{ (\Omegab^{-1} \Cb_{\Omegab} \Omegab^{-1})_{ij}  \Sigmab \right\}$ for $i, j=1, \ldots, p$, $(i,j)$th entry of $n \times n$ matrix $\Db_{\Sigmab}$ is $\tr \left\{ (\Omegab^{-1} \Cb_{\Omegab} \Omegab^{-1})_{ij}  \Psib \right\}$ for $i, j=1, \ldots, n$. Furthermore, \eqref{eq:obs-inf} implies that the five diagonal blocks in $\Ib_{\text{obs}}({\thetab})$ for the five parameter blocks are 
		\begin{align*}
			[\Ib_{\text{obs}}(\thetab)]_{\bbb \bbb} 
			&= \tilde \Xb^\top (\Psib^{-1} \otimes \Sigmab^{-1}) (\Vb^*_{c_{\mu} y} + 2 \cb^*_{c_{\mu}y} \ab^\top (\Ib_p \otimes \one^\top_n) + (\Ib_p \otimes \one_n) \ab \ab^\top (\Ib_p \otimes \one^\top_n)) (\Psib^{-1} \otimes \Sigmab^{-1}) \tilde \Xb,\\
			[\Ib_{\text{obs}}(\thetab)]_{\ab \ab} 
			&=(\Psib^{-1} \otimes \one_n^\top \Sigmab^{-1} ) (\Vb^*_{ y} + 2 \cb^*_{c_{\gamma}y} \ab^\top (\Ib_p \otimes \one^\top_n) + v^*_{c_{\gamma} }  (\Ib_p \otimes \one_n) \ab \ab^\top (\Ib_p \otimes \one^\top_n)) (\Psib^{-1} \otimes \Sigmab^{-1} \one_n),\\
			[\Ib_{\text{obs}}(\thetab)]_{\vech \Psib \vech \Psib} &=  \Db_p^\top \EE(\db_{\Psib} \db_{\Psib}^\top) \Db_p,\\
			[\Ib_{\text{obs}}(\thetab)]_{\vech \Sigmab \vech \Sigmab} &=  \Db_n^\top \EE(\db_{\Sigmab} \db_{\Sigmab}^\top) \Db_n,\\
			[\Ib_{\text{obs}}(\thetab)]_{\nu \nu} &= \EE(c_{1 \nu}^2) + 
			\EE(c_{2 \nu}^2) + 2 \EE(c_{1 \nu} c_{2 \nu}).
		\end{align*}
		
	\end{theorem}
	
	\begin{proof}
		Using the proof of Proposition \ref{prop1}, 
		\begin{align*}
			\frac{\di \ell(\thetab)} {\di \bbb} &=  \frac{\di \ell(\thetab)} {\di \mub} \frac{\di \mub} {\di \bbb} \\
			&= \left[\left\{ \frac{K_{\frac{\nu + np}{2}}' \left( s (\yb) \right)}{K_{\frac{\nu + np}{2}} \left( s (\yb) \right)} + \frac{\nu + np}{2 s(\yb)} \right\}  \frac{\gammab^\top \Omegab^{-1} \gammab } {s(\yb)} - \frac{\nu + np}{\{\nu + \rho(\yb) \}}  \right](\mub - \yb)^\top \Omegab^{-1} \frac{\di \mub} {\di \bbb} - \gammab^\top \Omegab^{-1}\frac{\di \mub} {\di \bbb}\\
			&\equiv \left\{ c_{\bbb}(\yb) (\tilde \Xb \bbb - \yb)^\top - \ab^\top (\Ib_p \otimes \one^\top_n) \right\} (\Psib^{-1} \otimes \Sigmab^{-1}) \tilde \Xb,
		\end{align*}
		where $d = np$, $\gammab = (\Ib_p \otimes \one_n) \ab$, and $\Omegab = \Psib \otimes \Sigmab$. Similarly, 
		\begin{align*}
			\frac{\di \ell(\thetab)} {\di \ab} &=   \frac{\di \ell(\thetab)} {\di \gammab} \frac{\di \gammab} {\di \ab} =
			\left\{ \frac{K_{\frac{\nu + np}{2}}' \left( s (\yb) \right)}{K_{\frac{\nu + np}{2}} \left( s (\yb) \right)} + \frac{\nu + np}{2 s(\yb)} \right\}  \frac{\di s(\yb)}{\di \gammab} \frac{\di \gammab} {\di \ab} - (\mub - \yb)^\top \Omegab^{-1} \frac{\di \gammab} {\di \ab}\\  
			&=   
			\left\{ \frac{K_{\frac{\nu + np}{2}}' \left( s (\yb) \right)}{K_{\frac{\nu + np}{2}} \left( s (\yb) \right)} + \frac{\nu + np}{2 s(\yb)} \right\}  \frac{\nu + \rho(\yb)} {s(\yb)} \gammab^\top \Omegab^{-1} \frac{\di \gammab} {\di \ab} - (\mub - \yb)^\top \Omegab^{-1} \frac{\di \gammab} {\di \ab}\\  
			&\equiv \left\{ c_{\ab}(\yb) \ab^\top (\Ib_p \otimes \one^\top_n) - (\tilde \Xb \bbb - \yb)^\top \right\} (\Psib^{-1} \otimes \Sigmab^{-1})(\Ib_p \otimes \one_n) \\
			&= \left\{ c_{\ab}(\yb) \ab^\top (\Ib_p \otimes \one^\top_n) - (\tilde \Xb \bbb - \yb)^\top \right\} (\Psib^{-1} \otimes \Sigmab^{-1} \one_n).
		\end{align*}
		
		Finally, the derivative with respect to $\nu$ remains unchanged from Proposition \ref{prop1} and the derivatives with respect to 
		$\vech(\Sigmab)$ and $\vech(\Psib)$ follows by noting that
		\begin{align*}
			\di \ell(\thetab) &= \tr \left( \Omegab^{-1} \Cb_{\Omegab} \Omegab^{-1} \di \Omegab   \right) = \tr \left( \Omegab^{-1} \Cb_{\Omegab} \Omegab^{-1} \di \Psib \otimes \Sigmab   \right) + \tr \left( \Omegab^{-1} \Cb_{\Omegab} \Omegab^{-1} \Psib \otimes \di \Sigmab  \right),\\
			\Cb_{\Omegab} &=\left[\frac{\nu + np}{2 \{\nu + \rho(\yb) \}} - \left\{ \frac{K_{\frac{\nu + np}{2}}' \left( s (\yb) \right)}{K_{\frac{\nu + np}{2}} \left( s (\yb) \right)} + \frac{\nu + np}{2 s(\yb)} \right\} \frac{\gammab^\top \Omegab^{-1} \gammab } {2s(\yb)}  \right] (\mub - \yb) (\mub - \yb)^\top \\
			&\quad - 
			\left\{ \frac{K_{\frac{\nu + np}{2}}' \left( s (\yb) \right)}{K_{\frac{\nu + np}{2}} \left( s (\yb) \right)} + \frac{\nu + np}{2 s(\yb)} \right\} \frac{\nu + \rho(\yb)} {2s(\yb)} \gammab \gammab^\top 
			+ \frac{1}{2}\{(\mub - \yb) \gammab  ^\top  +  \gammab  (\mub - \yb)^\top \} - \frac{1}{2} \Omegab \\
			&\equiv c_{\bbb \bbb}(\yb) (\tilde \Xb \bbb - \yb) (\tilde \Xb \bbb - \yb)^\top + c_{\ab \ab}(\yb) (\Ib_p \otimes \one_n ) \ab \ab^\top (\Ib_p \otimes \one^\top_n ) + \\
			&\quad \; \frac{1}{2}\{(\tilde \Xb \bbb - \yb) \ab^\top (\Ib_p \otimes \one^\top_n )   +  (\Ib_p \otimes \one_n ) \ab  (\tilde \Xb \bbb - \yb)^\top \}- \frac{1}{2} (\Psib \otimes \Sigmab).
		\end{align*}
		If $\di \Psib_{ij} \Sigmab$ is the $(i,j)$th $n \times n$ block of $\di \Psib \otimes  \Sigmab$ and $(\Omegab^{-1} \Cb_{\Omegab} \Omegab^{-1})_{ij}$ is the corresponds $n \times n$ block of $\Omegab^{-1} \Cb_{\Omegab} \Omegab^{-1}$, then 
		\begin{align*}
			\tr \left( \Omegab^{-1} \Cb_{\Omegab} \Omegab^{-1} \di \Psib \otimes  \Sigmab   \right) &= \sum_{ij} \di \Psib_{ij}  \tr \left\{ (\Omegab^{-1} \Cb_{\Omegab} \Omegab^{-1})_{ij}  \Sigmab \right\} = \sum_{ij} \tr \left\{ (\Omegab^{-1} \Cb_{\Omegab} \Omegab^{-1})_{ij}  \Sigmab \right\} \di \Psib_{ij}  \\
			&= \sum_{ij} \tr \left\{ (\Omegab^{-1} \Cb_{\Omegab} \Omegab^{-1})_{ji}  \Sigmab \right\} \di \Psib_{ij}  = \tr (\Db_{\Psib} \di \Psib), 
		\end{align*}
		where the $(i,j)$ entry of $p \times p$ matrix $\Db_{\Psib}$ is $\tr \left\{ (\Omegab^{-1} \Cb_{\Omegab} \Omegab^{-1})_{ij}  \Sigmab \right\}$ for $i, j=1, \ldots, p$. Similarly, if $\Psib \di \Sigmab_{ij}$ is the $(i,j)$th block of $\Psib \otimes \di \Sigmab$ and $(\Omegab^{-1} \Cb_{\Omegab} \Omegab^{-1})_{ij}$ is the corresponds $p \times p$ block of $\Omegab^{-1} \Cb_{\Omegab} \Omegab^{-1}$, then 
		\begin{align*}
			\tr \left( \Omegab^{-1} \Cb_{\Omegab} \Omegab^{-1} \Psib \otimes \di \Sigmab  \right) &= \sum_{ij} d \Sigmab_{ij}  \tr \left\{ (\Omegab^{-1} \Cb_{\Omegab} \Omegab^{-1})_{ij}  \Psib \right\} = \sum_{ij} \tr \left\{ (\Omegab^{-1} \Cb_{\Omegab} \Omegab^{-1})_{ij}  \Psib \right\} \di \Sigmab_{ij}  \\
			&= \sum_{ij} \tr \left\{ (\Omegab^{-1} \Cb_{\Omegab} \Omegab^{-1})_{ji}  \Psib \right\} \di \Sigmab_{ij}  = \tr (\Db_{\Sigmab} \di \Sigmab), 
		\end{align*}
		where the $(i,j)$ entry of $n \times n$ matrix $\Db_{\Sigmab}$ is $\tr \left\{ (\Omegab^{-1} \Cb_{\Omegab} \Omegab^{-1})_{ij}  \Psib \right\}$ for $i, j=1, \ldots, n$. The previous two displays imply that 
		\begin{align*}
			\frac{\di \ell(\thetab)}{\di \vech(\Psib)} &= \vv(\Db_{\Psib})^\top \Db_p \equiv \db_{\Psib}^\top \Db_p, \\
			\frac{\di \ell(\thetab)}{\di \vech(\Sigmab)} &= \vv(\Db_{\Sigmab})^\top \Db_n \equiv \db_{\Sigmab}^\top \Db_n.
		\end{align*}

		The form of the information matrix implies the forms of the diagonal blocks for $\mub$, $\gammab$, $\vech(\Omegab)$, and $\nu$. Define
		\begin{align}%\label{eq:vars}
			\Vb^*_{c_{\bbb} y} &= \EE 
			\left[  \{c_{\bbb}(\yb)\}^2 (\yb - \tilde \Xb \bbb)  (\yb - \tilde \Xb \bbb)^\top \right],\quad
			\cb^*_{c_{\bbb}y} = \EE 
			\left\{  c_{\bbb}(\yb) (\yb - \tilde \Xb \bbb)   \right\}, \nonumber\\
			v^*_{c_{\ab} } &= \EE 
			\left[  \{c_{\ab}(\yb)\}^2 \right],\quad
			\cb^*_{c_{\ab}y} = \EE 
			\left\{  c_{\ab}(\yb) (\yb - \tilde \Xb \bbb)   \right\}, \quad
			\Vb^*_{y} = 
			\EE\{(\yb - \tilde \Xb \bbb)  (\yb - \tilde \Xb \bbb)^\top \}, 
		\end{align}
		where all the expectations are with respect to the $\text{MST}$($\tilde \Xb \bbb, (\Ib_p \otimes \one_n) \ab, \Psib \otimes \Sigmab, \nu$) distribution. When $\nu > 4$, 
		\begin{align*}
			[\Ib_{\text{obs}}(\thetab)]_{\bbb \bbb} &= \EE \left(\frac{\di \ell(\thetab)}{\di \bbb^\top} \frac{\di \ell(\thetab)}{\di \bbb}\right) \\
			&= \tilde \Xb^\top (\Psib^{-1} \otimes \Sigmab^{-1}) (\Vb^*_{c_{\mu} y} + 2 \cb^*_{c_{\mu}y} \ab^\top (\Ib_p \otimes \one^\top_n) + (\Ib_p \otimes \one_n) \ab \ab^\top (\Ib_p \otimes \one^\top_n)) (\Psib^{-1} \otimes \Sigmab^{-1}) \tilde \Xb,\\
			[\Ib_{\text{obs}}(\thetab)]_{\ab \ab} &= \EE \left(\frac{\di \ell(\thetab)}{\di \ab^\top} \frac{\di \ell(\thetab)}{\di \ab}\right) \\
			&=(\Psib^{-1} \otimes \one_n^\top \Sigmab^{-1} ) (\Vb^*_{ y} + 2 \cb^*_{c_{\gamma}y} \ab^\top (\Ib_p \otimes \one^\top_n) + v^*_{c_{\gamma} }  (\Ib_p \otimes \one_n) \ab \ab^\top (\Ib_p \otimes \one^\top_n)) (\Psib^{-1} \otimes \Sigmab^{-1} \one_n),\\
			[\Ib_{\text{obs}}(\thetab)]_{\vech \Psib \vech \Psib} &= \EE \left(\frac{\di \ell(\thetab)}{\di \vech(\Psib)^\top} \frac{\di \ell(\thetab)}{\di \vech(\Psib)}\right)  = \Db_p^\top \EE(\db_{\Psib} \db_{\Psib}^\top) \Db_p,\\
			[\Ib_{\text{obs}}(\thetab)]_{\vech \Sigmab \vech \Sigmab} &= \EE \left(\frac{\di \ell(\thetab)}{\di \vech(\Sigmab)^\top} \frac{\di \ell(\thetab)}{\di \vech(\Sigmab)}\right)  = \Db_n^\top \EE(\db_{\Sigmab} \db_{\Sigmab}^\top) \Db_n,\\
			[\Ib_{\text{obs}}(\thetab)]_{\nu \nu} &= \EE \left(\frac{\di \ell(\thetab)}{\di \nu} \frac{\di \ell(\thetab)}{\di \nu}\right) = \EE(c_{1 \nu}^2) + 
			\EE(c_{2 \nu}^2) + 2 \EE(c_{1 \nu} c_{2 \nu}).
		\end{align*}

		The off-diagonal blocks, $[\Ib_{\text{obs}}]_{\bbb \ab}$, $[\Ib_{\text{obs}}]_{\bbb  \vech \Omegab}$, $[\Ib_{\text{obs}}]_{\bbb  \nu}$, 
		$[\Ib_{\text{obs}}]_{\ab \vech \Omegab}$, $[\Ib_{\text{obs}}]_{\ab \nu}$,  $[\Ib_{\text{obs}}]_{\vech \Omegab \nu}$, are found similarly using the following expectations:  
		\begin{align*} 
			[\Ib_{\text{obs}}]_{\bbb  \ab} &= \EE \left\{\frac{\di \ell(\theta)}{\di \bbb ^\top} \frac{\di \ell(\theta)}{\di \ab} \right\}, \\
			[\Ib_{\text{obs}}]_{\bbb  \vech \Psib} &= \EE \left\{\frac{\di \ell(\theta)}{\di \bbb ^\top} \frac{\di \ell(\theta)}{\di \vech \Psib} \right\},\\
			[\Ib_{\text{obs}}]_{\bbb  \vech \Sigmab} &= \EE \left\{\frac{\di \ell(\theta)}{\di \bbb ^\top} \frac{\di \ell(\theta)}{\di \vech \Sigmab} \right\},\\
			[\Ib_{\text{obs}}]_{\bbb  \nu} &= \EE \left\{\frac{\di \ell(\theta)}{\di \bbb ^\top} \frac{\di \ell(\theta)}{\di \nu} \right\}, \\
			[\Ib_{\text{obs}}]_{\ab \vech \Psib} &= \EE \left\{\frac{\di \ell(\theta)}{\di \ab^\top} \frac{\di \ell(\theta)}{\di \vech \Psib} \right\}, \\
			[\Ib_{\text{obs}}]_{\ab \vech \Sigmab} &= \EE \left\{\frac{\di \ell(\theta)}{\di \ab^\top} \frac{\di \ell(\theta)}{\di \vech \Sigmab} \right\}, \\
			[\Ib_{\text{obs}}]_{\ab \nu} &= \EE \left\{\frac{\di \ell(\theta)}{\di \ab^\top} \frac{\di \ell(\theta)}{\di \nu} \right\},\\
			[\Ib_{\text{obs}}]_{\vech \Psib \nu} &= \EE \left\{\frac{\di \ell(\theta)}{\di \vech \Psib^\top} \frac{\di \ell(\theta)}{\di \nu} \right\},\\
			[\Ib_{\text{obs}}]_{\vech \Sigmab \nu} &= \EE \left\{\frac{\di \ell(\theta)}{\di \vech \Sigmab^\top} \frac{\di \ell(\theta)}{\di \nu} \right\}.
		\end{align*}
		
		The theorem is proved.
	\end{proof}

	\section{Proof of the Rate of Convergence}\label{rate-proof}
	
	Our next proposition uses Theorems \ref{thm1} and \ref{thm2}  to define the matrix rate of convergence of an ADECME algorithm for estimating $\varthetab$. Let $\hat \varthetab$ be the stationary point of the ADECME sequence $\{\varthetab^{(t)} \}$,  $N$ be the sample size, $\Rb$ be the matrix rate of convergence, $\Sbb$ be the matrix speed of convergence, $\Ib_{c,i}$ and $\Ib_{o,i}$ be the complete data and observed data information matrix for the $i$the sample ($i=1, \ldots, N$). Theorems \ref{thm1} and \ref{thm2} define the analytic forms of  $\Ib_{c,i}$ and $\Ib_{o,i}$ for every $i$. 
	Then, \citet{Men94} shows that $\Rb$ and $\Sbb$ are defined as follows:
	\begin{align}\label{eq:rate-em2}
		\Ib_{c_N}  = \sum_{i=1}^N \Ib_{c,i}, \quad \Ib_{o_N} = \sum_{i=1}^N \Ib_{o,i}, \quad \Sbb = \Ib_{c_N}^{-1} \Ib_{o_N}, \quad \Rb = \Ib -  \Ib_{c_N}^{-1} \Ib_{o_N},\quad \Rb  = \Ib - \Sbb ,
	\end{align}
	where $\Ib$ is a $d \times d$ identity matrix, $\Sbb$ and $\Rb$ are $d \times d$ positive definite matrices,  and $d = {pq + p + n(n+1)/2 + p(p+1)/2 + 1}$. The rate and speed of convergence equal $r_{\max} = \lambda_{\text{max}}(\Rb)$ and $s_{\min} = \lambda_{\text{min}}(\Sbb) = 1 - r_{\max}$. \citet{Men94} shows that $r_{\max}, s_{\min} \in (0, 1)$. We 
	estimate $\varthetab$ using the complete data model in \eqref{eq:complete_data_loglikelihood} with $\Sigmab_i = \Sigmab$ for every $i$; see the vectorized REGMVST model in \eqref{eq:vec-regmvst} and its complete data  model in \eqref{eq:px-sim-regmvst}.

	\begin{proof}
		The Taylor series expansion of the log likelihood gradient, $\ell'(\varthetab)$, at $\varthetab^{(t)}$ gives
		\begin{align}\label{eq:rpf1}
			\ell'(\varthetab) \approx \ell'(\varthetab^{(t)}) + \ell''(\varthetab^{(t)}) (\varthetab - \varthetab^{(t)}), \quad  
			0 = \ell'(\hat \varthetab) \approx \ell'(\varthetab^{(t)}) + \ell''(\varthetab^{(t)}) (\hat \varthetab - \varthetab^{(t)}),
		\end{align}
		where the last equation uses the fact that $\hat \varthetab$ is the stationary point of $\ell(\varthetab)$. Eq. \eqref{eq:rpf1} implies that $\hat \varthetab \approx \varthetab^{(t)} - \ell^{''}(\varthetab^{(t)})^{-1} \ell'(\varthetab^{(t)} ) = \varthetab^{(t)} +\Ib_{o_N}^{-1} \ell'(\varthetab^{(t)}) $.
		
		We use Taylor expansion again to relate, $\ell'(\varthetab^{(t)})$, with the gradient of ADECME's $ Q(\varthetab \mid \varthetab^{(t)})$ function. At the end of the $t$th ADECME iteration, let 
		$\Qcal(\cdot \mid \varthetab^{(t')})$ be the $\Qcal(\cdot \mid \cdot)$ function for the $(1-\gamma)$-fraction of samples that are on the worker machines that did not return their results to the manager, where $t' < t$.  For the remaining $\gamma$-fraction of samples, the $\Qcal(\cdot \mid \cdot)$ function used in the distributed CM step is  $\Qcal(\cdot \mid \varthetab^{(t)})$. Expanding the gradient of the ADECME's $\Qcal(\cdot \mid \cdot)$ function at $\varthetab^{(t)}$ gives
		\begin{align*}
			0 = \Qcal'(\varthetab^{(t+1)} \mid \varthetab^{(t)}) \approx &\gamma 
			\Qcal'(\varthetab^{(t)} \mid \varthetab^{(t)}) + (1 - \gamma) \Qcal'(\varthetab^{(t)} \mid \varthetab^{(t')}) + \nonumber\\
			&\gamma 
			\Qcal''(\varthetab^{(t)} \mid \varthetab^{(t)}) (\varthetab^{(t+1)} - \varthetab^{(t)} ) + (1 - \gamma) \Qcal''(\varthetab^{(t)} \mid \varthetab^{(t')}) (\varthetab^{(t+1)} - \varthetab^{(t)} ),       
		\end{align*}
		where all gradients are $\mathrm{D}^{10}$ and Hessians are $\mathrm{D^{20}}$. 
		Noting that $\Qcal'(\varthetab^{(t)} \mid \varthetab^{(t)}) = \ell'( \varthetab)^{(t)}$, and \eqref{eq:q-regular} implies that $\Qcal'(\varthetab^{(t)} \mid \varthetab^{(t')}) \approx \ell'(\varthetab^{(t)})$.   Substituting these identities in the previous display gives
		\begin{align}\label{eq:rpf2}
			\ell'(\varthetab^{(t)}) &\approx   
			- \{\gamma \Qcal''(\varthetab^{(t)} \mid \varthetab^{(t)})  + (1 - \gamma) \Qcal''(\varthetab^{(t)} \mid \varthetab^{(t')}) \} (\varthetab^{(t+1)} - \varthetab^{(t)} ) \nonumber\\
			&=   - \{-\gamma \Ib_{c_N}  + (1 - \gamma) \Qcal''(\varthetab^{(t)} \mid \varthetab^{(t')}) \} (\varthetab^{(t+1)} - \varthetab^{(t)} ) \nonumber\\
			&\approx - \{-\gamma \Ib_{c_N}  + (1 - \gamma) [\Qcal''(\varthetab^{(t)} \mid \varthetab^{(t)}) - \Deltab] \} (\varthetab^{(t+1)} - \varthetab^{(t)} ) \nonumber\\
			&= - \{-\gamma \Ib_{c_N}  + (1 - \gamma) [ -\Ib_{c_N} - \Deltab] \} (\varthetab^{(t+1)} - \varthetab^{(t)} )
			\nonumber\\
			&= \{\Ib_{c_N}  + (1 - \gamma) \Deltab \} (\varthetab^{(t+1)} - \varthetab^{(t)} ),
		\end{align}
		where we used $- \Qcal''(\varthetab^{(t)} \mid \varthetab^{(t)}) = \Ib_{c_N} $ in the second line and \eqref{eq:q-regular}  in the third.
		
		Finally, substituting \eqref{eq:rpf2} in \eqref{eq:rpf1} gives
		\begin{align*}
			\hat \varthetab - \varthetab^{(t)} \approx \Ib_{o_N}^{-1} \{\Ib_{c_N}  + (1 - \gamma) \Deltab \} (\varthetab^{(t+1)} - \hat \varthetab + \hat \varthetab - \varthetab^{(t)} ).
		\end{align*}
		If we collect terms involving $(\varthetab^{(t)} - \hat \varthetab)$  on the right hand side, then
		\begin{align}
			(\varthetab^{(t+1)} - \hat \varthetab ) &\approx \left[\Ib -   \{ \Ib_{o_N}^{-1} \Ib_{c_N}  + (1 - \gamma) \Ib_{o_N}^{-1}\Deltab \}^{-1} \right] (\varthetab^{(t)} - \hat \varthetab) \nonumber \\
			&=\left[\Ib -  \Ib_{c_N}^{-1} \Ib_{o_N}\{  \Ib  + (1 - \gamma) \Ib_{c_N}^{-1}\Deltab \}^{-1} \right] (\varthetab^{(t)} - \hat \varthetab) \nonumber \\
			&= \left[ \Ib - \Sbb \{  \Ib  + (1 - \gamma) \Ib_{c_N}^{-1}\Deltab \}^{-1} \right] (\varthetab^{(t)} - \hat \varthetab) \nonumber\\
			&= \left[ \Ib - \Sbb  + (1 - \gamma) \Sbb \{  \Ib  + (1 - \gamma) \Ib_{c_N}^{-1}\Deltab \}^{-1}  \Ib_{c_N}^{-1}\Deltab\right] (\varthetab^{(t)} - \hat \varthetab) \nonumber\\
			&= \left[ \Rb  + \tilde \Deltab_{\gamma}\right] (\varthetab^{(t)} - \hat \varthetab) \equiv \Rb_{\text{ADEM}} (\varthetab^{(t)} - \hat \varthetab)
		\end{align}
		where $\tilde \Deltab_{\gamma}$ is a positive definite matrix depending on $(\gamma, \Deltab, \Sbb, \Ib_{c_N})$ and the second last equality uses the identity $\{\Ib + (1 - \gamma) \Ib_{c_N}^{-1}\Deltab\}^{-1} = \Ib - \{\Ib + (1 - \gamma) \Ib_{c_N}^{-1}\Deltab\}^{-1} (1 - \gamma) \Ib_{c_N}^{-1}\Deltab$. The last equality implies that the rate of convergence matrix is $\Rb$ plus a positive definite matrix depending on $(1-\gamma)$, which is the fraction of samples ignored in every iteration of the ADECME algorithm. The proof is complete.
		
	\end{proof}

	\section{Architectural Overview} \label{sec:Architectural Overview}
	
	To further compare the differences between the PECME and ADECME algorithms, we present architectural overviews in Figures~\ref{fig:architectural overview of the PECME algorithm} and \ref{fig:architectural overview of the ADECME algorithm}, respectively. In Figure~\ref{fig:architectural overview of the PECME algorithm}, the distributed E step updates all sufficient statistics based on the subsets assigned to each worker. Communication between the manager and workers occurs five times per iteration for the distributed E step, updating $\nu$, $\mathcal{A}$, $\boldsymbol{\Psi}$, and the DEC parameters ($\rho_1$ and $\rho_2$). In contrast, Figure~\ref{fig:architectural overview of the ADECME algorithm} shows that only the first iteration updates all sufficient statistics. In subsequent iterations, if we wait for $k-1$ workers to complete their computations (for example, if worker 2 is the slowest in a particular iteration), the sufficient statistics from worker 2 are not updated. Instead, the most recent values of sufficient statistics 2 are used in the subsequent CM steps. Additionally, after the asynchronous distributed E step, no further communication occurs between the manager and workers.
	
	\begin{figure}[ht]
		\centering
		\includegraphics[width=0.9\linewidth]{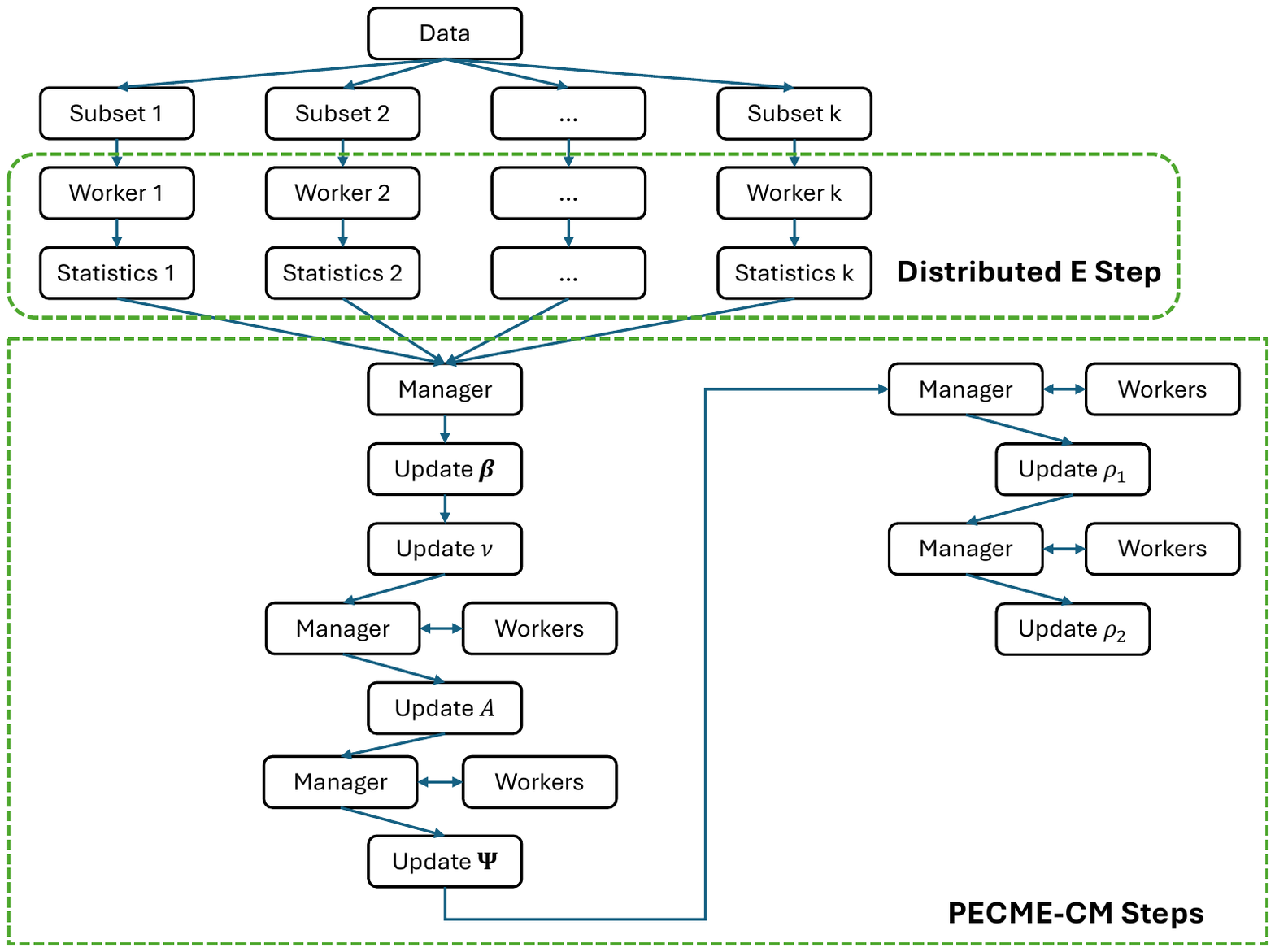}
		\caption{\label{fig:architectural overview of the PECME algorithm}The architectural overview of the PECME algorithm.}
	\end{figure}
	
	\begin{figure}[ht]
		\centering
		\includegraphics[width=0.9\linewidth]{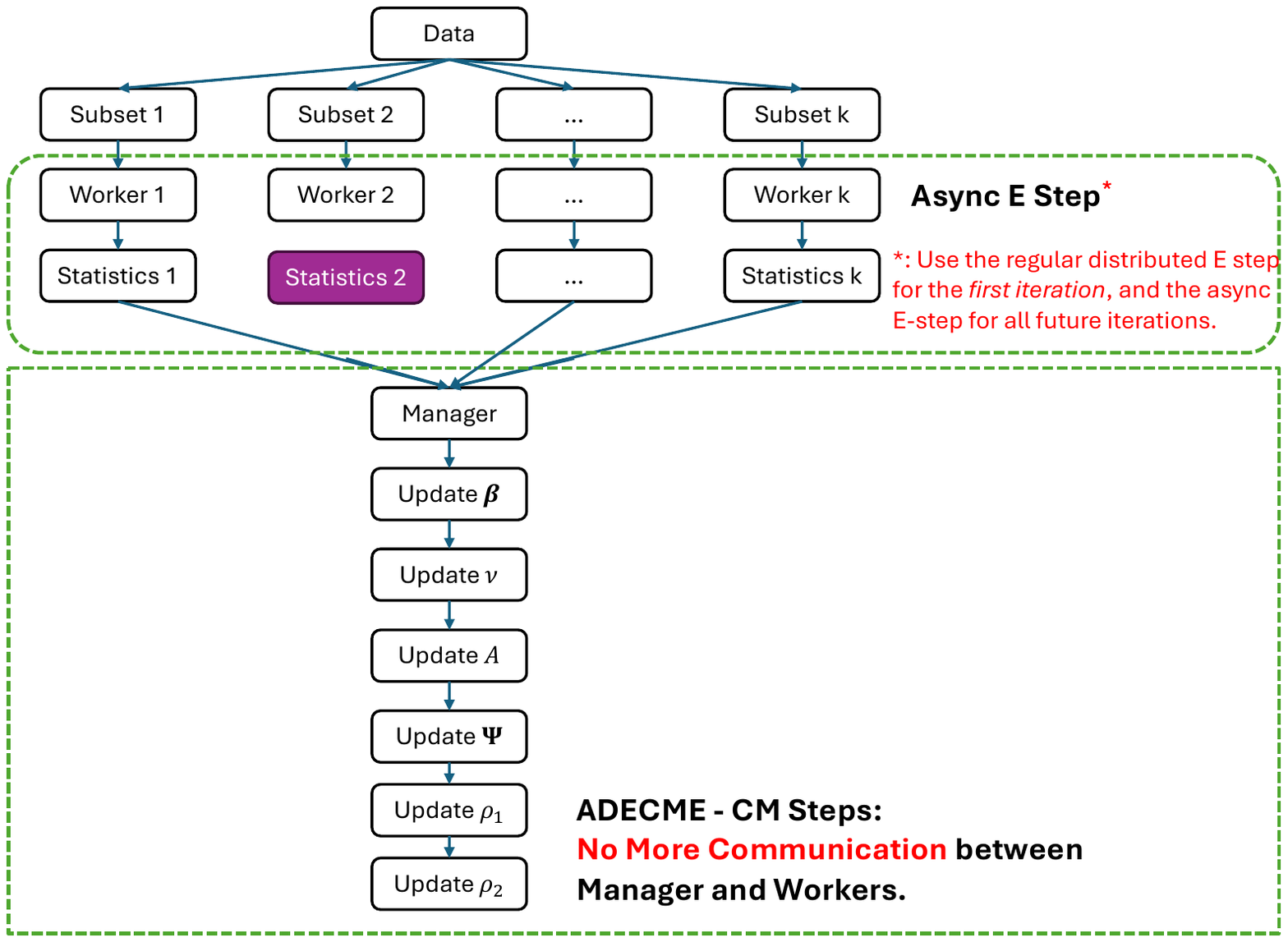}
		\caption{\label{fig:architectural overview of the ADECME algorithm}The architectural overview of the ADECME algorithm.}
	\end{figure}
	
	
	
	
\end{appendices}

\clearpage

\bibliographystyle{apalike}
\bibliography{ref}% common bib file

\end{document}